%% file: TIST_extension/main-TIST.tex
\documentclass[acmsmall,submission]{acmart}

\AtBeginDocument{%
  }

\setcopyright{acmlicensed}
\copyrightyear{2018}
\acmYear{2018}
\acmDOI{XXXXXXX.XXXXXXX}

\acmJournal{JACM}
\acmVolume{37}
\acmNumber{4}
\acmArticle{111}
\acmMonth{8}

\input{header.tex}

\begin{document}

\title{Sequential Diversification with Provable Guarantees}


\author{Honglian Wang}
\affiliation{%
  \institution{KTH Royal Institute of Technology}
  \city{Stockholm}
  \country{Sweden}
  }
\email{honglian@kth.se}

\author{Sijing Tu}
\affiliation{%
  \institution{KTH Royal Institute of Technology}
  \city{Stockholm}
  \country{Sweden}
  }
\email{sijing@kth.se}

\author{Aristides Gionis}
\affiliation{%
  \institution{KTH Royal Institute of Technology}
  \city{Stockholm}
  \country{Sweden}
  }
\email{argioni@kth.se}


\renewcommand{\shortauthors}{Honglian Wang, Sijing Tu, and Aristides Gionis}

\begin{abstract}

Diversification is a useful tool for exploring large collections of information items. It has been used to reduce redundancy and cover multiple perspectives in information-search settings. Diversification finds applications in many different domains, including presenting search results of information-retrieval systems and selecting suggestions for recommender systems.

Interestingly, existing measures of diversity are defined over \emph{sets} of items, 
rather than evaluating \emph{sequences} of items.
This design choice comes in contrast with commonly-used relevance measures, 
which are distinctly defined over sequences of items, 
taking into account the ranking of items. 
The importance of employing sequential measures is that
information items are almost always presented in a sequential manner,
and during their information-exploration activity users tend to prioritize items with higher~ranking.

In this paper, we study the problem of \emph{maximizing sequential diversity}. 
This is a new measure of \emph{diversity}, 
which accounts for the \emph{ranking} of the items, and incorporates \emph{item relevance} and \emph{user behavior}.
The overarching framework can be instantiated with different diversity measures, 
and here we consider the measures of \emph{sum~diversity} and \emph{coverage~diversity}. 
The problem was recently proposed by Coppolillo et al.~\citep{coppolillo2024relevance}, 
where they introduce empirical methods that work well in practice.
Our paper is a theoretical treatment of the problem:
we establish the problem hardness and
present algorithms with constant approximation guarantees for both diversity measures we consider.
Experimentally, we demonstrate that our methods are competitive against strong baselines. 
\end{abstract}

\ccsdesc[500]{Information systems~Information retrieval}
\ccsdesc[500]{Theory of computation~Approximation algorithms analysis}

\keywords{Diversification, Ranking algorithms, Approximation algorithms}

\received{20 February 2007}
\received[revised]{12 March 2009}
\received[accepted]{5 June 2009}

\maketitle

\section{Introduction}

Diversification has been widely adopted in information retrieval and recommendation systems to facilitate information exploration, that is, to present users with content that is not only relevant, but also novel and diverse. 
Diversification benefits both users and service providers~\cite{wu2024result}: it enhances user experience by mitigating the filter-bubble effect and promoting a fair representation of perspectives, while helping service providers avoid over-exposure of popular items and cater to broader information needs.

%
%
Numerous diversification approaches have been proposed~\citep{wu2024result}, 
including methods that aim to balance the trade\-off 
between \emph{diversity} and \emph{relevance}~\cite{borodin2012max,gollapudi2009axiomatic}.
However, most existing approaches have two shortcomings. 
First, they assume a fixed number of items to be selected, overlooking the fact that users may lose interest and terminate their information-seeking task early. Second, many approaches formulate the diversification task as a \emph{set-selection problem}, rather than a \emph{ranking problem}, thereby ignoring the importance of the ordering of selected items.

These shortcomings arise from failing to incorporate the users' engagement behavior.
Users may exit due to tiredness or disinterest, making it challenging for the system to determine beforehand how many content items the users will examine. 
Additionally, the ranking of selected items is crucial, as users are more likely to examine items that are higher in the ranking. 
Note that many diversification methods employ greedy strategies \citep{carbonell1998use,chen2018fast,agrawal2009diversifying,puthiya2016coverage,borodin2012max}, which implicitly provide a ranking of the items. However, such a
ranking is not an integral part of the problem definition. It is not designed to meet any desirable properties and is merely a by-product
of the algorithmic~solution.


\vspace{1mm}
\para{Our contributions.}
In this paper, we build on the setting of \emph{sequential diversity}, 
introduced recently by Coppolillo et al.~\citep{coppolillo2024relevance}, 
aiming to address the above-mentioned shortcomings. 
The sequential-diversity setting defines a new diversity measure, 
which accounts for the \emph{ranking} of items, 
and it incorporates \emph{item relevance} together with \emph{user behavior}.
While the work of Coppolillo et al.~\citep{coppolillo2024relevance}
presents practical methods for this problem, 
our paper offers a rigorous analysis and algorithms with provable guarantees. 

The sequential-diversity setting is formalized as follows.
We~consider that items are presented to users in a specific order and 
users examine them sequentially. 
We define the \emph{continuation probability} of an item for a user as the probability 
that the user will \emph{accept} the item and continue \emph{examining} the next item in the sequence.
We assume that continuation probabilities are related to the \emph{item relevance}
for a given user,
and thus can be learned from observed~data.

It follows that the model is \emph{stochastic}, with the number of items examined being a random variable depending on the ranking of items and their continuation probabilities. 
The \emph{sequential diversity} of an ordered sequence is defined as the \emph{expected diversity} 
of the items accepted by the user, with the expectation taken over the continuation probabilities.
Our objective is to maximize sequential diversity by computing an optimal ranking of items, which we refer to as \emph{sequential-diversity maximization}. 
This model effectively captures the interplay between relevance, diversity, and user engagement, as achieving high sequential diversity requires rankings that incorporate \emph{both diverse and relevant items} to maintain user engagement with the platform.

The model needs to be instantiated with a diversity measure for item sets, and any standard measure can be used.
In this paper, we consider two commonly-used diversity measures: \emph{pair-wise sum diversity}~\cite{borodin2012max} and 
\emph{coverage diversity}~\cite{ashkan2015optimal, puthiya2016coverage,zhai2015beyond}.
For coverage diversity, 
we prove the problem is \emph{ordered-submodular} \cite{kleinberg2022ordered}, and thus,
the standard greedy algorithm achieves a 1/2-approximation.

For pair-wise sum diversity measure and its corresponding maximization problem (\omsd, \Cref{problem:omsd}),
we need to develop novel methods.
We establish the computational hardness of the problem and 
present algorithms with constant-factor approximation guarantees.
Our techniques leverage connections to an \emph{ordered} variant of the \emph{Hamiltonian path} 
problem (\Cref{section:reduction}), which seeks to maximize the sum of edge weights in the Hamiltonian path, 
where the edge weights depend on their order in the path.
Our algorithms are grounded on the observation that the sequential diversity score is largely determined by the top items in the ranking; thus, we prioritize optimizing the selection and positioning of these top~items.

Depending on the values of the continuation probabilities, we introduce two distinct algorithms. 
When continuation probabilities are relatively small, e.g., constant values that are smaller than $1$, 
we design an algorithm, \bke, that focuses on selecting the top-$\kpar$ items, 
where the value of $\kpar$ controls the trade-off between runtime and approximation quality.
Notably, when $\kpar = 2$, our approach aligns with a
simple greedy approach, and hence, we provide the approximation
ratio for this greedy approach.
In contrast, when continuation probabilities are arbitrarily close to $1$, 
it's not just the first few items that affect the sequential diversity score; rather, more items ranked toward the front impact the score.
To address this case, we propose a second algorithm using a greedy matching approach on the set of all items, and offers a constant approximation guarantee. 

\edit{In this journal extension, we propose two new greedy algorithms that directly tackle the \omsd problem. We show that when the continuation probabilities are uniform, i.e., all items share the same probability, ignoring the probabilities and greedily constructing rankings that maximize pairwise distance sums yields a $\frac{1}{2}$-approximation. For non-uniform probabilities,
we propose an algorithm, \Sbke, that provides a constant-factor approximation guarantee by selecting the best $\tau$ items to form a length-$\tau$ ordering that maximizes the sequential sum over all size-$\tau$ rankings.

Our experiments demonstrate the superiority of our proposed algorithms over the baselines. In particular, the \Sbke algorithms generally achieve better results than the \bke algorithms, consistent with our theoretical analysis. In addition to the \osdo values, we evaluate user satisfaction and engagement metrics for the rankings generated by each approach. Due to space limitations, complete proofs and extensive experimental results are provided in the appendices.
}

\section{Related Work}

Content diversification is commonly formulated as an optimization problem
\cite{agrawal2009diversifying,puthiya2016coverage,carbonell1998use,borodin2012max,gollapudi2009axiomatic,chen2018fast,ashkan2015optimal,kleinberg2022ordered,cevallos2015max,cevallos2017local,cevallos2019improved}.
Two common diversity functions are coverage diversity and pairwise-sum diversity.

\citet{agrawal2009diversifying} introduce the concept of user coverage, 
aiming to recommend a relevant and diverse set while maximizing the number of users who encounter at least one relevant document. 
\citet{ashkan2015optimal} and \citet{puthiya2016coverage} explore taxonomy coverage, defining diversity as the number of topics or users covered by the selected set of items. 
These studies formulate their objective as monotone submodular functions and 
apply greedy strategies with provable guarantees.

For pairwise sum diversity, 
\citet{carbonell1998use} introduce the \emph{marginal maximal relevance} (MMR), 
which is a function that combines the pairwise distance and the relevance of the selected items to generate a ranking via greedily selection.
\citet{borodin2012max} and \citet{gollapudi2009axiomatic} investigate the notion of 
\emph{max-sum diversity} (MSD) by combining sum of pairwise distance with a submodular relevance function. Subsequent works on the MSD problem apply convex-programming \cite{cevallos2015max} and 
local-search \cite{cevallos2017local,cevallos2019improved} techniques to achieve better approximation ratios.
\citet{chen2018fast} and \citet{gan2020enhancing} use a \emph{determinantal point process} (DPP)  
model by maintaining a kernel matrix whose off-diagonal entries measure the similarity between items. 
\citet{chen2018fast} use the maximum a posterior (MAP) inference for DPP and propose a greedy MAP inference algorithm for efficient computation. 

The main limitation of the aforementioned studies is the assumption that users consider all results equally, without accounting for the order of items
Technically, although a greedy heuristic can produce an ordered sequence
\cite{carbonell1998use,chen2018fast,agrawal2009diversifying,puthiya2016coverage,borodin2012max}, 
the objective value does not depend on item ranking.

Some studies incorporate the ranking of items into their models.
\citet{ashkan2014diversified,ashkan2015optimal} propose the 
\emph{diversity-weighted utility maximization} (DUM) approach which ranks items solely based on their utilities (relevance scores, in our setting).
Unlike DUM, our method integrates relevance and diversity directly into the ranking.
\citet{ieong2014advertising} and \citet{tang2020optimizing} address advertisement allocation in a sequential setting similar to ours, where users sequentially view content and exit with fixed probabilities. However, they solve a distinct set of problems. 

The most closely related works to our problem are by \citet{coppolillo2024relevance} and \citet{kleinberg2022ordered}. \citet{coppolillo2024relevance} introduce a user-behavior model that simulates user interactions with recommendation systems and propose the \explore\ algorithm to maximize their proposed diversity measures. 
Their framework allows users to interact with multiple items at each time step and accept one of them.
Moreover, they do not provide any provable guarantees of their approaches. 
\citet{kleinberg2022ordered} assumes that users’ patience decays as they progress through a result list. 
They propose an \emph{ordered-submodular} coverage function and prove that a greedy algorithm achieves a $\frac{1}{2}$-approximation;
this result extends to our problem when the coverage function is selected as the diversity measure.
When pairwise distance is used as the diversity measure, we develop novel methods to approach our problem.

Recent works have applied deep-learning techniques to content diversification \cite{abdool2020managing, xu2023multi, huang2021sliding, zhou2018deep}, 
typically incorporating diversification as part of a broader optimization task.
In contrast, our focus is on defining sequential diversity as a standalone optimization problem and developing methods with provable guarantees. Notably, our framework can leverage deep-learning approaches to learn continuation probabilities.
Consequently, we view deep-learning approaches as complementary to our framework.

\section{Preliminaries}
\label{sec:setting}

We consider a system where users interact with distinct content items,
such as videos or articles, presented in an ordered sequence.
Let $\unordSet = \{1, \ldots, n\}$ represent the set of items.
For a given user \user, 
each item $i \in \unordSet$ has an associated \emph{continuation probability} $\itempr{i}$, representing the probability that,
upon examining item $i$,
the user \user \emph{accepts} this item and continues examining subsequent items.
Conversely, $1 - \itempr{i}$ is the probability of terminating the session
and quit the system after examining item $i$. 

We assume that $\itempr{i}$ 
is determined by the \emph{relevance} of item $i$ for user~\user and is provided as input to the problem.
While estimating $\itempr{i}$ is an interesting task, it lies outside the scope of this work. In our empirical evaluation, we explore simple methods that map relevance scores to continuation~probabilities.

The user examines items in \(\unordSet\) sequentially, following an order determined by the permutation \(\order: \unordSet \rightarrow \unordSet\), where \(\order{i}\) denotes the \(i\)-th item in the sequence. 
Let $\ordSet_{(\order)} = (\order{i})_{i=1}^n$ represent the ordered sequence defined by~$\order$;
for simplicity, we omit the subscript
$\order$ and write $\ordSet = \ordSet{(\order)}$.
We denote by \(\ordSet_{\kprefix} = (\order{i})_{i=1}^k\) the \(\kprefix\)-prefix of \ordSet, i.e., the subsequence of $\ordSet$ consisting of the first $k$ items. Equivalently, we write $\ordSet{\kprefix} \sqsubseteq \ordSet$.

After examining item 
\(\order{i}\), the user either \emph{accepts} it and proceeds to examine the next item with probability 
$\itempr{\order{i}}$, or rejects it and quits the system immediately.

Given an ordered sequence of items \ordSet, let us denote by  \(\randomordSet \sqsubseteq \ordSet\)
the ordered sequence of items the user examines and \emph{accepts} before
quitting. 
Notice that \(\randomordSet\) is a random variable. 
The probability of the user accepting exactly the items of \(\ordSet_\kprefix\) is \(\pr(\randomordSet = \ordSet_\kprefix)\). Let \(\divf(\cdot)\) denote a \emph{diversity function} measuring the diversity of the accepted items. As we will discuss shortly, our goal is to maximize the expected diversity score \(\divf(\randomordSet)\).

In our paper, an algorithm \(\alg\) achieves an \(\alpha\)-approximation (\(\alpha \leq 1\)) if, for any instance \(\instance\), it guarantees \(\alg(\instance) \geq \alpha \, \opt(\instance)\).

\subsection{Problem Definition}

We first introduce the \emph{sequential-diversity} objective
and then formalize the \emph{maximization} problem we consider. 

\begin{definition}[Sequential diversity (\odo)]
    \label{def:od}
    Let $\unordSet = \{1, \ldots, n\}$ be a finite set of $n$ distinct items, 
    and $\itempr{1}, \ldots, \itempr{n}$ be continuation probabilities assigned to each item.
    Let $\ordSet = (\order{i})_{i=1}^n$ be an ordered sequence of~\unordSet~according to a permutation $\order$. 
    Let $\divf(\randomordSet)$ be a diversity function of the items in $\randomordSet$. 
    The \emph{sequential diversity} of $\ordSet$, denoted by~$\odo(\ordSet)$,
    is defined as the expectation of the diversity of the items that the user accepts before quitting, 
    i.e., 
    $\odo(\ordSet) = \Exp_{\randomordSet \sqsubseteq \ordSet}[\divf(\randomordSet)]$.
\end{definition}

The expectation in Definition~\ref{def:od} is taken over the probability that the user \emph{accepts} exactly $\ordSet{\kprefix} \sqsubseteq \ordSet$, where $\kprefix = 0, 1, \ldots, n$. 

We will work with two commonly-used diversity functions. 
The first is the \emph{pair-wise sum diversity}:
for a distance function
$d \colon \unordSet \times \unordSet \rightarrow \mathbb{R}_{\geq 0}$, 
the diversity function 
$\divfsum(\randomordSet) = \sum_{i,j \in \randomordSet} \dist{i,j}$
sums all pair-wise distances of items in $\unordSet$.

Accordingly, we define the \emph{sequential sum diversity} objective.

\begin{definition}[Sequential sum diversity (\osdo)]
    \label{def:osd} 
    The \emph{sequential sum diversity} of a sequence $\ordSet$, denoted by~$\osdo(\ordSet)$,
    is defined as in Definition~\ref{def:od}, 
    with diversity function $\divf=\divfsum$, i.e., 
    \begin{equation}
    \label{eq:osdo}
    \osdo(\ordSet) = \Exp_{\randomordSet \sqsubseteq \ordSet}[\divfsum(\randomordSet)] = 
    \Exp_{\randomordSet \sqsubseteq \ordSet}\left[\sum\nolimits_{i,j \in \randomordSet} \dist{i,j}\right].
    \end{equation}
\end{definition}

Next, we define the problem \omsd of finding an ordering of items in~\unordSet 
that maximizes the sequential sum diversity objective. 

\begin{problem}[Maximizing sequential sum diversity (\omsd)]
    \label{problem:omsd}
    Given a finite set $\unordSet = \{1, \ldots, n\}$ of $n$ distinct items
    and associated continuation probabilities $\itempr{1}, \ldots, \itempr{n}$, 
    find an ordering $\ordSet^{*}$ of the items in $\unordSet$
    that maximizes the sequential sum diversity objective, i.e., 
    \begin{equation}
        \label{eq:omsd:form}
        \ordSet^{*} = \arg\max\nolimits_{\ordSet = \order{\unordSet}} \osdo(\ordSet).
    \end{equation}
\end{problem}

In Definition~\ref{def:osd}, the items in \unordSet can be used to define a \emph{complete weighted graph}, 
with edge weights being the distances between the items.
We make heavy use of this observation 
as we employ \emph{graph-theoretic ideas}, 
such as \emph{Hamiltonian paths}, \emph{graph matchings},~etc.

The second (alternative) notion of diversity is based on \emph{coverage}~\cite{ashkan2015optimal,puthiya2016coverage}.
Here we assume that each item $i \in \unordSet$ is associated with a set of attributes $\attr(i)\subseteq \genres$,
where \genres is the set of all attributes. 
For example, in a move dataset, 
\genres is the set of all movie genres, 
and~$\attr(i)$ the genres of movie $i$.
The \emph{coverage diversity} of a sequence~$\randomordSet$ 
is then defined as the set of attributes of all items in~\randomordSet, i.e., 
$\divfcov(\randomordSet) = \bigcup_{i \in \randomordSet} \attr(i)$. 

Analogously to Problem \ref{problem:omsd}, we define the problem 
of \emph{maximizing the sequential coverage diversity} (\omcd, Problem~\ref{problem:omcd}).

\begin{definition}[Sequential coverage diversity (\ocdo)]
    \label{def:ocd} 
    The \emph{sequential coverage diversity} of a sequence $\ordSet$, denoted by~$\ocdo(\ordSet)$,
    is defined as in Definition~\ref{def:od}, 
    with diversity function $\divf=\divfcov$, i.e., 
    \begin{equation}
    \label{eq:ocdo}
    \ocdo(\ordSet) = \Exp_{\randomordSet \sqsubseteq \ordSet}[\divfcov(\randomordSet)] = 
    \Exp_{\randomordSet \sqsubseteq \ordSet}\left[\bigcup\nolimits_{i \in \randomordSet} \attr(i)\right].
    \end{equation}
\end{definition}

\begin{problem}[Maximizing sequential coverage diversity (\omcd)]
    \label{problem:omcd}
    Given a finite set $\unordSet = \{1, \ldots, n\}$ of $n$ distinct items
    and associated probabilities $\itempr{1}, \ldots, \itempr{n}$, 
    find an ordering $\ordSet^{*}$ of the items in $\unordSet$
    that maximizes the sequential coverage diversity objective, i.e., 
    \begin{equation}
        \label{eq:omcd:form}
        \ordSet^{*} = \arg\max\nolimits_{\ordSet = \order{\unordSet}} \ocdo(\ordSet).
    \end{equation}
\end{problem}

\subsection{Reformulation and Complexity}

We start our analysis with an observation that simplifies our objective.
Observation~\ref{obs:prob-of-A} stems from the fact that 
the probabilities assigned to all items are independent.

To simplify our calculations, we use the notation 
$\itempr{\ordSet{i}} = \prod_{t=1}^{i}\itempr{\order{t}}$ and
$~~\dist{\order{i+1}, \ordSet{i}} = \sum\nolimits_{t =1}^{i} \dist{\order{i+1},\order{t}}$.

\begin{observation}
    \label{obs:prob-of-A}
    The probability that the user does not accept any item is 
    $\pr (\randomordSet = \ordSet{0}) = 1 - \itempr{\order{1}}$;
    the probability that the user quits after accepting the first $\kprefix$ items is
    $\pr (\randomordSet = \ordSet{\kprefix}) = \prod_{i=1}^\kprefix \itempr{\order{i}} (1 - \itempr{\order{\kprefix+1}})$, if $\kprefix \leq n-1$; 
    and the probability that the user accepts all items is 
    $\pr (\randomordSet = \ordSet{n}) = \prod_{i=1}^n \itempr{\order{i}}$.
\end{observation}
Based on Observation~\ref{obs:prob-of-A}, we reformulate Equation~(\ref{eq:osdo}).
%
%
\begin{restatable}{lemma}{eqform}
    \label{lem:eqform}
    The sequential sum diversity objective in Equation~(\ref{eq:osdo}) can be  reformulated as 
    \begin{align}
        \label{eq:eqform_sequential_diversity}
        \osdo(\ordSet) = \sum\nolimits_{i=1}^{n-1} \itempr{\ordSet{i+1}} \dist{\order{i+1}, \ordSet{i}}.
    \end{align}
\end{restatable}

\begin{proof}
    We start with demonstrating that the sequential sum-diversity objective can be precisely reformulated via the expression $2 \sum_{i=2}^{n} \sum_{j=1}^{i-1} \prod_{t=1}^{i} \pr_{\order{t}} \dist{\order{i}, \order{j}}$. 
    Next, the formulation can be simplified using the definition of $\itempr{\ordSet{i}}$ and 
    $\dist{\order{i}, \ordSet{i-1}}$. 
    
    We write the definition of sequential sum-diversity objective, 
    and proceed by merging the terms according to our observations. 
    \begin{equation}
\begin{aligned}
    &\sum_{\kprefix=1}^n \pr (\randomordSet = \ordSet{\kprefix}) \sum_{\order{i}, \order{j} \in \ordSet{\kprefix}} \dist{\order{i},\order{j}} 
    \stackrel{(a)}= \sum_{\kprefix=1}^n \pr(\randomordSet= \ordSet{\kprefix}) \sum_{i, j \leq \kprefix} \dist{\order{i},\order{j}} \\
    & \stackrel{(b)}= \sum_{\kprefix=1}^{n-1} (\prod_{t=1}^\kprefix \itempr{\order{t}}- \prod_{t=1}^{\kprefix+1} \itempr{\order{t}})  \sum_{i, j \leq \kprefix} \dist{\order{i},\order{j}} 
    + \prod_{t=1}^n \itempr{\order{t}} \sum_{i, j \leq n} \dist{\order{i},\order{j}}\\
    &= \sum_{\kprefix=1}^{n-1} \prod_{t=1}^{\kprefix+1} \itempr{\order{t}} (- \sum_{i,j\leq \kprefix} \dist{\order{i}, \order{j}} + \sum_{i,j\leq \kprefix+1} \dist{\order{i}, \order{j}}) \\
    &= \sum_{\kprefix=1}^{n-1} \prod_{t=1}^{\kprefix+1} \itempr{\order{t}} (\sum_{i = \kprefix+1,j\leq \kprefix} \dist{\order{i}, \order{j}} + \sum_{j = \kprefix+1,i\leq \kprefix} \dist{\order{i}, \order{j}}) \\
    &= 2 \sum_{\kprefix=1}^{n-1} \prod_{t=1}^{\kprefix+1} \itempr{\order{t}} (\sum_{i = \kprefix+1,j\leq \kprefix} \dist{\order{i}, \order{j}} ) 
    = 2 \sum_{\kprefix=2}^{n} \prod_{t=1}^{\kprefix} \itempr{\order{t}} (\sum_{i = \kprefix,j\leq \kprefix-1} \dist{\order{i}, \order{j}} ) \\
    &= 2 \sum_{i=2}^{n} \sum_{j=1}^{i-1} \prod_{t=1}^{i} \itempr{\order{t}} \dist{\order{i}, \order{j}}.
\end{aligned}
\end{equation}
Notice that 
equality $(a)$ holds by the definition of $\order{i}$: 
indeed, $\order{i}$ indicates the item that is placed at the $i$-th position. 
Equality $(b)$ holds by substituting the expression of $\pr(\randomordSet = \ordSet{\kprefix})$ from the \cref{obs:prob-of-A}.
For the following equations, we re-arrange the simplify the formulas. 
\end{proof}

We give a concrete example of how the ordering of the items influences the sequential sum diversity \osdo. 
\begin{example}
    Let $\unordSet = \{u_1, u_2, u_3\}$, where 
    $\dist{u_1, u_2} = 0.3$, $\dist{u_1, u_3} = 1$ and $\dist{u_2, u_3} = 1$. 
    Let $\itempr{u_1} = \itempr{u_2} = 1$ and $\itempr{u_3} = 0$. 
    The sequential sum diversity scores differ according to the order of the items, specifically:
\begin{itemize}
    \item $\osdo((u_1, u_2, u_3)) = \osdo((u_2, u_1, u_3)) = 0.3$, as the user accepts both $u_1$ and $u_2$.
    \item $\osdo((u_1, u_3, u_2)) = \osdo((u_2, u_3, u_1)) = 0$, as the user only accepts $u_1$ or $u_2$.
    \item $\osdo((u_3, u_1, u_2)) = \osdo((u_3, u_2, u_1)) = 0$, as the user directly quits the system after examining $u_3$.
\end{itemize}
\end{example}

\para{Complexity.}
We can show that the problem of maximizing sequential sum diversity is \NP-hard.

\begin{restatable}{theorem}{omsdNPHard}\label{thm:omsdNPHard}
  \omsd is \NP-hard, even when all $\itempr{i}$ are equal.
\end{restatable}

This result indicates our problem is non-trivial, and approximation algorithms are required
to guarantee high-quality orderings.


\section{Sequential Coverage Diversity}

In this section, we show that the problem of maximizing sequential coverage diversity, \omcd, 
is {ordered submodular}~\cite{kleinberg2022ordered}.

\begin{definition}[Ordered submodularity \cite{kleinberg2022ordered}]
    \label{def:ordered_submodular}
A sequence function~$f$ is ordered-submodular if 
for all sequence $X$ and $Y$, 
the following property holds for all elements $s$ and $\bar{s}$:
\begin{equation}
    f(X || s) - f(X) \geq f(X ||s||Y) - f(X||\bar{s}||Y)
\end{equation}
where $||$ denotes the concatenation of sequences and elements.
\end{definition}

Ordered submodularity is introduced by \citet{kleinberg2022ordered}, 
who extend the concepts of monotonicity and submodularity from set functions to sequence functions. 

\begin{restatable}{observation}{ocdoOrderedSubmodular}\label{obs:ocdo_ordered_submodular}
   The sequential coverage diversity function $\ocdo( \cdot )$ is ordered-submodular.
\end{restatable}

As proved by \citet{kleinberg2022ordered}, a simple greedy algorithm
provides a $1/2$ approximation for an ordered submodular function, 
and thus, for $\omcd$.
On the other hand, sequential sum diversity is not ordered submodular.

\begin{restatable}{observation}{omsdOrderedSubmodular}\label{obs:omsd_ordered_submodular}
   The sequential sum diversity function $\osdo( \cdot )$ is not ordered-submodular.
\end{restatable}

In the remainder of our paper, we focus on the technically more intriguing \omsd problem. 
Nonetheless, it is important to highlight that our definition of the \omcd problem, as well, 
provides a novel perspective to recommender systems and information-retrieval applications.

\section{Ordered Hamiltonian Path}
\label{section:reduction}

In the previous section, we establish that $\omsd$ is \NPhard through a reduction from the well-known \emph{clique problem}~\cite{garey1979computers}, highlighting the inherent complexity of directly solving the problem.
Furthermore, we observe that $\omsd$ is not ordered submodular, precluding the use of standard \emph{greedy} approaches for constant-factor approximations.

To address this, we introduce a new problem, \emph{maximum ordered Hamiltonian path} ($\omshp$), and show in \Cref{thm:reduction_non_uniform_p} that \omsd can be reduced to \omshp with only a constant-factor approximation loss.  
This reduction allows us to focus on \omshp, a more tractable problem, for which we propose approximation algorithms in subsequent sections.


\begin{definition}[Ordered Hamiltonian path (\ohpo)]
    \label{def:oshp}
    We are given a fin\-i\-te set $\unordSet = \{1,\ldots, n\}$ of $n$ distinct items, 
    a distance function~$\dist{\cdot, \cdot}$, and 
    probabilities $\{\itempr{1}, \ldots, \itempr{n}\}$ assigned to each item $i\in\unordSet$.
    Let $\ordSet = (\order{i})_{i=1}^n$ be the ordered sequence of items of $\unordSet$ 
    according to an order~$\order$. 
    The ordered Hamiltonian path, denoted by \ohpo, is defined as 
    \begin{equation}
      \ordpath{\ordSet} = 
        \sum\nolimits_{i = 1}^{n-1} \wi \dist{\order{i}, \order{i+1}},
        \label{eq:ohp}
    \end{equation}
    where $\wi = \sum\nolimits_{j=i+1}^{n} \itempr{\ordSet{j}}$.
\end{definition}

Note that the \emph{ordered} Hamiltonian-path
differs significantly from the classic Hamiltonian-path \cite{gurevich1987expected}, where the edge weights 
do not depend on the \emph{order} in which edges are placed.
In contrast, in the ordered Hamiltonian-path, 
the edge weights $\dist{\order{i}, \order{i+1}}$ are multiplied by a coefficient~$\wi$, 
which depends on the order \order 
as well as on the probabilities of all the items in \unordSet. 

Next, we define the problem
of maximizing the ordered Hamiltonian path (\omshp).

\begin{problem}[\omshp]
    \label{prob:omshp}
    We are given a fin\-i\-te set $\unordSet = \{1,\ldots, n\}$ of $n$ distinct items, 
    a distance function~$\dist{\cdot, \cdot}$, and 
    probabilities $\itempr{1}, \ldots, \itempr{n}$ assigned to each item $i\in\unordSet$.
    The goal is to find an order $\ordSet^*$ of the items in \unordSet 
    so as to maximize the $\ordpath$ objective,
    that is, 
    \begin{equation*}
      \ordSet^* 
         = \argmax_{\ordSet = \order{\unordSet}} \ordpath{\ordSet},
    \end{equation*}
     where $\wi = \sum_{j=i+1}^{n} \itempr{\ordSet{j}}$.
\end{problem}



\begin{restatable}{lemma}{oshpreform}
    \label{lem:oshp-reform}
    $\ordpath{\ordSet}$ can be equivalently formulated as $$\ordpath{\ordSet} = \sum\nolimits_{i=1}^{n-1} \itempr{\ordSet{i+1}} \distpath{\ordSet{i+1}},$$
    where $\distpath{\ordSet{\kprefix}} \coloneqq \sum_{t=1}^{\kprefix-1} \dist{\order{t}, \order{t+1}}$.
\end{restatable}

We observe that $\ordpath{\ordSet}$ can be reformulated according to \Cref{lem:oshp-reform}. 
The reformulation helps us to relate $\ordpath{\ordSet}$ with $\osdo(\ordSet)$.
Notice that $\itempr{\ordSet{i+1}}$ is the common coefficient of $\distpath{\ordSet{i+1}}$ in $\ordpath{\ordSet}$ and $\dist{\order{i+1}, \ordSet{i}}$ in $\osdo(\ordSet)$. 
Hence, we essentially only compare $\distpath{\ordSet{i+1}}$ with $\dist{\order{i+1}, \ordSet{i}}$, 
for all $1\leq i \leq n-1$.

The following corollary is a direct implication from the triangle inequality, 
$2 \dist{\order{i+1}, \ordSet{i}} \geq \distpath{\ordSet{i+1}}$, for any sequence~$\ordSet$.

\begin{restatable}{corollary}{anyinequality}
\label{lemma:any_inequality}
    For any ordered sequence $\ordSet$ of items in $\unordSet$ it holds 
    \begin{align*}
    2\osdo{(\ordSet)} \geq \ordpath{\ordSet}.
    \end{align*}
\end{restatable}

Let $\ordSet^*$ be the optimal sequence of $\osdo$:  $\ordSet^* = \arg\allowbreak \max_{\ordSet = \order{\unordSet}} \osdo{(\ordSet)}$.
Let $(\order{i}, \order{i+1})$ be any pair of items swapped in $\ordSet^*$, 
and denote $\ordSetswap= (\ldots, \order{i+1}, \order{i}, \ldots)$.
By optimality, it follows that $\osdo{(\ordSet^*)} \geq \osdo{(\ordSetswap)}$.
With some simplification, we get the following lemma.  

\begin{restatable}{lemma}{localoptimalnonuniformp}
\label{lemma:local_optimal_nonuniform_p}
    For an optimal solution $\ordSet^* = \argmax \osdo{(\ordSet)}$ to the~\omsd problem, it holds 
    \begin{align*}
        \dist{\order{i+1}, \ordSet{i}^{*}} \leq \frac{1-\itempr_{\order{i+1}}}{\itempr_{\order{i+1}}} \sum\nolimits_{j=1}^{i} \frac{\itempr_{\order{j+1}}}{1-\itempr_{\order{j+1}}} \dist{\order{j},\order{j+1}},
    \end{align*}
     for any $1 \leq i \leq n-1$. 
\end{restatable}

We observe that $\frac{1 - \itempr{i}}{\itempr{i}} \leq \frac{1-a}{a}$, 
and $\frac{\itempr{i}}{1- \itempr{i}} \leq \frac{b}{1-b}$ if $\itempr{i} \in [a, b]$.
The following corollary holds directly by the implication of
\Cref{lemma:local_optimal_nonuniform_p} that $\dist{\order{i+1}, \ordSet{i}^{*}} \leq \frac{b(1-a)}{a(1-b)}\distpath{\ordSet{i+1}^*}$.

\begin{restatable}{corollary}{optimalinequalitynonuniformp}
\label{lemma:optimal_inequality_nonuniform_p}
    For an optimal sequence  $\ordSet^* = \argmax \osdo{(\ordSet)}$ to the~\omsd problem, 
    and assuming that $\itempr{i} \in [a, b]$, for all $i \in \unordSet$, 
    it holds that $\osdo{(\ordSet^*)} \leq \frac{b(1-a)}{a(1-b)}\ordpath{\ordSet^*}$.
\end{restatable}

The main result of this section is the following.
\begin{restatable}{theorem}{reductionnonuniformp}
\label{thm:reduction_non_uniform_p}
    Let $\itempr{i} \in [a, b]$, with $0<a < b<1$.
    An $\alpha$-approximation solution for $\omshp$
    is a $\frac{a(1-b)}{2b(1-a)} \alpha$-approximation for $\omsd$.
\end{restatable}

\begin{proof}
    Let $\ordSet^*$ be an optimal sequence for the $\omsd$ problem, 
    and $\ordSet^o$ be an optimal sequence for the $\omshp$ problem.
    Let $\ordSet^{\alpha}$ be an $\alpha$-approximation sequence for $\omshp$. 
    It holds
        \begin{align*}
        \osdo{(\ordSet^*)} 
        & \stackrel{(\dagger)}{\leq} \frac{b(1-a)}{a(1-b)}\ordpath{\ordSet^*} 
          \leq \frac{b(1-a)}{a(1-b)}\ordpath{\ordSet^o} \\ 
        & \leq  \frac{b(1-a)}{a(1-b)} \frac{\ordpath{\ordSet^{\alpha}}}{\alpha} 
          \stackrel{(\ddagger)}{\leq}  \frac{2 b(1-a)}{\alpha a(1-b)  }\osdo{(\ordSet^{\alpha})},
    \end{align*}
   where inequality $(\dagger)$ follows from \Cref{lemma:optimal_inequality_nonuniform_p}, 
   and $(\ddagger)$ from \Cref{lemma:any_inequality}. 
   Thus we proved $\frac{\osdo{(\ordSet^{\alpha})}}{\osdo{(\ordSet^*)}} \geq \frac{a(1-b)}{2b(1-a)} \alpha$. We note that for uniform continuation probabilities, i.e., $a=b$, $\frac{2 b(1-a)}{\alpha a(1-b)} = 1/2$. 
\end{proof}

\section{Uniform Continuation Probabilities}
\label{section:uniform}

In this section, we address the case where all continuation probabilities \itempr{i} are equal, and we devise approximation algorithms for \omsd under this special case.
Specifically, we assume $\itempr$ may vary as a function of $n$ and can asymptotically approach 0 or 1, excluding the trivial cases of $\itempr = 0$ or $\itempr = 1$.

We analyze three regimes for $\itempr$, and design algorithms that achieve constant-factor approximations for \omsd in each regime.
Due to space constraints, two regimes are presented in the main content, with the third deferred to the full version of this paper.

As outlined earlier, we reformulate the problem using the intemediate problem \omshp, which then leads to a solution for  \omsd.
Recall that for an order \ordSet, the ordered Hamiltonian-path objective is defined as 
$\ordpath{\ordSet} = \sum_{i = 1}^{n-1} \wi \dist{\order{i}, \order{i+1}}$, 
where $\wi = \sum_{j=i+1}^{n} \itempr{\ordSet{j}}$.
In the uniform probability setting, since $\itempr{i}=\itempr$ for all $i$,  
the coefficient $\wi$ is entirely determined by its position in the order, 
rather than the actual items being ranked.  
In this case, 
\begin{equation}
\wi = \sum\nolimits_{j=i+1}^n \itempr^j = \frac{\itempr^{i+1} - \itempr^{n+1}}{1 - \itempr}.
\label{eq:w_ai}
\end{equation}

\subsection{The best-$\kpar$ items Algorithm}
\label{sec:uniform:small}

We first consider a case where $\itempr < 1$. 
From \Cref{eq:w_ai}, we observe that the term $\itempr^{n+1}$ in the numer\-ator diminishes as $n$ increases, and the coefficient~$\wi$ is dominated by $\nicefrac{\itempr^{i+1}}{(1 - \itempr)}$.
Specifically, when $\lim_{n \rightarrow +\infty}\itempr^{n} = 0$,
$\wi$ \emph{decays  exponentially}, meaning the value of $\ordpath{\ordSet}$ is dominated by the top items in the order \ordSet. 

To obtain a better intuition for this regime, consider a scenario where $\itempr$ is very small, i.e., close to $0$.
In this case, a user is likely to quit after examining only a few items, so the ranking of the first few items becomes critical.
To maximize diversity, the most diverse items should be placed at the top, as they contribute most significantly to the objective of the ordered Hamiltonian path.

Building on this insight, we propose an algorithm for the general case where $\itempr$ is bounded away from~1.
We call this algorithm \emph{best-$\kpar$ items} (\bke), where $\kpar$ is a parameter controlling the approximation quality. The algorithm always outputs a ranking of all items, with the pseudocode provided in \Cref{alg:bke}.

\begin{algorithm}[t]
  \label[algorithm]{alg:bke}
  \caption{Best-$\kpar$ items (\bke) algorithm}
  \Input{Finite set $\unordSet$, integer $\kpar$, probabilities $\itempr{i}$, diversity function $d$}
  \Output{An ordered sequence of $\unordSet$ according to $\order$, $\ordSet{n} \gets \order{\unordSet}$}
    \If{$\itempr{i}$ are uniform}{
$\ordSet{\kpar} \gets$ Solution of \Cref{eq:best_k_nodes_eq} \;
    }
    \If{$\itempr{i}$ are non-uniform}{
$\ordSet{\kpar} \gets$ Solution of \Cref{def:bestkedge_nonuniform}  \label{line:nonuniform}\;
    }
      $\ordSet{n} \gets$ Extend $\ordSet{\kpar}$ to a Hamiltonian path if $\kpar<n$\;
  \Return $\ordSet{n}$\;
\end{algorithm}

Concretely, the {best-$\kpar$ items} (\bke) algorithm works as follows:
for any integer $\kpar$, let 
\begin{equation*}
 \ordpathhat{\ordSet{\kpar}} = \sum\nolimits_{i = 1}^{\kpar-1} \frac{\itempr^{i+1}}{1-\itempr} \dist{\order{i}, \order{i+1}}   
\end{equation*}
denote the contribution of $\ordSet{\kpar}$ to $\ordpath{\ordSet}$. 
The $\bke$ algorithm seeks the optimal $\kpar$-item sequence, denoted as $\bp_{\kpar}$, 
which maximizes $\ordpathhat{\ordSet{\kpar}}$; in other words,
\begin{equation}
\label{eq:best_k_nodes_eq}
    \bp_{\kpar} = \argmax_{\ordSet{\kpar} \sqsubseteq \order{\unordSet}} \ordpathhat{\ordSet{\kpar}}.
\end{equation}

The $\bke$ algorithm first construct the
ordered sequence of best-$\kpar$ items, $\bp_{\kpar}$, and then extends this sequence to include the remaining items in $U$.
The extension can be done using any arbitrary ordering for the remaining items, as it does not affect the algorithm's approximation ratio. However, a more refined approach, such as a greedy heuristic, could be used to append items that maximize the marginal gain on $\ordpathhat{\ordSet{\kpar}}$ at each step.

The performance of \bke algorithm is as follows.

\begin{restatable}{theorem}{pconstantapprox}
\label{thm:p-constant-approx}
Assume that $\itempr{i} = \itempr < 1$, for all $i \in \unordSet$. 
For any integer $\kpar$, where $2\leq \kpar \leq n$, 
the \bke algorithm obtains a
$(1 - \itempr^{\kpar-1} - \itempr^{n - \kpar} + \itempr^n)$-approximation for \omshp.
\end{restatable}

\begin{proof}
We start by introducing some additional notation. 
Let $\order$ be the ordering of $\unordSet$ obtained from \Cref{alg:bke}, and $\ordSet$ be the corresponding sequence. By definition, 
\begin{align*}
   \ordpath(\ordSet) = \sum\nolimits_{i=1}^{n-1} \wi \dist{\order{i}, \order{i+1}}. 
\end{align*}

Let $\optorder$ be the optimal ordering for $\omshp$, 
and~$\optordSet$ be the corresponding optimal sequence. 
Hence, 
\begin{align*}
\ordpath(\optordSet) = \sum\nolimits_{i=1}^{n-1} \woi \dist{\optorder{i}, \optorder{i+1}}.
\end{align*}

Since $\itempr{i} = \itempr$, for all $i \in \unordSet$, 
we notice that $\itempr{\ordSet{j}} = \itempr{\optordSet{j}} = \itempr^j$. 
Hence, $\wi = \woi = \sum_{j=i+1}^n \itempr^j$, after expansion, we get
\begin{align*}
     \sum_{j=i+1}^n \itempr^j &\stackrel{(a)}{=} \itempr^{i+1} \frac{1 - \itempr^{n-i}}{1 - \itempr} 
     =  \frac{\itempr^{i+1}}{1-p} - \frac{\itempr^{n+1}}{1-p},
\end{align*}
where equality $(a)$ holds by \Cref{lem:geo-sum}.

Recall that the $\bke$ algorithm 
finds the optimal $\kpar$-item sequence that maximizes 
the following equation:
\begin{equation}
    \ordSet{\kpar} = \argmax_{\ordSet{\kpar} \sqsubseteq \order{\unordSet}}  \sum_{i=1}^{\kpar-1}\frac{\itempr^{i+1}}{1-\itempr} \dist{\order{i},\order{i+1}}.
\end{equation}

For the convenience in our calculations, 
we let $$\ell_{\order} = \sum_{i=1}^{\kpar-1} \frac{p^{i+1}}{1-p} \dist{\order{i}, \order{i+1}}.$$

Let us derive a lower bound on $\ordpath(\ordSet)$ as follows:
\begin{equation}
\label{eq_bk_alg}
\begin{aligned} 
        \ordpath(\ordSet)  &= \sum_{i=1}^{n-1} \frac{\itempr^{i+1}}{1-p} \dist{\order{i}, \order{i+1}} - \sum_{i=1}^{n-1} \frac{\itempr^{n+1}}{1-p} \dist{\order{i}, \order{i+1}} \\
        &\geq \sum_{i=1}^{\kpar-1} \frac{\itempr^{i+1}}{1-p} \dist{\order{i}, \order{i+1}} - \sum_{i=1}^{\kpar-1} \frac{\itempr^{n+1}}{1-p} \dist{\order{i}, \order{i+1}} \\
    & =  \ell_{\order} - \frac{\itempr^{n+1}}{1-p} \sum_{i=1}^{\kpar-1}\dij.\\
\end{aligned}
\end{equation}

Next, we derive an upper bound on $\ordpath{\optordSet}$. 
Let $T = \left\lceil \frac{n}{\kpar - 1} \right\rceil$,
\begin{equation}
\label{eq_bk_opt}
\begin{aligned}
    \ordpath(\optordSet) & = \sum_{i=1}^{n-1} \frac{\itempr^{i+1}}{1-p} \dist{\optorder{i}, \optorder{i+1}} - \sum_{i=1}^{n-1} \frac{\itempr^{n+1}}{1-p} \dist{\optorder{i}, \optorder{i+1}} \\
    &\leq  \sum_{i=1}^{n-1} \frac{\itempr^{i+1}}{1-\itempr} \dist{\optorder{i}, \optorder{i+1}}  \leq \sum_{t=1}^{T} \sum_{i=(t-1)(\kpar-1) + 1}^{t\cdot (\kpar-1)} \frac{\itempr^{i+1}}{1-\itempr} \doij\\
    &\stackrel{(a)}{\leq} \ell_{\order}  (1 + \itempr^{\kpar-1} + \cdots + \itempr^{(\numseg-1)(\kpar-1)} )  \leq \frac{\ell_{\pi}}{1 - \itempr^{\kpar-1}}.
\end{aligned}
\end{equation}
Notice that $(a)$ holds due to the optimality of $\ell_{\order}$. Notice that $T \cdot (\kpar-1)$ might be greater than $n$, we can hypothetically set $d(\pi^o(i), \pi^o(i+1)) =0$ for $i \geq n$, for our analysis.

Combining Equations~(\ref{eq_bk_alg}) and~(\ref{eq_bk_opt}), we get

\begin{equation*}
\begin{aligned}
    \frac{\ordpath(\ordSet)}{\ordpath(\optordSet)} &\geq (1 - \itempr^{\kpar-1}) \frac{ \ell_{\order} - \frac{\itempr^{n+1}}{1-p} \sum_{i=1}^{\kpar-1}\dij}{\ell_{\order}} \\
     & = (1 - \itempr^{\kpar-1}) \left(1 - \frac{\frac{\itempr^{n+1}}{1-p} \sum_{i=1}^{\kpar-1}\dij}{\sum_{i=1}^{\kpar-1} \frac{p^{i+1}}{1-p} \dist{\order{i}, \order{i+1}}}\right) \\
      & \geq (1 - \itempr^{\kpar-1}) \left(1 - \frac{\frac{\itempr^{n+1}}{1-p} \sum_{i=1}^{\kpar-1}\dij}{\sum_{i=1}^{\kpar-1} \frac{p^{\kpar}}{1-p} \dist{\order{i}, \order{i+1}}}\right)\\
    & = (1 - \itempr^{\kpar-1})(1 - \itempr^{n-\kpar+1}) = 1 - \itempr^{\kpar-1} - \itempr^{n-\kpar+1} + \itempr^n 
    \geq 1 - \itempr^{\kpar-1} - \itempr^{n-\kpar} + \itempr^n.
\end{aligned}
\end{equation*}
We have demonstrated that the $\bke$ algorithm provides an asymptotic $(1 - \itempr^{\kpar-1} - \itempr^{n-\kpar} + \itempr^n)$-approximation.
\end{proof}

Notice that when there is a constant $\epsilon >0$ 
such that $\itempr < 1 - \epsilon$, and $\kpar \ll n$, both $\itempr^{n - \kpar}$ and $\itempr^n$ become negligible. In this case, \bke is a $\left(1 - \itempr^{\kpar-1} - \Theta(p^n)\right)$-approximation algorithm.  
This aligns with our intuition: optimizing the dominant term $\ordpath(\ordSet{\kpar})$ yields a good approximation for the $\omshp$ problem.

Observe that searching the best $\kpar$ items for $\bp_{\kpar}$ requires $\bigO(\binom{n}{\kpar} \kpar!)$ time. 
\Cref{thm:p-constant-approx} 
illustrates the inherent trade-off between approx\-i\-ma\-tion-quality and efficiency:
{increasing $\kpar$ improves the approximation but also increases the computational cost of the algorithm.}

\begin{algorithm}[t]
  \label[algorithm]{alg:bkm}
  \caption{Greedy matching (\bkm) algorithm}
  \Input{Finite set $\unordSet$, diversity function $\dist$}
  \Output{An ordered sequence of $\unordSet$ according to $\order$, $\ordSet{n} = \order{\unordSet}$} 
  $\kpar \gets \lfloor \frac{\abs{\unordSet}}{2} \rfloor$, $M \gets $ empty sequence, $S \gets \emptyset$\;
  $E \gets$ all possible pairs $(u, v)$ in decreasing order of $d(u, v)$\;

  \For{$(u, v) \in E$}{
  \If{$M || (u,v)$ forms a matching}{
  $M \gets M || (u, v)$\;
  $S \gets S \cup \{u,v\}$}
  } 
  Naming the edges in $M$: $M \gets ((u_1, v_1),(u_3, v_3) \cdots, (u_{2{\kpar}-1},v_{2{\kpar}-1}))$\;
  \If{$\abs{U} = 2\kpar + 1$}{
    Let $v_{2{\kpar}+1} \in \unordSet \setminus S$, $\order{2{\kpar}+1} \gets v_{2{\kpar}+1}$\;
  }
  \For{$i \gets \kpar$ to $1$}{
    \If{$i=\kpar$ and $\abs{U} = 2\kpar$}{$\order{2{\kpar}-1} \gets v_{2{\kpar}-1}$, $\order{2\kpar} \gets u_{2{\kpar}-1}$\;}
    \uIf{$\dist{v_{2i-1}, \order{2i + 1}} \geq \dist{u_{2i-1}, \order{2i + 1}}$}{$\order{2i-1} \gets u_{2i-1}$, $\order{2i} \gets v_{2i-1}$}
    \Else{$\order{2i-1} \gets v_{2i-1}$, $\order{2i} \gets u_{2i-1}$}
  }
  \Return $\ordSet{n} \gets \order{\unordSet}$\;
\end{algorithm}

\subsection{The Greedy-matching Algorithm}
\label{section:uniform:large:2}

In this section, we examine the case where $\itempr$ is close to 1, but not a constant, 
specifically, when $\lim_{n \rightarrow \infty} \itempr = 1$ and $\lim_{n \rightarrow \infty} \itempr^n = 0$.
This suggests that as the number of items increases, the user accepts more items, but eventually quits before examining all of them.
For instance, setting $\itempr = 1 - \frac{1}{\log\log n}$ implies that, in expectation, the user accepts $\log\log n$ items before quitting. For simplicity, we assume $\itempr = 1 - \frac{1}{t_n}$, where $t_n$ is an increasing function of $n$.

In this scenario, applying \bke algorithm to find best-$\kpar$ items becomes computationally infeasible. As \Cref{thm:p-constant-approx} shows, achieving a constant approximation requires setting $\kpar = \Omega(t_n)$. To address this, we introduce a polynomial-time \emph{greedy-matching} ($\bkm$) algorithm.

The idea behind $\bkm$ is to construct a matching of size-$(t_n/2)$ from $t_n$ items. We observe that for the first $t_n$ items, the coefficient $\wi$ does not decrease significantly, since $\itempr^{t_n} = 1/e - \Theta(1/t_n)$. 
Intuitively, the algorithm should find a path of length $t_n$ such that the sum of the edge weights is large.
In addition, the algorithm should assign larger $\wi$ to the edges with larger weights; in other words, the edge weights should be ordered in decreasing order in the path.
The proposed $\bkm$ algorithm, 
presented in \Cref{alg:bkm}, achieves both goals.

The \bkm algorithm starts by applying a greedy algorithm to obtain a matching~$M$,
which is a sequence of edges in decreasing order of edge weights. 
The size of $M$ is equal to $\lfloor \frac{n}{2}\rfloor$. 
We denote the matching as $M = ((u_{2j-1}, v_{2j-1}))_{j=1}^{\lfloor \frac{n}{2}\rfloor}$.

The next step is to extend $M$ to a Hamiltonian path. 
Notice that the (undirected) edge $(u_{2j-1}, v_{2j-1})$ is always the $(2j-1)$-th edge on the path, 
and the algorithm needs to decide which item comes first in the path. 
To make this decision, 
the algorithm utilizes the triangle inequality. 

Let $\order$ be order of the items on the Hamiltonian path, 
as obtained by the $\bkm$ algorithm. 
The above-mentioned construction satisfies two properties:
\begin{enumerate}
    \item \label{property:edge_weight_decreasing} $d((\order{2i-1},\order{2i}))\geq d((\order{2i+1},\order{2i+2}))$, for $i < \lfloor \frac{n}{2}\rfloor$;
    \item \label{property:even_greater_than_half_odd} $d((\order{2i},\order{2i+1})) \geq \frac{1}{2}d((\order{2i-1},\order{2i}))$, for $i \leq \lfloor \frac{n}{2}\rfloor$.
\end{enumerate}

For the approximation quality of \bkm, the following holds.

\begin{restatable}{theorem}{pnonconstantsmalltwo}
\label{thm:p-nonconstant-small-two}
Assume that $\itempr{i} = \itempr$, for any $i \in \unordSet$. 
Let $\itempr$ be a function of $n$ such that $\lim_{n \rightarrow \infty} \itempr =1$ 
and  $\lim_{n \rightarrow \infty} \itempr^n = 0$. 
Moreover, assume that $\itempr = 1 - \frac{1}{t_n}$, 
where $t_n$ is a increasing function of $n$  and $t_n = o(n)$. 
The algorithm \bkm yields 
a ~$\left(\frac{3(e-1)}{16 e^2} - \Theta \left(\frac{1}{t_n}\right)\right)$-approximation 
for \omshp.
\end{restatable}

\section{Non-uniform Continuation Probabilities}
\label{section:nonuniformcase}

In this section, we discuss a more general case when the continuation probabilities $\itempr{i}$ are non-uniform. 
We assume $\itempr{i} \in [a,b]$, for all $i \in \unordSet$, where $0<a \leq b<1$. 
We start by adapting the $\bke$ algorithm, introduced in \Cref{section:uniform}, to this setting. 
Given the computational complexity of the $\bke$ algorithm, 
we also propose an efficient greedy algorithm, 
whose performance guarantee is an implication of the $\bke$ algorithm by setting $\kpar = 2$.

\subsection{The best-$\kpar$ items Algorithm}

Recall that the ordered Hamiltonian-path is defined as 
$\ordpath{\ordSet} = \sum_{i = 1}^{n-1} \wi \dist{\order{i}, \order{i+1}}$, 
where $\wi = \sum_{j=i+1}^{n} \itempr{\ordSet{j}}$.
When the continuation probabilities $\itempr{i}$ are uniformly equal to $\itempr$,
we observe that the value of $\wi$ is dominated by $\nicefrac{\itempr^{i+1}}{(1 - \itempr)}$, hereby, we can use $\nicefrac{\itempr^{i+1}}{(1 - \itempr)}$ to approximate $\wi$. 
However, when $\itempr{i}$ are not uniform, it is hard to find a simple and closed-form equation to approximate $\wi$, which is convenient for our analysis. 

Our solution is to simply truncate $\wi$, 
so that the algorithm can determine the best-$\kpar$ items 
by only using the information of the first $\kpar$ items. 
In particular, we define
\begin{equation*}
    \ordpathtilde{\ordSet{\kpar}} = \sum\nolimits_{i=1}^{\kpar-1} \sum\nolimits_{j=i+1}^{\kpar} \itempr{\ordSet{j}} \dist{\order{i}, \order{i+1}},
\end{equation*}
and we find the best-$\kpar$ items by maximizing $\ordpathtilde{\ordSet{\kpar}}$.
We redefine
\begin{equation}
\label{def:bestkedge_nonuniform}
    \bp_{\kpar} = 
     \arg\max\nolimits_{\ordSet{\kpar} \sqsubseteq \order{\unordSet}} \ordpathtilde{\ordSet{\kpar}}.
\end{equation}
The \bke algorithm (\Cref{alg:bke})
has the following quality guarantee.
\begin{restatable}{theorem}{approbkenonuniform}
\label{thm:appro_bke_nonuniform}
    Let $\itempr{i} \in [a,b]$ for all $i \in U$, where $0 < a \leq  b < 1$.
    Let $\kpar$ be an integer such that $2 \leq \kpar \leq n$.
    The \bke algorithm provides a $\frac{a^2(1 - b)(1-b^{\kpar-1}) }{a^2 + (\kpar-1) b^{\kpar+1} }$-approximation guarantee
    for $\omshp$.
\end{restatable}

\subsection{The Greedy Algorithm}
Since the $\bke$ algorithm is impractical when $\kpar$ is large, we propose a greedy algorithm, presented in \Cref{alg:greedy_MSD}. Its approximation ratio is obtained by setting $\kpar = 2$ in \Cref{thm:appro_bke_nonuniform}.

\begin{restatable}{corollary}{greedyapproxnonuniform}
\label{corollary:greedy_approx_nonunoformp}
The greedy algorithm (\Cref{alg:greedy_MSD}) provides a $\frac{a^2(1 - b)^2}{a^2 + b^2}$-approximation guarantee for $\omshp$.
\end{restatable}

\begin{algorithm}[t]
  \label[algorithm]{alg:greedy_MSD}
  \caption{Greedy algorithm}
  \Input{Finite set $\unordSet$, probability $\itempr_i$, diversity function $\dist$}
  \Output{An ordered sequence of $\unordSet$ according to $\order$, $\ordSet{n} = \order{\unordSet}$} 
  $\ordSet{2} \gets$ Solution of \Cref{def:bestkedge_nonuniform} by setting $\tau = 2$ \;
  
\For{$t = 2;\ t < n;\ t = t + 1$}{
		$u \gets \argmax_{v \in U \setminus \ordSet{t}} \osdo(\ordSet{t} || v)- \osdo(\ordSet{t})$\;
		$\ordSet{t+1} \gets \ordSet{t} || u$\;
}
  \Return $\ordSet{n}$.\;
\end{algorithm}


\section{Directly Optimizing the \omsd Problem}
In this section, we optimize \osdo directly. We again consider two settings: when the continuation probability is uniform, and when it is non-uniform. For each setting, we propose a greedy algorithm and establish theoretical guarantees.

\subsection{Uniform Continuation Probabilities}

According to \Cref{eq:eqform_sequential_diversity}, we can represent the \emph{sequential sum diversity} of a sequence $\ordSet$, $\osdo(\ordSet)$, as follow:

\begin{equation}
\label{eq:osdo_exp_formulation}
\begin{aligned}
    \osdo(\ordSet) 
    &= \sum_{\kprefix=1}^{n-1} \prod_{t=1}^\kprefix \itempr{\order{t}}(1 - \itempr{\order{\kprefix+1}})  \left[\sum\nolimits_{i,j \in \ordSet{\kprefix}} \dist{i,j}\right]+ \prod_{t=1}^n \itempr{\order{t}} \left[\sum\nolimits_{i,j \in \ordSet} \dist{i,j}\right]
\end{aligned}
\end{equation}

Let $\diversity(\randomordSet) = \sum\nolimits_{i,j \in \randomordSet} \dist{i,j}$ denote the max-sum diversification measure that has been well studied in the literature \cite{ravi1994heuristic, gollapudi2009axiomatic, borodin2012max}. When the item continuation probability $\itempr$ is uniform across all items, \Cref{eq:osdo_exp_formulation} is equivalent as follows.
\begin{equation}
\label{eq:osdo_div_formulation}
\begin{aligned}
  \osdo(\ordSet)& =\sum_{\kprefix=1}^n \pr(\randomordSet= \ordSet{\kprefix}) \diversity(\ordSet{\kprefix}) = \sum_{\kprefix=1}^{n-1} p^k (1-p)  \diversity(\ordSet{\kprefix}) + \itempr^n \diversity(\ordSet)
\end{aligned}
\end{equation}

We observe that: (1) when all items have the same continuation probability $\itempr$, the coefficient $\pr(\randomordSet= \ordSet{\kprefix})$ for the term $\diversity(\ordSet{\kprefix})$ is identical for all possible orderings $\ordSet$, and (2) $\pr(\randomordSet= \ordSet{\kprefix})$ is non-increasing in $\kprefix \in \{1, \ldots, n-1\}$, i.e., top-rank items in the ordering \ordSet contributes larger than low-rank items to the \osdo value.
Thus, to maximize \osdo, we can ignore the effect of $p$ and construct an ordering that greedily optimizes $\diversity(\ordSet{\kprefix})$ for all $\kprefix \in \{2, \cdots, n\}$. We show this method in \Cref{alg:greedy_for_diversity}.

\begin{algorithm}[t]
  \label[algorithm]{alg:greedy_for_diversity}
\caption{Greedy Algorithm for $\mathit{\omsd}$ with uniform $p$ \cite{borodin2012max}}
  \Input{Finite set \unordSet, probability \itempr, diversity function \diversity}
  \Output{An ordered sequence of $\unordSet$} 
  $\divS \gets \emptyset$ \;
  \While{$|\divS| < n$}{
    $u \gets \argmax_{u \in \unordSet \setminus \divS} \diversity(\divS)$ \;
    $\divS \gets \divS + u$ \;
  }
  \Return $\divS$ \;
\end{algorithm}

\begin{restatable}{theorem}{onlydistancegreedy}
\label{corollary:greedy_only_distance}
Assume that $\itempr{i} = \itempr$, for all $i \in \unordSet$. 
\Cref{alg:greedy_for_diversity} provides a $\nicefrac{1}{2}$-approximation guarantee for $\omsd$.
\end{restatable}

\begin{proof}

Let $\divSopt{\kprefix} = \argmax_{S \subseteq \unordSet, |S| \leq \kprefix} \diversity(S)$. According to \citet{borodin2012max}, the classical greedy algorithm achieves a $\nicefrac{1}{2}$-approximation for the problem $\max_{S \subseteq \unordSet, |S|\leq k}\diversity(S)$. Thus, let $\divS$ denote the solution returned by \Cref{alg:greedy_for_diversity}, and let $\divS{\kprefix}$ denote the obtained sequence of length $k$ after the $\kprefix$-th iteration. Then $\diversity(\divS{\kprefix}) \geq \nicefrac{1}{2} \diversity(\divSopt{\kprefix})$ for all $\kprefix \in \{1, \ldots, n\}$.

Recall that $\ordSetstar$ denote the optimal sequence for $\osdo$: $\ordSetstar = \argmax_{\ordSet = \order{\unordSet}} \osdo(\ordSet)$. It holds that
\begin{equation}
\label{}
\begin{aligned}
  \osdo(\divS)& =\sum_{\kprefix=1}^n \pr(\randomordSet= \divS{\kprefix}) \diversity(\divS{\kprefix})  \stackrel{(a)}{\geq} \frac{1}{2} \sum_{\kprefix=1}^n \pr(\randomordSet= \divS{\kprefix}) \diversity(\ordSetstar{\kprefix})\\
  &\stackrel{(b)}{=}\frac{1}{2} \sum_{\kprefix=1}^n \pr(\randomordSet= \ordSetstar{\kprefix}) \diversity(\ordSetstar{\kprefix}) \stackrel{(c)}{=} \frac{1}{2} \osdo(\ordSetstar)
\end{aligned}
\end{equation}
where inequality (a) holds because $\diversity(\divS{\kprefix}) \geq \frac{1}{2} \diversity(\divSopt{\kprefix}) \geq \frac{1}{2}\diversity(\ordSetstar{\kprefix})$ for all $\kprefix \in \{1, \ldots, n\}$. Equality (b) holds because the continuation probabilities are equal for all items, and (c) holds by the definition of $\osdo$.

\end{proof}

\subsection{Non-uniform Continuation Probabilities}

When the continuation probabilities are non-uniform, the greedy algorithm in \Cref{alg:greedy_for_diversity} can perform poorly. Consider a worst-case scenario where the first two items selected by the greedy algorithm have continuation probabilities equal to 0, then the expected sequential diversity \osdo is 0. In this case, we must account for the probabilities when constructing the ranking. To this end, we propose \Sbke, which adapts the \bke algorithm by selecting the best sequence of $\kpar$ items that maximizes the sum of the first $\kpar$ terms of $\osdo$, rather than selecting the best $\kpar$ items that maximize $\ordpathtilde{\ordSet{\kpar}}$. We denote this value as $\osdo(\ordSet{\kpar})$, which is defined as follows:
\begin{align}
    \label{eq:first_k_term_osdo}
    \osdo(\ordSet{\kpar}) = \sum\nolimits_{i=1}^{\kpar-1} \itempr{\ordSet{i+1}} \dist{\order{i+1}, \ordSet{i}}.
\end{align}

\begin{restatable}{theorem}{approbkenonuniform}
\label{thm:appro_bke_nonuniform}
    Let $\itempr{i} \in [a,b]$ for all $i \in U$, where $0 < a \leq  b < 1$.
    Let $\kpar$ be an integer such that $2 \leq \kpar \leq n$.
    The \Sbke algorithm, which essentially runs \bke algorithm by setting $\ordSet{\kpar} = \argmax_{\ordSet{\kpar} \sqsubseteq \order{\unordSet}} \osdo(\ordSet{\kpar})$, provides a $\frac{a^\kpar (1-b)(1-b^{\kpar-1})}{a^\kpar(1-b) + b^{\kpar+1}}$-approximation guarantee
    for $\omsd$.
\end{restatable}

\begin{proof}

Let $\ordSetstar$ denote the optimal ranking for the \omsd problem, and let $\orderstar$ denote the corresponding optimal permutation function.
We partition the $n-1$ terms of $\osdo(\ordSetstar)$ into $\numchunk = \lceil \frac{n}{\kpar-1}\rceil$ chunks, with each of the first $\numchunk-1$ chunks containing $\kpar-1$ terms, and denote the sum of the $j$-th chunk as $\chunk{j}$, which is calculated as follows:
\begin{equation*}
\chunk{j} = \sum_{i = (j-1)(\kpar-1) + 1}^{j(\kpar -1)} \prod_{t=1}^{i+1} \itempr{\orderstar{t}} \dist{\orderstar{i+1}, \ordSetstar{i}}
\end{equation*}
where we set $\dist{\orderstar(i+1), \ordSetstar(i)} = 0$ for $i \geq n$, and it holds that $\osdo(\ordSetstar) = \sum_{j=1}^{\numchunk} \chunk{j}$.

Let $L_{max} = \max_{\ordSet{\kpar-1} \sqsubseteq \order{\unordSet}} \osdo(\ordSet{\kpar})$, then it holds that $\chunk{1} \leq \chunkmax$. Let $\ell = (j-1)(\tau-1)$, we can upper bound $\chunk{j}$ with \chunkmax as follows
\begin{align*}
    \chunk{j} & = \prod_{k=1}^{\ell } \itempr{\orderstar{k}} \left[ \sum_{i = \ell + 1}^{j(\kpar -1)} \prod_{t=\ell+1}^{i+1} \itempr{\orderstar{t}} \dist{\orderstar{i+1}, \ordSetstar{i}}  \right] \leq b^{\ell} \left[ \sum_{i = \ell + 1}^{j(\kpar -1)} \prod_{t=\ell+1}^{i+1} \itempr{\orderstar{t}} \dist{\orderstar{i+1}, \ordSetstar{i}}  \right]\\
    & = b^{\ell} \Bigg[ 
    \sum_{i = \ell + 1}^{j(\kpar -1)} \prod_{t=\ell+1}^{i+1} \itempr{\orderstar{t}} \dist{\orderstar{i+1}, \Big\{\orderstar{\ell + 1},\cdots, \orderstar{i}\Big\}} \\
    & \qquad +  \sum_{i = \ell + 1}^{j(\kpar -1)} \prod_{t=\ell+1}^{i+1} \itempr{\orderstar{t}} \dist{\orderstar{i+1}, \Big\{\orderstar{1},\cdots, \orderstar{\ell}\Big\}}
    \Bigg]\\
    & \stackrel{(a)}{\leq} b^{\ell} \chunkmax 
     +  \sum_{i = \ell + 1}^{j(\kpar -1)} b^{i+1}\dist{\orderstar{i+1}, \Big\{\orderstar{1},\cdots, \orderstar{\ell}\Big\}}   \stackrel{(b)}{\leq} b^{\ell} \chunkmax 
     +   \frac{(j-1) b^{\ell}\chunkmax}{a^\kpar} \sum_{i=2}^{\kpar} b^i \\
\end{align*}
where inequality (a) holds by the definition of $\chunkmax$, and inequality (b) holds because for any $\kpar$ items $\{x_1,\ldots, x_\kpar\}$, we have $\dist(x_\kpar,\{x_1,\ldots, x_{\kpar-1}\}) \leq \sum_{i \neq j \leq \kpar} \dist(x_i,x_j) = \dist(x_1, x_2) + \dist(x_3, \{x_1, x_2\}) + \cdots + \dist(x_\kpar, \{x_1, \ldots, x_{\kpar-1}\}) \leq \frac{\chunkmax}{a^{\kpar}}$, and thus
\begin{align*}
& \dist{\orderstar{i+1}, \Big\{\orderstar{1},\cdots, \orderstar{(j-1)(\kpar-1)}\Big\}} \\ & = \sum_{t=1}^{j-1} \dist(\orderstar{i+1}, \Big\{\orderstar{(t-1)(\kpar-1)+1},\cdots, \orderstar{t(\kpar-1)}\Big\}) \leq \frac{(j-1)\chunkmax}{a^\kpar} 
\end{align*}

Summing over all chunks, we obtain
\begin{align*}
    \osdo(\ordSetstar) & = \sum_{j=1}^{\numchunk} \chunk{j} \leq \sum_{t=0}^{\numchunk-1} b^{(\kpar-1)t} \chunkmax + (\sum_{t=2}^{\kpar} b^t) (\sum_{t=1}^{T-1} t b^{t(\kpar-1)}) \frac{\chunkmax}{a^\kpar} \\
    & \stackrel{(a)}{=}  \chunkmax \left[ \frac{1 - b^{T(\kpar-1)}}{1 - b^{\kpar-1}} + \frac{ b^2(1 - b^{\kpar-1})}{a^ \kpar (1-b)} \cdot \frac{b^{\kpar-1}\left(1 - T b^{(T-1)(\kpar-1)} + (T-1)b^{T(\kpar-1)}\right)}{(1 - b^{\kpar-1})^2} \right] \\
    & \leq \chunkmax \left[\frac{1}{1 - b^{\kpar-1}} + \frac{ b^2 (1 - b^{\kpar-1})}{a^\kpar(1-b)} \cdot \frac{b^{\kpar-1}}{(1 - b^{\kpar-1})^2}\right] = \chunkmax \frac{a^\kpar(1-b) + b^{\kpar+1}}{a^\kpar (1-b)(1-b^{\kpar-1})}
\end{align*}
where inequality (a) holds because 
\[
\sum_{t=1}^{T-1} t \cdot b^{t} = \frac{b\left(1 - T \cdot b^{(T-1)} + (T-1)b^{T}\right)}{(1 - b)^2}
\]
is a standard arithmetico-geometric series summation, as documented in Equation 2.26 of \citet{graham1994concrete}.

Let $\ordSet$ denote the output of the \bke algorithm with $\ordSet{\kpar} = \argmax_{\ordSet{\kpar} \sqsubseteq \order{\unordSet}} \osdo(\ordSet{\kpar})$. It holds that
\begin{equation*}
    \frac{\osdo(\ordSet)}{\osdo(\ordSetstar)} \geq \frac{\chunkmax}{ \chunkmax \frac{a^\kpar(1-b) + b^{\kpar+1}}{a^\kpar (1-b)(1-b^{\kpar-1})}} \geq  \frac{a^\kpar (1-b)(1-b^{\kpar-1})}{a^\kpar(1-b) + b^{\kpar+1}}
\end{equation*}
\end{proof}

\subsection{\bke and \Sbke Approximation Ratio Analysis}

\edit{Let $A$ denote the approximation ratio of the \Sbke algorithm and $B$ denote the approximation ratio of the \bke algorithm, where 
$A = \frac{a^\kpar (1-b)(1-b^{\kpar-1})}{a^\kpar(1-b) + b^{\kpar+1}}$
and 
$B = \frac{a^3(1 - b)^2(1-b^{\kpar-1})}{2b(1-a)(a^2 + (\kpar-1) b^{\kpar+1})}.$ We observe that both approximation ratios $A$ and $B$ exhibit consistent monotonicity properties with respect to the parameters $a$ and $b$. Specifically, for fixed $b$, both $A$ and $B$ are monotonically increasing in $a$. Conversely, for fixed $a$, both $A$ and $B$ are monotonically decreasing in $b$. When both parameters increase simultaneously along the feasible region defined by $a \leq b$, the decreasing effect with respect to $b$ dominates, leading to an overall decrease in both approximation ratios. This dominance can be attributed to the higher-order dependence on $b$ in the denominators of both expressions. 

To compare the performance of these two algorithms, we numerically compare their approximation ratios by visualizing $A-B$. We show the plots in \Cref{fig:approximation_ratio_compare}. For small \kpar (3, 4, 5), $A > B$ holds in the majority of the valid region (100.0 \%, 99.7\%, and 84.1\%, respectively). As $\kpar$ increases (5, 6, 7), the region where $A < B$ expands, with $A > B$ holding in only 65.5\%, 52.1\%, and 42.7\% of the valid region, respectively. This suggests that \Sbke theoretically outperforms \bke when $\tau$ is small, while \bke becomes more competitive as $\kpar$ increases. Consequently, neither algorithm consistently dominates the other, indicating that both approaches remain theoretically valid depending on the choice of $\kpar$.
}

\begin{figure}[t]
\centering
        \includegraphics[width=0.9 \textwidth]{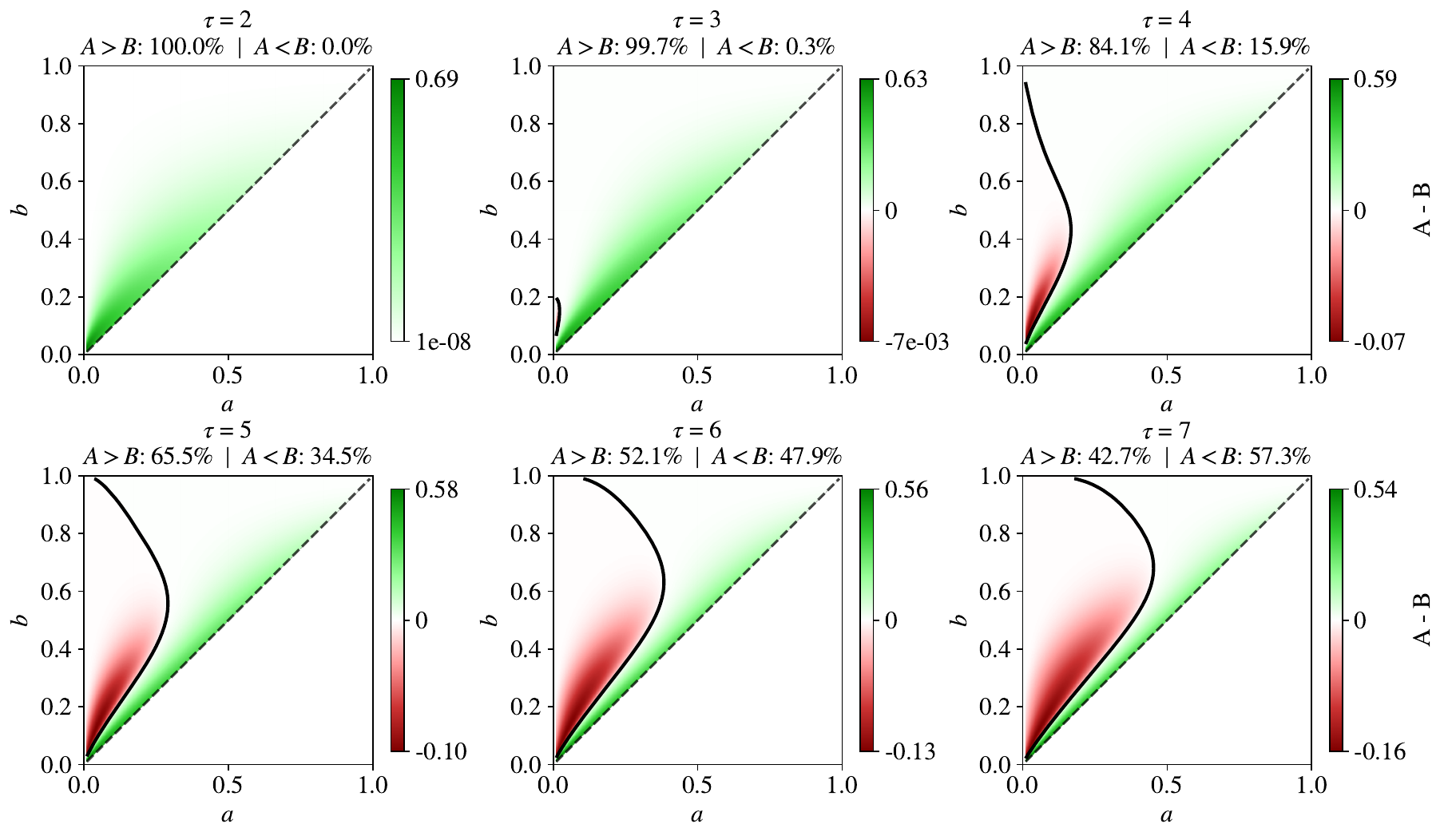}
        \caption{Heatmap of $A - B$ for different values of $\kpar$, where A is the approximation ratio of the \Sbke algorithm, and B the approximation ratio of the \bke algorithm and $0 < a \leq b < 1$. Green regions indicate $A > B$ and red regions indicate $A < B$. The black contour line marks where $A = B$, and the dashed diagonal line represents the boundary $a = b$. The region below the dashed line is invalid as it violates the constraint $a \leq b$}
        \label{fig:approximation_ratio_compare}
\end{figure}

\section{Experimental Evaluation}
\label{sec:exp}

All experiments are conducted on a dual-socket server, 
each socket housing an Intel Xeon Gold 6326 CPU at 2.90GHz, with 128 cores and 256 threads. 
Our implementation and the datasets are publicly available.\footnote{\url{https://github.com/HongLWang/Sequential-diversification-with-provable-guarantees}}

\subsection{Datasets}
\label{sec:exp:datasets}
We evaluate our methods on seven public datasets: five commonly used in recommender systems and two, \letor and \ltrc, from the information retrieval field. 
We refer to these datasets as \emph{recommendation datasets} and \emph{information retrieval datasets}, respectively.
\Cref{table:dataset_information} provides a summary of dataset statistics. Due to space constraints, the dataset URLs are provided in the full paper.

\para{\Coat}:
\Coat ratings in the range $[1,5]$, accompanied by categorical features, such as gender, jacket type, and color. 


\para{\KuaiRec} \cite{gao2022kuairec}:
Recommendation logs from a video-sharing mobile app, 
containing each video's duration, categories, and user viewing times. 
The \emph{watch ratio}, calculated as the user's viewing time divided by the video's duration, 
and normalized to range [1, 5],
serves as personalized rating. 

\para{\Netflix}:
Movie ratings in range $[1,5]$ with genre information. 
We sample $5$\,K movies and 
exclude users with less than $20$~interactions.

\para{\Movielens}:
Movie ratings in range $[1,5]$ with movie genres.

\para{\Yahoo}:
Song ratings in range $[1,5]$ with song genres. 
We sample $3$\,K songs and exclude users with less than $20$~interactions. 

\para{\letor}:
Web-search dataset with relevance scores (in $\{0,1,2\}$) on queries-document pairs and document feature vectors. 

\para{\ltrc}:
Web-search dataset with relevance scores (in $\{0,1,2,3,4\}$)
on queries-document pairs and document feature vectors.



\begin{table}[t]
    \centering
    \caption{Dataset statistics of the $5$ recommendation datasets and $2$ information retrieval datasets. 
    $\mathcal{U}$ and $\mathcal{I}$ represent the set of users (or queries) and items (or documents), respectively. 
    $\# \mathcal{R}$ denotes the number of ratings (or relevance scores), $\mathrm{avg} (\mathcal{R})$ denotes the average number of ratings per user (or documents per query), and $\mathrm{avg}(dist)$ denotes the average item-item (or document-document) distance.
    }
    \label{table:dataset_information}
    \small
    \vspace{-3mm}
        \renewcommand{\arraystretch}{1} 
        \setlength{\tabcolsep}{7pt} 
        \begin{tabular}{lrrrrr}
            \toprule
            Datasets & $|\mathcal{U}|$ & $|\mathcal{I}|$ & $\#\mathcal{R} $& $\mathrm{avg} (\mathcal{R})  $& $\mathrm{avg}(dist)$ \\
            \midrule
            \Coat          & 290   & 300   & 6\,960   & 24.0 & 0.73 \\
            \KuaiRec   & 1\,411  & 3\,327  & 4\,676\,570 & 3\,314.4 & 0.91 \\
            \Netflix & 4\,999  & 1\,112  & 557\,176  & 1\,121.5 & 0.83 \\
            \Movielens     & 6\,040  & 3\,706  & 1\,000\,208 & 165.6 & 0.83 \\
            \Yahoo      & 21\,181 & 3\,000  & 963\,296  & 45.8 & 0.26 \\ \midrule 
             \letor     & 1\,691 & 65\,316  & 69\,614  & 41.2 & 0.33 \\
             \ltrc     & 5\,154 & 146\,995  & 152\,772  &29.60 & 0.26 \\
            \bottomrule
        \end{tabular} 
\end{table}

\subsection{Experimental Setting} 
\label{subsec_expsetting}

\para{Continuation probability}. 
To better understand how the proposed best-\kpar algorithm performs across various continuation probabilities, we interpolate these probabilities into four different probability regimes: small ($[0.1, 0.3]$), medium ($[0.4, 0.6]$), large ($[0.7,0.9]$) and full ($[0.1, 0.9]$). 
The first three regimes characterize different assumptions about the users' browsing behaviors: users might all be highly likely to accept the system's recommendations or tend to easily quit the system or in the middle.
The last probability regime does not impose any specific assumptions. 

We obtain the continuation probabilities in different ways for the recommendation datasets and the information-retrieval datasets. 
For the former, we first complete the user-item rating matrix using a matrix-factorization approach 
\citep{koren2009matrix}.
The estimated ratings, which range from $[1,5]$, are then interpolated into the four probability regimes.
For the latter datasets, we treat each query and its corresponding documents as a dataset and directly interpolate the relevance scores into probability regimes.
This is meaningful as the documents of different queries have almost no intersection.

\para{Distance function}. 
For the recommendation datasets, given the categorical nature of the data, 
we use the \emph{Jaccard distance} as the item-item distance function. 
In particular, if $\mathcal{C}_i$ denotes the set of categories of item $i$, 
the Jaccard distance~$d(i,j)$ is defined as
$d(i,j) = 1- {\abs{\mathcal{C}_i \cap \mathcal{C}_j}}/{\abs{\mathcal{C}_i \cup \mathcal{C}_j}}$.
For the information-retrieval datasets, with the provided document feature vectors, 
we use the \emph{$1-$cosine distance} as the document-document distance function.

\para{Additional metrics}.
To evaluate the performance of our algorithms in terms of user engagement and satisfaction, we use two metrics: expected discounted cumulative gain (\expdcg) and expected serendipity (\expseren).
Discounted cumulative gain \cite{jarvelin2002cumulated} measures the accumulated relevance of items in a ranking list, discounted by their position. 
Serendipity \cite{herlocker2004evaluating} quantifies how relevant and unexpected the recommended items are to the user.
We adapt these metrics to account for users' continuation probabilities. 
For any user, let $\itempr_i$ be the continuation probability of item $i$.  

\expdcg is defined as 
$\sum_{j=1}^{|\mathcal{I}|}  \left(\sum_{t=1}^{j} \frac{p_t}{log_2(t+1)} \right) (1-p_{j+1})\prod_{t = 1}^{j} p_t$.
For serendipity, let $H$ be the set of previously rated items, and $\mathcal{C}(H)$ be the set of categories covered by $H$. 
\expseren is defined as 
$\sum_{j=1}^{|\mathcal{I}|}  \left(\sum_{t=1}^{j} p_t \cdot I(t,u) \right) (1-p_{j+1})\prod_{t = 1}^{j} p_t$,
where $I(t,u) = 1$ if $\abs{\mathcal{C}(H \cup \{t\})} > \abs{\mathcal{C}(H)}$ and 0 otherwise. 
Here, $p_t$ assesses the relevance of item $t$, while $I(t,u)$ captures its unexpectedness.

Note that, by definition, the \expseren score for the \letor and \ltrc datasets is $0$ as the sequences for each user comprise only rated items. Similarly, the \KuaiRec dataset also has a \expseren score of $0$, as each user has already covered all categories in their history of ratings. Therefore, we omit the \expseren scores for these datasets.

\vspace{-2.5mm}
\subsection{Proposed Methods and Baselines}

\subsubsection{Baselines} First, we present the baseline methods used in our experiments.
We reformulate the baselines using our notation to enhance clarity.
We consistently denote by $R$ the set of already-selected items.
The parameter $\lambda$ models the trade-off diversity vs.\ relevance.
We conduct a grid search to find $\lambda \in [0, 1]$ that performs the best according to the $\omsd$ objective 
for baselines \MSD, \MMR, and \DPP.
We report the results using that optimal value of~$\lambda$ for each~method.



\para{\Random}. Randomly shuffle all items into a sequence.

\para{\explore \cite{coppolillo2024relevance}}. Details on the \explore algorithm and its adaptation to our setting are provided in the full paper.

\para{Diversity-weighted utility maximization (\DUM) \cite{ashkan2015optimal}}. 
\DUM seeks to find a permutation $\pi$ 
that maximizes a diversity-weighted relevance objective:
$\sum\nolimits_{i=1}^{|\unordSet|} \left( \abs{\mathcal{C}(\ordSet{i})}-\abs{\mathcal{C}(\ordSet{i-1})}\right) \itempr{\order(i)}$,
where 
$\mathcal{C}(\ordSet{i})$ denotes the set of categories that $\ordSet{i}$ covers.

\para{Maximal marginal relevance (\MMR) \cite{carbonell1998use}.} 
\MMR iteratively selects the item $i \in \unordSet$ that maximizes
$\lambda\, \itempr{i} - (1-\lambda) \max\nolimits_{j \in R} (1 - d(i,j))$.

\para{Max-sum diversification (\MSD) \cite{borodin2012max}.} 
\MSD iteratively selects the item $i \in \unordSet$ that maximizes  
$ p_i + \lambda  \sum_{j \in R} d(i,j)$.

\para{Determinantal point process (\DPP) \cite{chen2018fast}.} 
\DPP iteratively selects the item $i$ that maximizes 
$\lambda \itempr_i + (1-\lambda) (\text{log det}(\mathbf{S}_{R\cup\{i\}})-\text{log det}(\mathbf{S}_{R}))$,
where $\mathbf{S}$ is a similarity matrix, with $\mathbf{S}_{ij} = 1- d(i,j)$.

\subsubsection{Our Methods} We provide more details of our methods $\bke$, $\bkeh$, \Sbke and \Sbkeh. 

\para{$\bke$.} 
We refer to the $\bke$ algorithm as \bketwo, \bkethree, and \bkefour when $\kpar$ is set to be $2$, $3$, and $4$. 
Once the first $\kpar$-sequence is chosen,
we greedily extend the remaining items that maximize the incremental gain of $\omsd$. 
Notice that \bketwo in our implementation is the same as the greedy algorithm we propose in \Cref{alg:greedy_MSD}.

\para{$\bkeh$.} 
Since \bke can be computationally expensive when \kpar is large, we adopt a greedy approach by selecting the top 100 items that maximize the incremental gain of \omsd, thereby forming a candidate set.
The best \kpar items are exclusively selected within this candidate set. 
We extend the best \kpar items in the same way as in \bke.

\edit{\para{$\Sbke$.} 
This is the \bke algorithm with $\ordSet{\kpar} = \argmax_{\ordSet{\kpar} \sqsubseteq \order{\unordSet}} \osdo(\ordSet{\kpar})$. We refer to the \Sbke algorithm as \Sbkethree and \Sbkefour when $\kpar$ is set to $3$ and $4$, respectively. Note that when $\kpar$ is set to $2$, \Sbketwo is identical to \bketwo. 

\para{$\Sbkeh$.} 
Analogous to \bkeh, this is the heuristic version of \Sbke, where we select the best $\kpar$ items exclusively from the same candidate set as described for \bke.

Throughout the experiments, we evaluate \bketwo, \bkethree, \bkefour, \Sbkethree and \Sbkefour on dataset \letor and \ltrc. For the recommendation datasets, we evaluate \bketwo and heuristics \bkethreeH, \bkefourH,\SbkethreeH and \SbkefourH on these datasets. 
}

\subsection{Results and Discussion}

\subsubsection{Performance on Sequential Sum Diversity.} 

\begin{table*}[t]
	\centering
	\caption{Sequential sum diversity $\osdo(O)$ with item continuation probability mapped to $[0.1, 0.3]$. The results marked with $^{*}$ are obtained using the \bkeh and \Sbkeh heuristic}
	\label{table:ssd_small}
	\vspace{-3mm}
	\resizebox{\linewidth}{!}{%
		\renewcommand{\arraystretch}{1} 
		\setlength{\tabcolsep}{6pt} 
		\begin{tabular}{lccccccc}
		\toprule
		& \Coat & \KuaiRec & \Netflix & \Movielens & \Yahoo & \letor & \ltrc \\
		\midrule
		\Random & 0.039 $\pm$ 0.020 & 0.022 $\pm$ 0.006 & 0.063 $\pm$ 0.026 & 0.059 $\pm$ 0.028 & 0.016 $\pm$ 0.027 & 0.008 $\pm$ 0.008 & 0.019 $\pm$ 0.011 \\
		\explore & 0.103 $\pm$ 0.041 & 0.050 $\pm$ 0.000 & 0.152 $\pm$ 0.026 & \underline{0.173 $\pm$ 0.013} & 0.133 $\pm$ 0.044 & 0.026 $\pm$ 0.021 & 0.041 $\pm$ 0.018 \\
		\DUM & 0.102 $\pm$ 0.038 & 0.033 $\pm$ 0.008 & 0.145 $\pm$ 0.026 & 0.169 $\pm$ 0.018 & \underline{0.148 $\pm$ 0.044} & \underline{0.043 $\pm$ 0.028} & \underline{0.054 $\pm$ 0.020} \\
		\MSD & 0.095 $\pm$ 0.039 & 0.042 $\pm$ 0.006 & 0.161 $\pm$ 0.024 & \textbf{0.183 $\pm$ 0.003} & \underline{0.148 $\pm$ 0.044} & 0.039 $\pm$ 0.026 & 0.051 $\pm$ 0.020 \\
		\MMR & 0.105 $\pm$ 0.041 & \underline{0.051 $\pm$ 0.010} & {0.162 $\pm$ 0.024} & \textbf{0.183 $\pm$ 0.003} & \underline{0.148 $\pm$ 0.044} & 0.040 $\pm$ 0.026 & 0.051 $\pm$ 0.019 \\
		\DPP & 0.105 $\pm$ 0.041 & 0.044 $\pm$ 0.006 & {0.162 $\pm$ 0.024} & \textbf{0.183 $\pm$ 0.003} & \underline{0.148 $\pm$ 0.044} & 0.040 $\pm$ 0.026 & 0.051 $\pm$ 0.019 \\
		\bketwo & \textbf{0.109 $\pm$ 0.042} & \textbf{0.058 $\pm$ 0.010} & \underline{0.163 $\pm$ 0.024} & \textbf{0.183 $\pm$ 0.003} & \underline{0.148 $\pm$ 0.044} & \underline{0.043 $\pm$ 0.028} & \textbf{0.055 $\pm$ 0.020} \\
		\bkethree & \underline{0.106 $\pm$ 0.038}* & \textbf{0.058 $\pm$ 0.009}* & 0.160 $\pm$ 0.025* & \textbf{0.183 $\pm$ 0.003}* & {0.140 $\pm$ 0.045}* & {0.042 $\pm$ 0.027} & \underline{0.054 $\pm$ 0.020} \\
		\bkefour & 0.104 $\pm$ 0.037* & \textbf{0.058 $\pm$ 0.009}* & 0.159* $\pm$ 0.025* & \textbf{0.183 $\pm$ 0.003}* & 0.139 $\pm$ 0.045* & {0.042 $\pm$ 0.027} & \underline{0.054 $\pm$ 0.020} \\
        \edit{\Sbkethree}   & 0.083 \jj 0.046* & \textbf{0.058 $\pm$ 0.010}* & \textbf{0.164 \jj 0.023}* & \textbf{0.183 $\pm$ 0.003}* & \textbf{0.152 \jj 0.041}* & \textbf{0.044 \jj 0.028} & \underline{0.054 \jj 0.198}\\
        \edit{\Sbkefour}   & 0.084 \jj 0.046* & \textbf{0.058 $\pm$ 0.010}* & \textbf{0.164 \jj 0.023}* & \textbf{0.183 $\pm$ 0.003}*  & \textbf{0.152 \jj 0.041}* & \textbf{0.044 \jj 0.028} & \underline{0.054 \jj 0.198} \\
		\bottomrule
		\end{tabular}
	}
\end{table*}

\begin{table*}[t]
  \centering
    \caption{Sequential sum diversity $\osdo(O)$ with item continuation probability mapped to $[0.4, 0.6]$. The results marked with $^{*}$ are obtained using the \bkeh and \Sbkeh heuristic.}
  \label{table:medium}
  \label{table:ssd_medium}
  \vspace{-3mm}
  \resizebox{\linewidth}{!}{%
    \renewcommand{\arraystretch}{1} 
    \setlength{\tabcolsep}{4pt} 
    \begin{tabular}{lcccccccc}
      \toprule
{\small Methods} & \Coat & \KuaiRec & \Netflix & \Movielens & \Yahoo & \letor & \ltrc \\
\midrule
\Random & 0.646 $\pm$ 0.177 & 0.383 $\pm$ 0.066 & 0.921 $\pm$ 0.220 & 0.878 $\pm$ 0.246 & 0.231 $\pm$ 0.248 & 0.208 $\pm$ 0.106 & 0.359 $\pm$ 0.116 \\
\explore & 1.276 $\pm$ 0.374 & 0.681 $\pm$ 0.082 & 1.861 $\pm$ 0.263 & 2.142 $\pm$ 0.114 & 1.684 $\pm$ 0.465 & 0.440 $\pm$ 0.190 & 0.597 $\pm$ 0.151 \\
\DUM & 1.214 $\pm$ 0.348 & 0.411 $\pm$ 0.086 & 1.779 $\pm$ 0.241 & 2.080 $\pm$ 0.113 & {1.793 $\pm$ 0.466} & 0.564 $\pm$ 0.184 & 0.640 $\pm$ 0.134 \\
\MSD & 1.211 $\pm$ 0.349 & 0.704 $\pm$ 0.057 & 1.888 $\pm$ 0.254 & 2.228 $\pm$ 0.046 & {1.793 $\pm$ 0.466} & 0.577 $\pm$ 0.191 & 0.662 $\pm$ 0.146 \\
\MMR & 1.282 $\pm$ 0.377 & 0.725 $\pm$ 0.080 & 1.937 $\pm$ 0.258 & 2.232 $\pm$ 0.044 & {1.793 $\pm$ 0.466} & 0.570 $\pm$ 0.191 & 0.644 $\pm$ 0.145 \\
\DPP & \underline{1.284 $\pm$ 0.377} & 0.711 $\pm$ 0.073 & 1.933 $\pm$ 0.258 & 2.230 $\pm$ 0.046 & {1.793 $\pm$ 0.466} & 0.557 $\pm$ 0.181 & 0.642 $\pm$ 0.138 \\
\bketwo & \textbf{1.289 $\pm$ 0.376} & {0.778 $\pm$ 0.075} & {1.946 $\pm$ 0.264} & \textbf{2.235 $\pm$ 0.043} & \underline{1.795 $\pm$ 0.466} & {0.592 $\pm$ 0.195} & {0.670 $\pm$ 0.144} \\
\bkethree & 1.184 $\pm$ 0.368* & 0.776 $\pm$ 0.072* & 1.919 $\pm$ 0.270* & \underline{2.234 $\pm$ 0.046}* & 1.729 $\pm$ 0.495* & 0.582 $\pm$ 0.197 & 0.667 $\pm$ 0.145 \\
\bkefour & 1.184 $\pm$ 0.367* & \underline{0.782 $\pm$ 0.073}* & 1.910 $\pm$ 0.272* & 2.233 $\pm$ 0.050* & 1.715 $\pm$ 0.496* & {0.588 $\pm$ 0.197} & {0.668 $\pm$ 0.145} \\
\edit{\Sbkethree} &1.076 $\pm$ 0.420* & \underline{0.782 $\pm$ 0.078}* & \underline{1.948 $\pm$ 0.247}* & 2.228 $\pm$ 0.055* & \textbf{1.834 $\pm$ 0.448}* & \underline{0.597 $\pm$ 0.196} & \underline{0.673 $\pm$ 0.143}\\ 
\edit{\Sbkefour} &1.076 $\pm$ 0.420* & \textbf{0.786 $\pm$ 0.077}* & \textbf{1.951 $\pm$ 0.249}* & 2.227 $\pm$ 0.054* & \textbf{1.834 $\pm$ 0.448}* & \textbf{0.600 $\pm$ 0.199} & \textbf{0.677 $\pm$ 0.144} \\
\bottomrule
    \end{tabular}
  }
\end{table*}

\edit{
We conduct experiments on all datasets evaluating our methods and the baselines. 
For each dataset and method, we obtain a sequence $\ordSet$ of all items 
for each user and calculate the sequential sum diversity $\osdo(\ordSet)$.
We report the average $\osdo(\ordSet)$ values and standard deviations across all users for our proposed algorithms (\bke, \bkeh, \Sbke, \Sbkeh for $k = \{2,3,4\}$) and other baselines as shown in \Cref{table:ssd_small} to \Cref{table:ssd_full}. 


\paragraph{Low Engagement Settings}
We begin by analyzing performance under low engagement conditions, where continuation probabilities are mapped to the interval $[0.1, 0.3]$ (Table~\ref{table:ssd_small}). In this regime, users are more likely to terminate interaction early, making diversity among the top-ranked items particularly important. Our proposed algorithms consistently outperform all baseline, with the proposed \bketwo (which is equivalent to \Sbketwo) achieves the best performance on four of the seven datasets (\Coat, \KuaiRec, \Movielens, and \ltrc). On the remaining datasets, \Sbkethree and \Sbkefour attain the highest \osdo scores on \Netflix and \Yahoo, respectively.

\paragraph{Medium Engagement Settings}
Under medium engagement levels, with continuation probabilities in the range $[0.4, 0.6]$ (Table~\ref{table:ssd_medium}), the proposed methods continue to perform favorably, achieving the best or second best results on all datasets. \Sbkefour achieves the highest scores on five datasets (\KuaiRec, \Netflix, \Yahoo, \letor, and \ltrc), while \bketwo attains the best performance on \Coat and \Movielens. Overall, the \Sbke family exhibits stronger performance than the \bke family in this setting, with \Sbkethree and \Sbkefour consistently matching or outperforming \bketwo across most datasets.

\paragraph{High Engagement Settings}
We next consider high engagement settings, where continuation probabilities fall within $[0.7, 0.9]$ (Table~\ref{table:ssd_large}). In this regime, users are more likely to examine longer ranking prefixes, and diversity accumulated over extended sequences becomes increasingly important. Baseline methods such as \DPP and \MMR exhibit better performance on several datasets. Nonetheless, the proposed methods remain competitive: \bkefour achieves the best result on \Yahoo, and the \Sbke variants obtain the second-best performance on the same dataset. On \Movielens, \bkefour achieves a score of $71.317$, which is comparable to the best-performing \DPP result ($71.362$).

\paragraph{Performance Across the Full Probability Range}
We further evaluate performance when continuation probabilities span the full range $[0.1, 0.9]$ (Table~\ref{table:ssd_full}), capturing heterogeneous user engagement behaviors. In this setting, our proposed algorithms achieve competitive results when optimization focuses only on the top $3$ or $4$ ranked items. \bkefour achieves the highest performance on \Movielens, while \Sbkefour obtains the best diversity scores on \Yahoo and \letor. Across datasets, the proposed methods consistently rank among the top-performing approaches. While \MMR and \DPP demonstrate strong performance on certain datasets (\Coat, \KuaiRec, and \Netflix), reflecting their effectiveness in accumulating diversity over longer sequences. 

\paragraph{Comparison of \bke and \Sbke Variants}
We next examine the effect of the parameter $\tau$, which controls the number of top-ranked items selected during ranking construction. The \bke variants perform better under low engagement conditions, which matches our privous analysis that the approximation ratio decreases as $a$ and $b$ both increase. Increasing $\tau$ to $3$ or $4$ does not lead to consistent improvements in this regime. The possible reason for this is twofold. First, for $\tau=3$ and $4$, the \bke family optimizes an intermediate \omshp objective, which may introduce approximation errors relative to the target objective \osdo. Second, under low engagement conditions, users quit after checking two items with high probability, so optimizing more items does not improve the user experience but sacrifices diversity in the first two items.
In contrast, under high engagement and the full probability range, \bkethree and \bkefour exhibit modest improvements over \bketwo on several datasets, suggesting that selecting a larger set of top-ranked items can be beneficial when users explore deeper into the ranking.
The \Sbke family, on the contrary, shows consistent improvements when $\tau$ increases in all user engagement settings.

\paragraph{Efficiency and Overall Performance}
A notable property of the proposed methods is their stable performance across all engagement regimes. In contrast to baselines such as \DPP and \MMR, which tend to perform best under high engagement conditions, the proposed algorithms maintain competitive performance across low, medium, and high continuation probabilities.
Finally, results marked with an asterisk ($^{*}$) correspond to heuristic variants (\bkeh and \Sbkeh) of \bkethree, \bkefour, \Sbkethree, and \Sbkefour, which are introduced to improve computational efficiency. Despite relying on approximations, these heuristic methods achieve performance comparable to, and in some cases exceeding, that of the exact \bketwo. This indicates that the proposed heuristics scale effectively while preserving solution quality, making them suitable for large-scale recommendation and ranking applications.
}

\begin{table*}[t]
	\caption{Sequential sum diversity $\osdo(O)$ with item continuation probability mapped to $[0.7, 0.9]$. The results marked with $^{*}$ are obtained using the \bkeh and \Sbkeh heuristic}
	\label{table:ssd_large}
	\vspace{-3mm}
	\resizebox{\linewidth}{!}{%
		\renewcommand{\arraystretch}{1} 
		\setlength{\tabcolsep}{6pt} 
		\begin{tabular}{lccccccc}
			\toprule
			& \Coat & \KuaiRec & \Netflix & \Movielens & \Yahoo & \letor & \ltrc \\
			\midrule
			\Random & 9.703 $\pm$ 3.704 & 5.737 $\pm$ 1.094 & 16.901 $\pm$ 5.588 & 14.677 $\pm$ 4.644 & 4.014 $\pm$ 3.259 & 2.784 $\pm$ 1.638 & 4.757 $\pm$ 1.339 \\
			\explore & 23.044 $\pm$ 12.390 & 10.668 $\pm$ 1.874 & 47.316 $\pm$ 14.231 & 65.410 $\pm$ 6.715 & 40.505 $\pm$ 19.773 & 5.825 $\pm$ 3.450 & 7.026 $\pm$ 2.161 \\
			\DUM & 22.251 $\pm$ 11.328 & 11.599 $\pm$ 3.289 & 39.696 $\pm$ 9.856 & 48.412 $\pm$ 4.246 & 38.251 $\pm$ 18.866 & 5.317 $\pm$ 2.755 & 6.332 $\pm$ 1.691 \\
			\MSD & 18.586 $\pm$ 8.481 & 9.119 $\pm$ 1.090 & 35.729 $\pm$ 10.010 & 63.242 $\pm$ 8.458 & 38.520 $\pm$ 19.221 & 6.140 $\pm$ 2.825 & 7.088 $\pm$ 1.958 \\
			\MMR & \underline{24.151 $\pm$ 13.421} & \textbf{14.654 $\pm$ 3.874} & \underline{48.595 $\pm$ 14.424} & 69.341 $\pm$ 7.604 & 39.134 $\pm$ 19.191 & \textbf{6.523 $\pm$ 3.803} & \textbf{7.385 $\pm$ 2.352} \\
			\DPP & \textbf{24.213 $\pm$ 13.748} & \underline{12.565 $\pm$ 4.258} & \textbf{50.530 $\pm$ 15.407} & \textbf{71.362 $\pm$ 7.006} & 38.377 $\pm$ 19.021 & 6.025 $\pm$ 3.637 & \underline{7.204 $\pm$ 2.164} \\
			\bketwo & 23.030 $\pm$ 12.385 & 10.920 $\pm$ 1.728 & 47.224 $\pm$ 14.648 & 71.303 $\pm$ 7.748 & 41.677 $\pm$ 21.090 & 6.034 $\pm$ 3.006 & 6.972 $\pm$ 1.971 \\
			\bkethree & 22.839 $\pm$ 12.428* & 10.906 $\pm$ 1.714* & 47.053 $\pm$ 14.294* & 71.305 $\pm$ 7.738* & {41.924 $\pm$ 20.996}* & 6.086 $\pm$ 3.033 & 7.005 $\pm$ 1.990 \\
			\bkefour & 22.826 $\pm$ 12.415* & 10.968 $\pm$ 1.726* & 47.172 $\pm$ 14.266* & \underline{71.317 $\pm$ 7.703}* & \underline{41.969 $\pm$ 21.010}* & \underline{6.165 $\pm$ 3.048} & 7.058 $\pm$ 2.012 \\
            \edit{\Sbkethree} & 16.809 \jj 9.992* & 10.923 \jj 1.805* & 47.180 \jj 13.870* &  70.646 \jj 8.490* & \textbf{42.925 \jj 21.147}* & 6.072 \jj 3.010 & 7.007 \jj 1.980 \\
            \edit{\Sbkefour}  & 16.821 \jj 10.024* & 10.970 \jj 1.822* &  47.208 \jj 13.890* & 70.650 \jj 8.482* & \textbf{42.925 \jj 21.147}* & 6.123 \jj 3.024 & 7.060 \jj 1.990 \\
			\bottomrule
		\end{tabular}
	}
\end{table*}

\begin{table*}[t]
\centering
\caption{Sequential sum diversity $\osdo(O)$ with item continuation probabilities mapped to $[0.1, 0.9]$. The results marked with $^{*}$ are obtained using the \bkeh and \Sbkeh heuristic}
\label{table:ssd_full}
	\vspace{-3mm}
\resizebox{\linewidth}{!}{%
	\renewcommand{\arraystretch}{1} 
	\setlength{\tabcolsep}{6pt} 
	\begin{tabular}{lccccccc}
		\toprule
		& \Coat & \KuaiRec & \Netflix & \Movielens & \Yahoo & \letor & \ltrc \\
		\midrule
		\Random & 0.555 $\pm$ 0.722 & 0.255 $\pm$ 0.191 & 1.704 $\pm$ 1.822 & 1.157 $\pm$ 1.283 & 0.394 $\pm$ 0.927 & 0.127 $\pm$ 0.137 & 0.189 $\pm$ 0.216 \\
		\explore & 7.665 $\pm$ 10.628 & 1.140 $\pm$ 0.554 & 27.252 $\pm$ 19.761 & 61.736 $\pm$ 9.103 & 26.104 $\pm$ 24.181 & 0.436 $\pm$ 0.419 & 0.702 $\pm$ 0.640 \\
		\DUM & 8.409 $\pm$ 10.708 & 4.808 $\pm$ 3.282 & 22.197 $\pm$ 12.660 & 37.570 $\pm$ 5.342 & 21.018 $\pm$ 20.291 & 0.583 $\pm$ 0.478 & 1.050 $\pm$ 0.865 \\
		\MSD & \underline{9.660 $\pm$ 13.374} & 6.388 $\pm$ 4.981 & 32.270 $\pm$ 22.224 & 59.667 $\pm$ 13.080 & 22.055 $\pm$ 20.980 & 0.595 $\pm$ 0.491 & 1.096 $\pm$ 0.914 \\
		\MMR & \textbf{9.661 $\pm$ 13.608} & \textbf{6.768 $\pm$ 5.221} & \underline{32.703 $\pm$ 22.895} & 68.657 $\pm$ 9.789 & 26.303 $\pm$ 23.451 & {0.612 $\pm$ 0.521} & \textbf{1.136 $\pm$ 0.996} \\
		\DPP & 9.559 $\pm$ 13.534 & \underline{6.739 $\pm$ 5.140} & \textbf{33.157 $\pm$ 23.552} & 66.995 $\pm$ 13.290 & 21.039 $\pm$ 20.324 & 0.592 $\pm$ 0.510 & \textbf{1.136 $\pm$ 1.000} \\
		\bketwo & 9.075 $\pm$ 12.593 & 3.726 $\pm$ 2.575 & 30.790 $\pm$ 21.842 & 69.476 $\pm$ 11.786 & 27.567 $\pm$ 26.163 & 0.610 $\pm$ 0.515 & 1.091 $\pm$ 0.899 \\
		\bkethree & 9.119 $\pm$ 12.677* & 3.745 $\pm$ 2.570* & 30.848 $\pm$ 21.822* & \underline{69.478 $\pm$ 11.780}* & {27.709 $\pm$ 26.100}* & 0.608 $\pm$ 0.516 & 1.099 $\pm$ 0.911 \\
		\bkefour & 9.100 $\pm$ 12.686* & 3.834 $\pm$ 2.617* & 30.915 $\pm$ 21.813* & \textbf{69.484 $\pm$ 11.759}* & {27.746 $\pm$ 26.105}* & {0.615 $\pm$ 0.520} & \underline{1.112 $\pm$ 0.928} \\
        \edit{\Sbkethree}  &  5.621 \jj 9.873* & 3.714 \jj 2.644*  & 30.559 \jj 21.201 * & 68.445 \jj 13.339* & \underline{29.296 \jj 26.436}*  & \underline{0.618 \jj 0.516}  & 1.095 \jj 0.898 \\
        \edit{\Sbkefour}  & 5.634 \jj 8.737* & 3.796 \jj 2.668* & 30.586 \jj 21.202 *  & 68.448 \jj 13.333* & \textbf{29.312 \jj 26.421}*   & \textbf{0.625 \jj 0.523}  & 1.106 \jj 0.908 \\
		\bottomrule
	\end{tabular}
}
\end{table*}

\subsubsection{Performance on \expdcg and \expseren}

\begin{table*}[t]
	\centering
		\caption{\expnum, \expdcg and \expseren values with item continuation probabilities mapped to $[0.1, 0.3]$. The results with $^{*}$ are obtained using the \bkeh and \Sbkeh heuristic}
	\label{table:4metric_small}
	\vspace{-3mm}
  \resizebox{\textwidth}{!}{%
    \renewcommand{\arraystretch}{1} 
    \setlength{\tabcolsep}{2pt} 
    \begin{tabular}{lllccccccc}
    \toprule
      {\small Metrics} & {\small Methods} & \Coat & \KuaiRec & \Netflix & \Movielens & \Yahoo & \letor & \ltrc \\
    \midrule  
            \multirow{9}{*}{\expdcg}
			& \Random & 0.040 $\pm$ 0.020 & 0.040 $\pm$ 0.012 & 0.054 $\pm$ 0.020 & 0.054 $\pm$ 0.025 & 0.044 $\pm$ 0.026 & 0.022 $\pm$ 0.022 & 0.033 $\pm$ 0.019 \\
			& \explore & \underline{0.077 $\pm$ 0.027} & 0.080 $\pm$ 0.017 & 0.095 $\pm$ 0.017 & \underline{0.107 $\pm$ 0.010} & 0.090 $\pm$ 0.027 & 0.030 $\pm$ 0.028 & 0.042 $\pm$ 0.024 \\
			& \DUM & \textbf{0.082 $\pm$ 0.026} & \textbf{0.101 $\pm$ 0.013} & \textbf{0.105 $\pm$ 0.012} & \textbf{0.112 $\pm$ 0.001} & \textbf{0.099 $\pm$ 0.022} & 0.050 $\pm$ 0.038 & 0.056 $\pm$ 0.024 \\
			& \MSD & 0.078 $\pm$ 0.025 & 0.085 $\pm$ 0.010 & \textbf{0.105 $\pm$ 0.012} & \textbf{0.112 $\pm$ 0.001} & \textbf{0.099 $\pm$ 0.022} & \underline{0.065 $\pm$ 0.036} & \textbf{0.070 $\pm$ 0.026} \\
			& \MMR & \textbf{0.082 $\pm$ 0.026} & \underline{0.091 $\pm$ 0.014} & \textbf{0.105 $\pm$ 0.012} & \textbf{0.112 $\pm$ 0.001} & \textbf{0.099 $\pm$ 0.022} & \underline{0.065 $\pm$ 0.036} & \textbf{0.070 $\pm$ 0.026} \\
			& \DPP & \textbf{0.082 $\pm$ 0.025} & 0.085 $\pm$ 0.010 & \textbf{0.105 $\pm$ 0.012} & \textbf{0.112 $\pm$ 0.001} & \textbf{0.099 $\pm$ 0.022} & \textbf{0.066 $\pm$ 0.036} & \underline{0.069 $\pm$ 0.025} \\
			& \bketwo & 0.075 $\pm$ 0.026 & 0.070 $\pm$ 0.023 & 0.103 $\pm$ 0.012 & \textbf{0.112 $\pm$ 0.001} & 0.095 $\pm$ 0.025 & 0.050 $\pm$ 0.038 & 0.056 $\pm$ 0.024 \\
			& \bkethree & {0.076 $\pm$ 0.026}* & 0.090 $\pm$ 0.011* & \underline{0.104 $\pm$ 0.012}* & \textbf{0.112 $\pm$ 0.001}* & \textbf{0.099 $\pm$ 0.022}* & 0.063 $\pm$ 0.035 & 0.064 $\pm$ 0.025 \\
            & \bkefour & {0.076 \jj 0.036}* & \underline{0.091 \jj 0.011}* & \underline{0.104 \jj 0.012}* & \textbf{0.112 \jj 0.001}* & \textbf{0.099 \jj 0.022}* & 0.063 \jj 0.035 & 0.064 \jj 0.025 \\
			& \Sbkethree & 0.062 \jj 0.031* & 0.076 \jj 0.022* & \underline{0.104 \jj 0.012}* & \textbf{0.112 \jj 0.002}* & \underline{0.098 \jj 0.023}* & 0.052 \jj 0.038 & 0.058 \jj 0.025 \\
            & \Sbkefour & 0.062 \jj 0.031* & 0.077 \jj 0.021* & \underline{0.104 \jj 0.012}* & \textbf{0.112 \jj 0.002}* & \underline{0.098 \jj 0.023}* & 0.053 \jj 0.038 & 0.058 \jj 0.025 \\

        \midrule
            \multirow{9}{*}{\expseren}
			& \Random & 0.013 $\pm$ 0.020 & -- & 0.005 $\pm$ 0.015 & 0.006 $\pm$ 0.018 & 0.003 $\pm$ 0.012 & -- & -- \\
			& \explore & 0.029 $\pm$ 0.035 & -- & 0.014 $\pm$ 0.027 & \underline{0.018 $\pm$ 0.029} & 0.020 $\pm$ 0.021 & -- & -- \\
			& \DUM & 0.022 $\pm$ 0.029 & -- & \textbf{0.019 $\pm$ 0.033} & \textbf{0.019 $\pm$ 0.032} & \underline{0.021 $\pm$ 0.017} & -- & -- \\
			& \MSD & 0.017 $\pm$ 0.025 & -- & \textbf{0.019 $\pm$ 0.032} & 0.013 $\pm$ 0.026 & 0.019 $\pm$ 0.015 & -- & -- \\
			& \MMR & 0.019 $\pm$ 0.026 & -- & \textbf{0.019 $\pm$ 0.032} & 0.013 $\pm$ 0.026 & 0.019 $\pm$ 0.015 & -- & -- \\
			& \DPP & 0.019 $\pm$ 0.026 & -- & \textbf{0.019 $\pm$ 0.032} & 0.013 $\pm$ 0.026 & 0.019 $\pm$ 0.015 & -- & -- \\
			& \bketwo & \underline{0.037 $\pm$ 0.038} & -- & \underline{0.018 $\pm$ 0.031} & 0.013 $\pm$ 0.026 & \textbf{0.027 $\pm$ 0.023} & -- & -- \\
			& \bkethree & \textbf{0.041 $\pm$ 0.041}* & -- & 0.017 $\pm$ 0.030* & 0.013 $\pm$ 0.026* & 0.017 $\pm$ 0.014* & -- & -- \\
			& \bkefour & \textbf{0.041 $\pm$ 0.041}* & -- & 0.017 $\pm$ 0.030* & 0.013 $\pm$ 0.026* & 0.017 $\pm$ 0.014* & -- & -- \\
            & \edit{\Sbkethree} & 0.025 \jj 0.035* & -- & \underline{0.018 \jj 0.033}* & 0.013 \jj 0.027* & \textbf{0.027 \jj 0.022}* & -- & -- \\
            & \edit{\Sbkefour} & 0.025 \jj 0.035* & -- & \underline{0.018 \jj 0.033}* & 0.013 \jj 0.027* & \textbf{0.027 \jj 0.022}* & -- & -- \\
		\bottomrule
		\end{tabular}
	}
\end{table*}

\begin{table*}[t]
  \centering
    \caption{\expdcg and \expseren values with item continuation probabilities mapped to $[0.4, 0.6]$. The results marked with $^{*}$ are obtained using the \bkeh and \Sbkeh heuristic}
  \label{table:4metric_medium}
  \vspace{-3mm}
  \resizebox{\textwidth}{!}{%
    \renewcommand{\arraystretch}{1} 
    \setlength{\tabcolsep}{2pt} 
    \begin{tabular}{llccccccccccc}
    \toprule
      {\small Metrics} & {\small Methods} & \Coat & \KuaiRec & \Netflix & \Movielens & \Yahoo & \letor & \ltrc \\
    \midrule  
    \multirow{9}{*}{\expdcg}
    & \Random & 0.354 ± 0.070 & 0.347 ± 0.049 & 0.413 ± 0.065 & 0.394 ± 0.080 & 0.360 ± 0.094 & 0.274 ± 0.082 & 0.323 ± 0.068 \\
    & \explore & 0.496 ± 0.100 & 0.483 ± 0.048 & 0.575 ± 0.055 & \underline{0.620 ± 0.026} & 0.550 ± 0.100 & 0.330 ± 0.104 & 0.373 ± 0.082 \\
    & \DUM & \textbf{0.513 ± 0.096} & \textbf{0.590 ± 0.048} & \textbf{0.604 ± 0.044} & \textbf{0.635 ± 0.005} & {0.572 ± 0.086} & 0.327 ± 0.104 & 0.378 ± 0.078 \\
    & \MSD & 0.489 ± 0.091 & 0.459 ± 0.028 & 0.596 ± 0.044 & \textbf{0.635 ± 0.005} & {0.572 ± 0.086} & \textbf{0.416 ± 0.123} & \underline{0.443 ± 0.084} \\
    & \MMR & \underline{0.511 ± 0.095} & 0.465 ± 0.030 & \underline{0.600 ± 0.044} & \textbf{0.635 ± 0.005} & {0.572 ± 0.086} & 0.400 ± 0.115 & \textbf{0.447 ± 0.088} \\
    & \DPP & 0.510 ± 0.095 & 0.472 ± 0.032 & \underline{0.600 ± 0.044} & \textbf{0.635 ± 0.005} & {0.572 ± 0.086} & \underline{0.406 ± 0.117} & 0.440 ± 0.081 \\
    & \bketwo & 0.494 ± 0.096 & 0.409 ± 0.072 & 0.597 ± 0.046 & \textbf{0.635 ± 0.005} & 0.565 ± 0.092 & 0.326 ± 0.104 & 0.378 ± 0.078 \\
    & \bkethree & 0.499 ± 0.096* & 0.477 ± 0.031* & \underline{0.600 ± 0.044}* & \textbf{0.635 ± 0.005}* & \textbf{0.577 ± 0.085}* & 0.390 ± 0.108 & 0.402 ± 0.080 \\
    & \bkefour & 0.499 ± 0.096* & \underline{0.480 ± 0.032}* & \underline{0.600 ± 0.044}* & \textbf{0.635 ± 0.005}* & \textbf{0.577 ± 0.085}* & 0.393 ± 0.110 & 0.407 ± 0.080 \\
    & \edit{\Sbkethree} & 0.438 $\pm$ 0.120* & 0.395 $\pm$ 0.070* & \underline{0.600 $\pm$ 0.043}* &  \textbf{0.635 $\pm$ 0.006}*  & \underline{0.575 $\pm$ 0.086}* & 0.346 $\pm$ 0.117 & 0.393 $\pm$ 0.079 \\
    & \edit{\Sbkefour} & 0.438 $\pm$ 0.119* & 0.397 $\pm$ 0.071* & \underline{0.600 $\pm$ 0.043}* &  \textbf{0.635 $\pm$ 0.006}*  & \underline{0.575 $\pm$ 0.086}* & 0.352 $\pm$ 0.120 & 0.401 $\pm$ 0.081 \\
    
    \midrule
    \multirow{9}{*}{\expseren}
    & \Random & 0.141 ± 0.138 & -- & 0.028 ± 0.068 & 0.043 ± 0.091 & 0.040 ± 0.079 & -- & -- \\
    & \explore & 0.247 ± 0.186 & -- & 0.132 ± 0.151 & \textbf{0.141 ± 0.153} & 0.275 ± 0.145 & -- & -- \\
    & \DUM & 0.204 ± 0.159 & -- & 0.133 ± 0.171 & \underline{0.136 ± 0.159} & {0.309 ± 0.140} & -- & -- \\
    & \MSD & 0.225 ± 0.147 & -- & \textbf{0.152 ± 0.164} & 0.115 ± 0.144 & 0.303 ± 0.136 & -- & -- \\
    & \MMR & 0.214 ± 0.157 & -- & 0.148 ± 0.164 & 0.115 ± 0.144 & 0.303 ± 0.136 & -- & -- \\
    & \DPP & 0.218 ± 0.160 & -- & \underline{0.149 ± 0.163} & 0.115 ± 0.144 & 0.303 ± 0.136 & -- & -- \\
    & \bketwo & 0.288 ± 0.202 & -- & 0.143 ± 0.161 & 0.115 ± 0.144 & \underline{0.325 ± 0.143} & -- & -- \\
    & \bkethree & \textbf{0.354 ± 0.238}* & -- & 0.139 ± 0.162* & 0.115 ± 0.144* & 0.264 ± 0.143* & -- & -- \\
    & \bkefour & \underline{0.351 ± 0.240}* & -- & 0.139 ± 0.165* & 0.115 ± 0.144* & 0.259 ± 0.145* & -- & -- \\
    & \edit{\Sbkethree} & 0.200 $\pm$ 0.198* & -- & 0.139 $\pm$ 0.169* & 0.114 $\pm$ 0.144* & \textbf{0.333 $\pm$ 0.135}*  &-- & --\\
    & \edit{\Sbkefour} & 0.200 $\pm$ 0.198* & -- & 0.136 $\pm$ 0.166* & 0.115 $\pm$ 0.144* & \textbf{0.333 $\pm$ 0.135}* &-- &  --\\
    \bottomrule
    \end{tabular}
  }
\end{table*}

\edit{
We further evaluate all methods using \expdcg and \expseren and report the results in Tables~\ref{table:4metric_small} to \ref{table:4metric_full}.

\paragraph{Expected DCG Performance}
Baselines such as \DUM, \MSD, \MMR, and \DPP achieve strong \expdcg scores, as they explicitly incorporate relevance scores, particularly \DUM, which prioritizes selecting the most relevant items when constructing rankings. As expected, our proposed algorithms do not achieve the highest scores, since they do not explicitly optimize for relevance. Nevertheless, they achieve comparable performance to these baselines. \bkethree and \bkefour match the best-performing methods on \Movielens across all settings and achieve top results on \Yahoo under medium and high engagement regimes. The \Sbke variants similarly perform well on \Netflix and \Yahoo, demonstrating that optimizing for sequential sum diversity does not substantially compromise relevance.

\paragraph{Expected Serendipity Performance}
Our algorithms demonstrate particular strength in serendipity. The \bke and \Sbke families achieve the best or second-best results across all regimes and datasets, with only three exceptions: \Movielens in the small and medium probability regimes, and \Netflix in the medium regime. This strong serendipity performance stems from our sequential sum diversity objective, which naturally surfaces unexpected yet relevant items. These results suggest that our algorithms have greater potential for generating exploratory recommendations.
}

\begin{table*}[t]
	\centering
		\caption{\expnum, \expdcg and \expseren values with item continuation probabilities mapped to $[0.7, 0.9]$. The results with $^{*}$ are obtained using the \bkeh and \Sbkeh heuristic}
	\label{table:4metric_large}
	\vspace{-3mm}
  \resizebox{\textwidth}{!}{%
    \renewcommand{\arraystretch}{1} 
    \setlength{\tabcolsep}{2pt} 
    \begin{tabular}{lllccccccc}
    \toprule
      {\small Metrics} & {\small Methods} & \Coat & \KuaiRec & \Netflix & \Movielens & \Yahoo & \letor & \ltrc \\
    \midrule  
            \multirow{9}{*}{\expdcg}
			& \Random & 1.564 $\pm$ 0.275 & 1.565 $\pm$ 0.162 & 1.889 $\pm$ 0.280 & 1.782 $\pm$ 0.287 & 1.633 $\pm$ 0.405 & 1.228 $\pm$ 0.281 & 1.438 $\pm$ 0.235 \\
			& \explore & 2.267 $\pm$ 0.534 & 2.034 $\pm$ 0.174 & 2.881 $\pm$ 0.365 & 3.280 $\pm$ 0.101 & \textbf{2.704 $\pm$ 0.605} & 1.578 $\pm$ 0.443 & 1.659 $\pm$ 0.278 \\
			& \DUM & \underline{2.324 $\pm$ 0.532} & \underline{2.789 $\pm$ 0.308} & 2.880 $\pm$ 0.308 & 3.141 $\pm$ 0.069 & 2.636 $\pm$ 0.556 & 1.402 $\pm$ 0.355 & 1.562 $\pm$ 0.253 \\
			& \MSD & 2.101 $\pm$ 0.437 & 1.772 $\pm$ 0.104 & 2.681 $\pm$ 0.290 & 3.213 $\pm$ 0.109 & 2.637 $\pm$ 0.556 & 1.617 $\pm$ 0.411 & 1.744 $\pm$ 0.281 \\
			& \MMR & 2.308 $\pm$ 0.540 & 2.416 $\pm$ 0.443 & \textbf{3.003 $\pm$ 0.350} & 3.282 $\pm$ 0.097 & 2.644 $\pm$ 0.557 & \underline{1.748 $\pm$ 0.564} & \textbf{1.833 $\pm$ 0.341} \\
			& \DPP & \textbf{2.346 $\pm$ 0.555} & \textbf{2.905 $\pm$ 0.350} & \underline{2.993 $\pm$ 0.349} & \textbf{3.327 $\pm$ 0.075} & 2.637 $\pm$ 0.556 & \textbf{1.834 $\pm$ 0.575} & \underline{1.814 $\pm$ 0.310} \\
			& \bketwo & 2.219 $\pm$ 0.508 & 1.772 $\pm$ 0.145 & 2.841 $\pm$ 0.338 & 3.286 $\pm$ 0.097 & 2.653 $\pm$ 0.578 & 1.400 $\pm$ 0.341 & 1.567 $\pm$ 0.253 \\
			& \bkethree & 2.233 $\pm$ 0.512* & 1.887 $\pm$ 0.130* & 2.855 $\pm$ 0.332* & 3.287 $\pm$ 0.096* & 2.698 $\pm$ 0.553* & 1.523 $\pm$ 0.365 & 1.615 $\pm$ 0.259 \\
			& \bkefour & 2.232 $\pm$ 0.511* & 1.891 $\pm$ 0.131* & 2.865 $\pm$ 0.330* & \underline{3.288 $\pm$ 0.093}* & \underline{2.703 $\pm$ 0.553}* & 1.545 $\pm$ 0.373 & 1.637 $\pm$ 0.263 \\
            & \Sbkethree & 1.946 \jj 0.561* & 1.879 \jj 0.189* & 2.852 \jj 0.324* & 3.275 \jj 0.113* & 2.700 \jj 0.559* & 1.463 \jj 0.382* & 1.618 \jj 0.261* \\
            & \Sbkefour & 1.946 \jj 0.562* & 1.869 \jj 0.218* & 2.854 \jj 0.325* & 3.275 \jj 0.113* & 2.700 \jj 0.559* & 1.509 \jj 0.405* & 1.659 \jj 0.266* \\
        \midrule
            \multirow{9}{*}{\expseren}
			& \Random & 0.924 $\pm$ 0.570 & -- & 0.280 $\pm$ 0.386 & 0.289 $\pm$ 0.422 & 0.229 $\pm$ 0.300 & -- & -- \\
			& \explore & 1.766 $\pm$ 1.091 & -- & 0.965 $\pm$ 0.892 & \textbf{1.253 $\pm$ 1.078} & 2.754 $\pm$ 1.380 & -- & -- \\
			& \DUM & 1.583 $\pm$ 0.895 & -- & 1.187 $\pm$ 0.900 & 1.073 $\pm$ 0.984 & {3.858 $\pm$ 1.547} & -- & -- \\
			& \MSD & 1.862 $\pm$ 0.843 & -- & \textbf{1.336 $\pm$ 0.852} & 1.152 $\pm$ 1.053 & 3.857 $\pm$ 1.546 & -- & -- \\
			& \MMR & 1.585 $\pm$ 0.895 & -- & 0.830 $\pm$ 0.848 & 1.201 $\pm$ 1.151 & 3.833 $\pm$ 1.537 & -- & -- \\
			& \DPP & 1.508 $\pm$ 0.898 & -- & 1.042 $\pm$ 0.952 & 1.183 $\pm$ 1.146 & 3.852 $\pm$ 1.540 & -- & -- \\
			& \bketwo & 1.879 $\pm$ 0.993 & -- & \underline{1.290 $\pm$ 0.936} & 1.209 $\pm$ 1.097 & \underline{3.859 $\pm$ 1.563} & -- & -- \\
			& \bkethree & \textbf{2.054 $\pm$ 1.122}* & -- & 1.287 $\pm$ 0.942* & 1.209 $\pm$ 1.097* & 3.675 $\pm$ 1.659* & -- & -- \\
			& \bkefour & \underline{2.040 $\pm$ 1.117}* & -- & 1.285 $\pm$ 0.949* & \underline{1.210 $\pm$ 1.097}* & 3.643 $\pm$ 1.668* & -- & -- \\
		    & \Sbkethree & 1.398 \jj 0.917* & -- & 1.289 \jj 0.954* & 1.208 \jj 1.092* & \textbf{3.926 \jj 1.549}* & -- & -- \\
            & \Sbkefour & 1.399 \jj 0.915* & -- & 1.275 \jj 0.941* & 1.208 \jj 1.092* & \textbf{3.926 \jj 1.549}* & -- & -- \\
        \bottomrule
		\end{tabular}
	}
\end{table*}

\begin{table*}[t]
	\centering
		\caption{\expnum, \expdcg and \expseren values with item continuation probabilities mapped to $[0.1, 0.9]$. The results with $^{*}$ are obtained using the \bkeh and \Sbkeh heuristic}
	\label{table:4metric_full}
	\vspace{-3mm}
  \resizebox{\textwidth}{!}{%
    \renewcommand{\arraystretch}{1} 
    \setlength{\tabcolsep}{2pt} 
    \begin{tabular}{lllccccccc}
    \toprule
      {\small Metrics} & {\small Methods} & \Coat & \KuaiRec & \Netflix & \Movielens & \Yahoo & \letor & \ltrc \\
    \midrule  
            \multirow{9}{*}{\expdcg}
			& \Random & 0.293 $\pm$ 0.255 & 0.299 $\pm$ 0.181 & 0.581 $\pm$ 0.341 & 0.525 $\pm$ 0.369 & 0.407 $\pm$ 0.423 & 0.199 $\pm$ 0.150 & 0.233 $\pm$ 0.220 \\
			& \explore & 1.269 $\pm$ 0.832 & 0.925 $\pm$ 0.229 & 2.221 $\pm$ 0.772 & 3.251 $\pm$ 0.191 & 2.074 $\pm$ 1.064 & 0.304 $\pm$ 0.219 & 0.433 $\pm$ 0.359 \\
			& \DUM & 1.370 $\pm$ 0.896 & 2.100 $\pm$ 0.658 & 2.335 $\pm$ 0.674 & 3.005 $\pm$ 0.129 & 1.991 $\pm$ 0.947 & 0.377 $\pm$ 0.258 & 0.667 $\pm$ 0.407 \\
			& \MSD & \textbf{1.393 $\pm$ 0.932} & \textbf{2.259 $\pm$ 0.760} & \textbf{2.492 $\pm$ 0.777} & \textbf{3.347 $\pm$ 0.083} & 2.020 $\pm$ 0.944 & \underline{0.514 $\pm$ 0.282} & 0.835 $\pm$ 0.502 \\
			& \MMR & 1.381 $\pm$ 0.919 & \underline{2.201 $\pm$ 0.743} & 2.447 $\pm$ 0.765 & \underline{3.294 $\pm$ 0.119} & \textbf{2.278 $\pm$ 0.995} & 0.513 $\pm$ 0.300 & \underline{0.841 $\pm$ 0.518} \\
			& \DPP & \underline{1.387 $\pm$ 0.924} & 2.161 $\pm$ 0.717 & \underline{2.481 $\pm$ 0.773} & 3.281 $\pm$ 0.163 & 1.991 $\pm$ 0.947 & \textbf{0.523 $\pm$ 0.291} & \textbf{0.855 $\pm$ 0.525} \\
			& \bketwo & 1.264 $\pm$ 0.881 & 1.271 $\pm$ 0.436 & 2.347 $\pm$ 0.749 & 3.274 $\pm$ 0.157 & 2.051 $\pm$ 1.031 & 0.376 $\pm$ 0.260 & 0.666 $\pm$ 0.408 \\
			& \bkethree & 1.269 $\pm$ 0.890* & 1.399 $\pm$ 0.415* & 2.368 $\pm$ 0.739* & 3.274 $\pm$ 0.156* & 2.131 $\pm$ 0.984* & 0.465 $\pm$ 0.253 & 0.737 $\pm$ 0.433 \\
			& \bkefour & 1.317 $\pm$ 0.899* & 1.453 $\pm$ 0.430* & 2.378 $\pm$ 0.738* & 3.275 $\pm$ 0.152* & {2.137 $\pm$ 0.986}* & 0.472 $\pm$ 0.260 & 0.751 $\pm$ 0.443 \\
            & \Sbkethree & 0.960 \jj 0.856* & 1.480 \jj 0.426* & 2.363 \jj 0.718* & 3.254 \jj 0.198* & {2.155 \jj 0.991}* & 0.432 \jj 0.275 & 0.721 \jj 0.443 \\
            & \Sbkefour & 0.969 \jj 0.863* & 1.561 \jj 0.451* & 2.371 \jj 0.714* & 3.254 \jj 0.196* & \underline{2.163 \jj 0.983}* & 0.444 \jj 0.288 & 0.744 \jj 0.465 \\
        \midrule
            \multirow{9}{*}{\expseren}
			& \Random & 0.124 $\pm$ 0.179 & -- & 0.068 $\pm$ 0.164 & 0.055 $\pm$ 0.156 & 0.051 $\pm$ 0.143 & -- & -- \\
			& \explore & 0.843 $\pm$ 0.912 & -- & 0.715 $\pm$ 0.790 & \textbf{1.265 $\pm$ 1.125} & 1.842 $\pm$ 1.588 & -- & -- \\
			& \DUM & 0.778 $\pm$ 0.869 & -- & 0.836 $\pm$ 0.801 & 1.030 $\pm$ 0.931 & {2.461 $\pm$ 1.974} & -- & -- \\
			& \MSD & 0.815 $\pm$ 0.933 & -- & 0.680 $\pm$ 0.774 & 0.744 $\pm$ 0.998 & 2.452 $\pm$ 1.975 & -- & -- \\
			& \MMR & 0.800 $\pm$ 0.890 & -- & 0.869 $\pm$ 0.893 & 1.179 $\pm$ 1.134 & 2.143 $\pm$ 1.987 & -- & -- \\
			& \DPP & 0.774 $\pm$ 0.897 & -- & 0.816 $\pm$ 0.889 & 1.164 $\pm$ 1.139 & 2.448 $\pm$ 1.954 & -- & -- \\
			& \bketwo & 0.863 $\pm$ 0.909 & -- & 0.871 $\pm$ 0.876 & \underline{1.223 $\pm$ 1.099} & \underline{2.505 $\pm$ 2.043} & -- & -- \\
			& \bkethree & \underline{0.903 $\pm$ 0.990}* & -- & \underline{0.873 $\pm$ 0.881}* & \underline{1.223 $\pm$ 1.099}* & 2.439 $\pm$ 2.076* & -- & -- \\
			& \bkefour & \textbf{0.931 $\pm$ 1.007}* & -- & \textbf{0.876 $\pm$ 0.886}* & \underline{1.223 $\pm$ 1.099}* & 2.426 $\pm$ 2.078* & -- & -- \\
		    & \Sbkethree & 0.574 \jj 0.802* & -- & 0.839 \jj 0.889* & 1.218 \jj 1.094* & \textbf{2.639 \jj 2.035}* & -- & -- \\
            & \Sbkefour & 0.574 \jj 0.801* & -- & 0.834 \jj 0.882* & 1.218 \jj 1.094* & \textbf{2.636 \jj 2.038}* & -- & -- \\
        \bottomrule
		\end{tabular}
	}
\end{table*}


\section{Conclusion}

In this paper, we study the novel concept of \emph{sequential diversity}, 
aiming to instill diversity into rankings
while considering the relevance of items and modeling user behavior. 
The framework is designed to find rankings that consist 
of items that are both diverse and relevant.
The formulation gives rise to a novel computational problem, 
for which we establish a connection with the ordered Hamiltonian-path problem
and design approximation algorithms with provable guarantees. 
Our algorithms offer trade-offs of efficiency vs.\ approximation~quality. 

Our paper opens many exciting directions for future work. 
First, it will be interesting to devise more efficient combinatorial
algorithms without losing their approximation guarantees. 
Additionally, it will be valuable to incorporate user models of higher complexity into the framework, 
in order to capture more nuanced user behavior. 
Last, it will be valuable to study the proposed framework with a user study on a real-world system.

\bibliographystyle{ACM-Reference-Format}
\bibliography{reference}

\appendix
\clearpage
\input{TIST_extension/appendix_TIST}

\end{document}

%% file: header.tex
\usepackage{graphicx}
\usepackage{amsmath}
\usepackage{amsfonts}
\usepackage{subcaption}
\usepackage{xspace}
\usepackage{booktabs}
\usepackage{multirow}
\usepackage{thmtools}		
\usepackage{mleftright}
\usepackage{hyperref}
\usepackage{cleveref}
\usepackage{xcolor}
\usepackage{enumitem}
\usepackage{thm-restate}
\usepackage{tabularx}
\usepackage{nicefrac}
\usepackage{mathtools}

\setlist[enumerate, 1]{labelindent=\parindent, labelwidth=0.5em, leftmargin=!, align=left, topsep=1pt}
\setlist[itemize, 1]{labelindent=\parindent, labelwidth=0.5em, leftmargin=!, align=left, topsep=1pt}

\def\m+#1{\ensuremath{\mathbf{#1}}\xspace}
\def\v+#1{\ensuremath{\mathbf{#1}}\xspace}
\def\t+#1{\ensuremath{\tilde{\mathbf{#1}}}\xspace}
\def\e+#1{\ensuremath{#1}\xspace}
\def\arrow+#1{\ensuremath{\overrightarrow{#1}}\xspace}
\def\set+#1{\ensuremath{\mathbb{#1}}\xspace}
\def\arrow+#1{\ensuremath{\overrightarrow{#1}}\xspace}

\DeclareMathOperator*{\argmax}{arg\,max}
\newcommand{\NP}{\ensuremath{\mathbf{NP}}\xspace}

\newcommand{\NPhard}{\NP-hard\xspace}

\newcommand{\bigO}{\ensuremath{\mathcal{O}}\xspace}

\newcommand{\abs}[1]{\ensuremath{\lvert #1 \rvert }\xspace}

\numberwithin{equation}{section}

\newtheorem{corollary}{Corollary}
\newtheorem{lemma}{Lemma}

\newtheorem{definition}{Definition}
\newtheorem{problem}{Problem}

\newtheorem{example}{Example}
\newtheorem{observation}{Observation}

\newcommand{\sj}[1]{{\color{blue}Sijing: #1}}
\newcommand{\honglian}[1]{{\color{red}Honglian: #1}}

\newcommand{\edit}[1]{{\color{black} #1}}

\newcommand{\jj}{\ensuremath{\pm}\xspace}

\newcommand{\para}[1]{\noindent{\textbf{#1}}}

\newcommand{\odo}{\ensuremath{\mathcal{S}}\xspace} 
\newcommand{\osdo}{\ensuremath{\mathcal{S}_{+}}\xspace} 
\newcommand{\ocdo}{\ensuremath{\mathcal{S}_{c}}\xspace} 
\newcommand{\omsd}{\textsc{Max\-SSD}\xspace} 
\newcommand{\omcd}{\textsc{Max\-SCD}\xspace} 
\newcommand{\ohpo}{\ensuremath{\mathcal{H}}\xspace} 
\newcommand{\omshp}{\textsc{Max\-OHP}\xspace} 
\newcommand{\diversity}{\ensuremath{\operatorname{Div}}\xspace}

\newcommand{\bp}{\ensuremath{\textsc{BI}}\xspace} 

\newcommand{\bke}{\ensuremath{\small \textsc{B}\kpar\textsc{I}}\xspace} 
\newcommand{\bkeh}{\ensuremath{\small \textsc{B}\kpar\textsc{I-H}}\xspace} 
\newcommand{\bkm}{\ensuremath{\small \textsc{GM}}\xspace} 

\newcommand{\Sbke}{\ensuremath{\small \textsc{SB}\kpar\textsc{I}}\xspace} 
\newcommand{\Sbkeh}{\ensuremath{\small \textsc{SB}\kpar\textsc{I-H}}\xspace}

\newcommand{\Sbketwo}{\textsf{\small SB2I}\xspace}
\newcommand{\Sbkethree}{\textsf{\small SB3I}\xspace}
\newcommand{\Sbkefour}{\textsf{\small SB4I}\xspace}
\newcommand{\SbkethreeH}{\textsf{\small SB3I-H}\xspace}
\newcommand{\SbkefourH}{\textsf{\small SB4I-H}\xspace}

\newcommand{\alg}{\ensuremath{\textsc{ALG}}\xspace} 

\newcommand{\opt}{\ensuremath{\textsc{OPT}}\xspace} 

\makeatletter
\newcommand{\order}{\@ifnextchar\bgroup\order@i{\ensuremath{\pi}\xspace}}
\newcommand{\order@i}[1]{\ensuremath{\pi(#1)}\xspace}

\newcommand{\optorder}{\@ifnextchar\bgroup\optorder@i{\ensuremath{\pi^{o}}\xspace}}
\newcommand{\optorder@i}[1]{\ensuremath{\pi^{o}(#1)}\xspace}

\newcommand{\orderhat}{\@ifnextchar\bgroup\orderhat@i{\ensuremath{\hat{\pi}}\xspace}}
\newcommand{\orderhat@i}[1]{\ensuremath{\hat{\pi}(#1)}\xspace}

\newcommand{\orderstar}{\@ifnextchar\bgroup\orderstar@i{\ensuremath{\pi^{*}}\xspace}}
\newcommand{\orderstar@i}[1]{\ensuremath{\pi^{*}(#1)}\xspace}

\newcommand{\optordSet}{\@ifnextchar\bgroup\optordSet@i{\ensuremath{\ordSet^{o}}\xspace}}
\newcommand{\optordSet@i}[1]{\ensuremath{\ordSet^{o}_{#1}}\xspace}

\newcommand{\ordSetstar}{\@ifnextchar\bgroup\ordSetstar@i{\ensuremath{\ordSet^{*}}\xspace}}
\newcommand{\ordSetstar@i}[1]{\ensuremath{\ordSet^{*}_{#1}}\xspace}

\newcommand{\matchorder}{\@ifnextchar\bgroup\matchorder@i{\ensuremath{\pi^{M}}\xspace}}
\newcommand{\matchorder@i}[1]{\ensuremath{\pi^{M}(#1)}\xspace}

\newcommand{\matchordSet}{\@ifnextchar\bgroup\matchordSet@i{\ensuremath{\ordSet^{M}}\xspace}}
\newcommand{\matchordSet@i}[1]{\ensuremath{\ordSet^{M}_{#1}}\xspace}

\newcommand{\divS}{\@ifnextchar\bgroup\divS@i{\ensuremath{S}\xspace}}
\newcommand{\divS@i}[1]{\ensuremath{S_{#1}}\xspace}
\newcommand{\divSopt}{\@ifnextchar\bgroup\divSopt@i{\ensuremath{S^{*}}\xspace}}
\newcommand{\divSopt@i}[1]{\ensuremath{S^{*}_{#1}}\xspace}

\newcommand{\chunk}{\@ifnextchar\bgroup\chunk@i{\ensuremath{L}\xspace}}
\newcommand{\chunk@i}[1]{\ensuremath{L_{#1}}\xspace}
\newcommand{\chunkmax}{\ensuremath{L_{\max}}\xspace}
\newcommand{\numchunk}{\ensuremath{T}\xspace}

\newcommand{\instance}{\ensuremath{I}\xspace}

\newcommand{\user}{\ensuremath{u}\xspace}

\newcommand{\ordSet}{\@ifnextchar\bgroup\ordSet@i{\ensuremath{O}\xspace}}
\newcommand{\ordSet@i}[1]{\ensuremath{O_{#1}}\xspace}

\newcommand{\itempr}{\@ifnextchar\bgroup\itempr@i{\ensuremath{p}\xspace}}
\newcommand{\itempr@i}[1]{\ensuremath{p_{#1}}\xspace}

\newcommand{\dist}{\@ifnextchar\bgroup\dist@i{\ensuremath{d}\xspace}}
\newcommand{\dist@i}[1]{\ensuremath{d \left(#1\right)}\xspace}

\newcommand{\ordsum}{\@ifnextchar\bgroup\ordsum@i{\ensuremath{\mathcal{S}}\xspace}} 
\newcommand{\ordsum@i}[1]{\ensuremath{\mathcal{S}(#1)}\xspace}

\newcommand{\ordpath}{\@ifnextchar\bgroup\ordpath@i{\ensuremath{\mathcal{H}}\xspace}} 
\newcommand{\ordpath@i}[1]{\ensuremath{\mathcal{H}(#1)}\xspace}

\newcommand{\ordpathhat}{\@ifnextchar\bgroup\ordpathhat@i{\ensuremath{\hat{\mathcal{H}}}\xspace}} 
\newcommand{\ordpathhat@i}[1]{\ensuremath{\hat{\mathcal{H}}(#1)}\xspace}

\newcommand{\ordpathtilde}{\@ifnextchar\bgroup\ordpathtilde@i{\ensuremath{\tilde{\mathcal{H}}}\xspace}} 
\newcommand{\ordpathtilde@i}[1]{\ensuremath{\tilde{\mathcal{H}}(#1)}\xspace}

\makeatother

\newcommand{\kpar}{\ensuremath{\tau}\xspace} 
\newcommand{\kprefix}{\ensuremath{k}\xspace} 
\newcommand{\randomordSet}{\ensuremath{A}\xspace} 
\newcommand{\genres}{\ensuremath{\mathcal{X}}\xspace}
\newcommand{\attr}{\ensuremath{X}\xspace}

\newcommand{\unordSet}{\ensuremath{U}\xspace}
\newcommand{\ordSetswap}{\ensuremath{\ordSet'}\xspace} 
\newcommand{\pr}{\ensuremath{\mathrm{Pr}}\xspace}
\newcommand{\Exp}{\ensuremath{\mathbb{E}}\xspace} 

\newcommand{\divf}{\ensuremath{\mathcal{D}}\xspace} 
\newcommand{\divfsum}{\ensuremath{\mathcal{D}_{+}}\xspace} 
\newcommand{\divfcov}{\ensuremath{\mathcal{D}_{c}}\xspace} 
\newcommand{\distpath}[1]{\ensuremath{d_{L}(#1)}\xspace} 

\newcommand{\xxi}{\ensuremath{x_i}\xspace}
\newcommand{\xoi}{\ensuremath{x_i^o}\xspace}
\newcommand{\yi}{\ensuremath{y_i}\xspace}
\newcommand{\yoi}{\ensuremath{y_i^o}\xspace}

\newcommand{\wi}{\ensuremath{W_{\ordSet_i}}\xspace}
\newcommand{\wiodd}{\ensuremath{W_{\ordSet_{2i-1}}}\xspace}
\newcommand{\wieven}{\ensuremath{W_{\ordSet_{2i}}}\xspace}
\newcommand{\woi}{\ensuremath{W_{\ordSet_i^o}}\xspace}
\newcommand{\wmi}{\ensuremath{W_{\ordSet_i^M}}\xspace}

\newcommand{\halfn}{\ensuremath{\mathcal{T}_n}\xspace}

\newcommand{\numseg}{\ensuremath{T}\xspace}
\newcommand{\bestk}{\ensuremath{\ell_{\kpar}}\xspace}
\newcommand{\doij}{\ensuremath{d(\pi^o(i),\pi^o(i+1))}\xspace}

\newcommand{\dij}{\ensuremath{d(\pi(i),\pi(i+1))}\xspace}

\newcommand{\clique}{\ensuremath{\text{CLIQUE}}\xspace}

\usepackage[english,linesnumbered,vlined,ruled,nokwfunc]{algorithm2e}
\DontPrintSemicolon
\SetKwComment{note}{$\triangleright$ }{}
\SetFuncSty{textsc}
\SetCommentSty{textit}
\SetDataSty{texttt}
\SetKwInput{Initial}{Initial}
\SetKwInput{Parameter}{Param.}
\SetKwInput{Input}{Input}
\SetKwInput{Data}{Data}
\SetKwInput{Output}{Output}
\ResetInOut{Result} 
\let\oldnl\nl%
\newcommand{\nonl}{\renewcommand{\nl}{\let\nl\oldnl}}%
\SetKw{KwAnd}{and} %
\SetKw{KwOr}{or}
\SetKw{KwXor}{xor}
\SetKw{KwNot}{not}
\SetKw{Parallel}{parallel}
\SetKw{Return}{return}
\SetKw{Break}{break}

\newcommand{\Coat}{\textsf{\small Coat}\xspace}
\newcommand{\KuaiRec}{\textsf{\small KuaiRec}\xspace}
\newcommand{\Netflix}{\textsf{\small Netflix}\xspace}
\newcommand{\Movielens}{\textsf{\small Movielens}\xspace}
\newcommand{\Yahoo}{\textsf{\small Yahoo}\xspace}
\newcommand{\letor}{\textsf{\small LETOR}\xspace}
\newcommand{\ltrc}{\textsf{\small LTRC}\xspace}

\newcommand{\bketwo}{\textsf{\small B2I}\xspace}

\newcommand{\bkethree}{\textsf{\small B3I}\xspace}
\newcommand{\bkethreeH}{\textsf{\small B3I-H}\xspace}
\newcommand{\bkefour}{\textsf{\small B4I}\xspace}
\newcommand{\bkefourH}{\textsf{\small B4I-H}\xspace}
\newcommand{\Random}{\textsf{\small Random}\xspace}
\newcommand{\DUM}{\textsf{\small DUM}\xspace}
\newcommand{\MSD}{\textsf{\small MSD}\xspace}
\newcommand{\MMR}{\textsf{\small MMR}\xspace}
\newcommand{\DPP}{\textsf{\small DPP}\xspace}
\newcommand{\explore}{\textsf{\small EXPLORE}\xspace}

\newcommand{\expnum}{\ensuremath{\mathit{Expnum}}\xspace}
\newcommand{\expdcg}{\ensuremath{\mathit{Exp\-DCG}}\xspace}

\newcommand{\expseren}{\ensuremath{\mathit{Exp\-Serendip\-i\-ty}}\xspace}

%% file: TIST_extension/appendix_TIST.tex
\section{Omitted Proofs from Section~\ref{section:reduction}}

\bigskip
For the proof of \Cref{obs:ocdo_ordered_submodular} and \Cref{obs:omsd_ordered_submodular}, 
we use two lemma from \citet{kleinberg2022ordered}. 

\begin{lemma}[\citet{kleinberg2022ordered}]
\label{lemma:klein_lemma2}
    If $f$ and $g$ are ordered-submodular, then $\alpha f + \beta g$ is also ordered-submodular for any $\alpha, \beta \geq 0$.
\end{lemma}

\begin{lemma}[\citet{kleinberg2022ordered}]
\label{lemma:klein_lemma3}
    Suppose $h$ is a monotone submodular set function, Then the function $f$ constructed by evaluating $h$ on the set of the first $t$ elements of $S$, that is,
\begin{equation*}
    f(S) = \left\{\begin{matrix}
h(S) & \text{if } |S| \leq t\\ 
h(S_t) & \text{if } |S| >  t
\end{matrix}\right.
\end{equation*}
is ordered-submodular.
\end{lemma}

\bigskip
\ocdoOrderedSubmodular*
\begin{proof}
By the definition of the \ocdo objective we have:
\begin{equation}
    \ocdo(\ordSet) = \Exp_{\randomordSet \sqsubseteq \ordSet}[\divf(\randomordSet)] = \sum_{\kprefix=1}^{n} \pr(\randomordSet = \ordSet{\kprefix})\divf(\ordSet{\kprefix}),
    \end{equation}
where $\divf(\ordSet{\kprefix}) $ is a monotone submodular function. Let $f_{\kprefix}(\ordSet) = \divf(\ordSet{\kprefix})$ so we can rewrite $\ocdo(\ordSet)$ as $\ocdo(\ordSet) = \sum_{\kprefix=1}^{n}\itempr{\ordSet{\kprefix}}f_{\kprefix}(\ordSet)$. Since $f_{\kprefix}(\ordSet)$ is a monotone submodular function evaluated on the first $\kprefix$ items, by \Cref{lemma:klein_lemma3}, $f_{\kprefix}(\ordSet)$ is an ordered submodular function. Since $\pr(\ordSet{\kprefix}) \geq 0$ for all $\kprefix$, by \Cref{lemma:klein_lemma2}, $\ocdo(\cdot) $ is ordered-submodular.
\end{proof}

\bigskip
\omsdOrderedSubmodular*
\begin{proof}
    \label{example:not-ordered-submodular}
To show $\osdo(\cdot)$ is not ordered submodular, we only need one counter example.    
Let $X$ and $Y$ be two sequences
that satisfy $\itempr{X} > 0$ and $\itempr{Y} > 0$. Let $s$ be an item such that $p_s > 0$ and $d(s, Y) > 0$. Also, let $\bar{s}$ be an item such that $p_{\bar{s}} = 0$. It follows that 
    \begin{align*}
    \osdo( X|| s || Y) - \osdo( X|| \bar{s} || Y) 
        & = \osdo( X|| s || Y) -\osdo( X ) \\
        &  > \osdo( X|| s) -\osdo( X ).
    \end{align*}

    This example violates the definition of ordered submodularity, 
    which proves that $\osdo(\cdot)$ is not an ordered submodular function.
\end{proof}

\bigskip
\omsdNPHard*
\begin{proof}

We provide a reduction from the decision version of the maximum clique problem to the decision version of the $\omsd$ problem. Define $\clique = (G,V,\kprefix)$, with $|V| = n$ to describe instances of the maximum clique problem, where we want to decide whether there exists a clique of size $\kprefix$ in $G$. 
Also define an instance of the \omsd problem by $(U, d, p, \theta)$, where $U$ is the item set, $p$ is a uniform continuation probability, 
and $d(\cdot,\cdot)$ is a metric distance function, and we want to decide whether there exist an ordering $\order$ of $U$ such that $\osdo(\order(U)) > \theta$. 

Given an instance of \clique, we construct an instance of \omsd in the following way: 
we set $U = V$, and we assign $d(u,v) = 2$ if there is an edge between nodes $u$ and $v$ in $G$,
otherwise we set $d(u,v) = 1+ \epsilon$, where $\epsilon$ is any constant smaller than 1. 
We set $p = \frac{1-\epsilon}{2n^2}$ and 
$\theta = 2 \sum_{i=1}^{\kprefix-1} i \times \left(\frac{1-\epsilon}{2n^2} \right)^{i+1} $. 
One can easily verify that this transformation can be done in polynomial time, and that $d(\cdot,\cdot)$ is a metric.

If there is a clique of size $\kprefix \leq n$ in $G$, then for the corresponding item set $U=V$, we can construct an ordered sequence $\ordSet = (\order{i})_{i=1}^{n}$ by assigning $\ordSet{\kprefix}$ to be the $\kprefix$ nodes in $U$ that corresponds to the clique in $G$, and setting in an arbitrary way the rest of the sequence order.

Let $\mathcal{L}(\ordSet)$ be a lower bound of $\osdo(A) $. 
One can verify it is also a lower bound for the optimal objective value of the $\omsd$ problem.
We can define $\mathcal{L}(\ordSet)$ as follows: 
\begin{equation*}
    \begin{aligned}
        \mathcal{L}(\ordSet) & = \sum_{i=1}^{\kprefix-1} \itempr{\ordSet{i+1}} \dist(\order{i+1}, \ordSet{i})\\
        & = 2(\itempr^2 + 2\itempr^3 + 3\itempr^4 + \cdots + (\kprefix-2)\itempr^{\kprefix-1} + (\kprefix-1)\itempr^\kprefix)\\
        & = 2 \sum_{i=1}^{\kprefix-1} i \times p^{i+1} .
    \end{aligned}
\end{equation*}
Furthermore, it holds that $\mathcal{L}(\ordSet)  = \theta$.

On the other hand, if there does not exist a size $\kprefix$ clique in $G$, 
then the densest possible graph structure of $G$ contains $T = \left\lfloor \frac{n}{\kprefix-1} \right\rfloor$ 
size-$(\kprefix-1)$ cliques, and one clique of size $L = n - (\kprefix-1)T$. 
A node in a size-$(\kprefix-1)$ clique can connect to at most $\kprefix-2$ nodes in another size-$(\kprefix-1)$ clique, 
thus the maximum number of edges between any two size-$(\kprefix-1)$ cliques is $(\kprefix-1)(\kprefix-2)$. 
Similarly, the number of edges between a size-$L$ clique and a size-$(\kprefix-1)$ clique is $(\kprefix-2)L$. 
Given this graph structure, we can construct an ordered sequence $\tilde{\ordSet} = (\tilde{\order}(i))_{i=1}^{n}$, 
where sequence $[\tilde{\order}((j-1)(\kprefix-1) + 1), \cdots, \tilde{\order}((\kprefix-1) \times j)]$ 
corresponds to the $j$-th size-$(\kprefix-1)$ clique of $G$, with $j \in [1, T]$, 
and the remaining sequence corresponds to the size-$L$ clique of $G$.

Let $\mathcal{U}(\tilde{\ordSet}) = \osdo(\tilde{\ordSet})$.
One can verify  that $\mathcal{U}(\tilde{\ordSet})$ is an upper bound of the optimal objective value of the corresponding $\omsd$ problem. 
It holds that
\begin{equation}
    \begin{aligned}
        \mathcal{U}(\tilde{\ordSet}) &= \sum_{i=1}^{n-1} \itempr_{\tilde{\ordSet}_{i+1}} \dist(\tilde{\order}(i+1), \tilde{\ordSet}_i )\\
         &\stackrel{(a)}{=} 2(\itempr^2 + 2\itempr^3 + \cdots + (\kprefix-2)\itempr^{\kprefix-1}) \\
        & + \sum_{i=1}^{T-1} \sum_{j=0}^{\kprefix-2} \itempr^{j + 1 + i(\kprefix-1)} (2[i(\kprefix-2)+j]+i(1+\epsilon)) \\
        & + \sum_{j=0}^{L-1} \itempr^{T(\kprefix-1) + j + 1} (2[T(\kprefix-2) + j]+ T(1+\epsilon)) \\
        &\stackrel{(b)}{<} \mathcal{L}(\ordSet) + (\epsilon-1)\itempr^\kprefix + 2n(\itempr^{\kprefix+1} + \itempr^{\kprefix+2} + \cdots + \itempr^n),
    \end{aligned}
\end{equation}
where $(a)$ holds because of \Cref{lem:eqform} and $(b)$ holds because 
for $j \in [\kprefix+1, n]$, the $j$-th node contributes at most $2n\itempr^j$ to $\mathcal{U}(\tilde{\ordSet})$.

Furthermore, since 
\begin{align*}
    \mathcal{L}(\ordSet) - \mathcal{U}(\tilde{\ordSet})&   = \theta - \mathcal{U}(\tilde{\ordSet})\\
    & > (1-\epsilon)\itempr^\kprefix - 2n(\itempr^{\kprefix+1} + \itempr^{\kprefix+2} + \cdots + \itempr^n) \\
    & = (1-\epsilon)\itempr^\kprefix - 2np(\itempr^{\kprefix} + \itempr^{\kprefix+1} + \cdots + \itempr^{n-1}) \\
    & > (1-\epsilon) \itempr^\kprefix - 2n \itempr (n-\kprefix) \itempr^\kprefix \\
    & = \itempr^\kprefix (1 - \epsilon -2n \itempr (n-\kprefix)) = \frac{p^\kprefix (1-\epsilon) \kprefix}{n} > 0
\end{align*}
holds for all $\kprefix$, 
we conclude that a \texttt{yes} instance of the \clique problem indicates a \texttt{yes} instance of the \omsd problem, 
and a \texttt{no} instance indicates a \texttt{no} instance of the \omsd problem.

This finishes our reduction from the \clique problem to the \omsd problem and 
we conclude that the \omsd problem is \NP-hard. 
\end{proof}

\oshpreform*
\begin{proof}
We simply re-arrange the formulas. 
\begin{equation}
\begin{aligned}
            \ordpath(\ordSet) & = \sum_{i = 1}^{n-1} w_i \dist{\order{i}, \order{i+1}}\\
            & = \sum_{i = 1}^{n-1} \sum_{\kprefix=i+1}^{n} \itempr{\ordSet{\kprefix}} \dist{\order{i}, \order{i+1}} \\
            &= \itempr{\ordSet{2}} \dist{\order{1}, \order{2}} + \itempr{\ordSet{3}} (\dist{\order{1}, \order{2}} + \dist{\order{2},\order{3}}) \\
            & + \cdots  +\itempr(\ordSet) \sum_{t = 1}^{n-1} \dist{\order{t},\order{t+1}} \\
            &= \itempr{\ordSet{2}} \distpath{\ordSet{2}} + \itempr{\ordSet{3}} \distpath{\ordSet{3}} + \cdots +\itempr(\ordSet)\distpath{\ordSet} \\
            &= \sum_{i=1}^{n-1}\itempr{\ordSet{i+1}} \distpath{\ordSet{i+1}}.
\end{aligned}
\end{equation}
\end{proof}

\clearpage
\anyinequality*
\begin{proof}
First, we show that if $2\dist{\order{i+1}, \ordSet{i}} \geq \distpath{\ordSet{i+1}}$ holds for any $i \in [n-1]$, then our lemma is proven: 
\begin{align*}
    2\osdo(\ordSet) &= 2\sum_{i=1}^{n-1} \itempr{\ordSet{i+1}} \dist{\order{i+1}, \ordSet{i}} \\
    &= \sum_{i=1}^{n-1} \itempr{\ordSet{i+1}} 2\dist{\order{i+1}, \ordSet{i}} \\
    &\geq \sum_{i=1}^{n-1}   \itempr{\ordSet{i+1}} \distpath{\ordSet{i+1}} \\
    &=\ordpath{\ordSet}.
\end{align*}

Next, we show that $2\dist{\order{i+1}, \ordSet{i}} \geq \distpath{\ordSet{i+1}}$ holds for any $1\leq i \leq n-1$:
\begin{align*}
    d_L(\ordSet{i+1}) &= \sum_{j=1}^{i} \dist{\order{j}, \order{j+1}} \\
    &\stackrel{(a)}{\leq}  \sum_{j=1}^{i} \left [ d(\order{i+1},\order{j}) + d(\order{i+1}, \order{j+1}) \right ] \\
    &= \sum_{j=1}^{i} d(\order{i+1},\order{j}) + \sum_{j=1}^{i} d(\order{i+1}, \order{j+1}) \\
    &= d(\order{i+1}, \ordSet{i}) +  d(\order{i+1}, \ordSet{i+1}) - d(\order{i+1}, \order{1})\\
    &\stackrel{(b)}{=} d(\order{i+1}, \ordSet{i}) +  d(\order{i+1}, \ordSet{i}) - d(\order{i+1}, \order{1})\\ 
    &<  2\dist{\order{i+1}, \ordSet{i}},
\end{align*}
where $(a)$ holds by triangle inequality, and $(b)$ holds since\\
$\dist{\order{i+1}, \order{i+1}} = 0$.
\end{proof}

\localoptimalnonuniformp*
\begin{proof}
\label{appendix:proofs_of_local_optimality_for_nonuniformp}

Let $\ordSet^* = (\order{i})_{i=1}^{n}$ be an optimal sequence, and  
let us denote $\ordSetswap= (\ldots, \order{i+1}, \order{i}, \ldots)$ 
to be the sequence after swapping node pair $\{\order{i}, \order{i+1}\}$. 
Then, 
\begin{equation*}
    \begin{aligned}
        \osdo (\ordSetswap) &= \sum_{j=1}^{i-2} \itempr{\ordSet^*_{j+1}} \dist(\order{j+1}, \ordSet^*_{j}) 
        +  \itempr{\order{i+1}} \itempr{\ordSet^*_{i-1}} \dist (\order{i+1}, \ordSet^*_{i-1}) \\
        & + \itempr{\order{i}} \itempr{\order{i+1}}  \itempr{\ordSet^*_{i-1}} \dist (\order{i}, \ordSet^*_{i-1} \cup \order{i+1}) \\
        & + \sum_{j=i+1}^{n-1} \itempr{\ordSet^*_{j+1}} \dist(\order{j+1}, \ordSet^*_{j}),
    \end{aligned}
\end{equation*}
and
\begin{equation*}
    \begin{aligned}
        \osdo (\ordSet^*) &= \sum_{j=1}^{i-2} \itempr{\ordSet^*_{j+1}} \dist(\order{j+1}, \ordSet^*_{j}) 
        + \itempr{\order{i}} \itempr{\ordSet^*_{i-1}}  \dist (\order{i}, \ordSet^*_{i-1}) \\
        & +  \itempr{\order{i}} \itempr{\order{i+1}}  \itempr{\ordSet^*_{i-1}} \dist (\order{i+1}, \ordSet^*_{i-1} \cup \order{i}) \\
        & + \sum_{j=i+1}^{n-1} \itempr{\ordSet^*_{j+1}} \dist(\order{j+1}, \ordSet^*_{j}).
    \end{aligned}
\end{equation*}

Since $\ordSet^* = \argmax \osdo{(\ordSet)}$ is an optimal sequence, 
it must hold that $\osdo(\ordSet^*)-\osdo(\ordSetswap) \geq 0$, that is,
\begin{equation*}
\label{eq:nonuniform_p_local_swap}
    \begin{aligned}
\osdo(\ordSet^*)-\osdo(\ordSetswap)
& = \itempr{\order{i}} \itempr{\ordSet^*_{i-1}}  \dist (\order{i}, \ordSet^*_{i-1})\\
& - \itempr{\order{i+1}} \itempr{\ordSet^*_{i-1}} \dist (\order{i+1}, \ordSet^*_{i-1}) \\ 
& + \itempr{\order{i}} \itempr{\order{i+1}}  \itempr{\ordSet^*_{i-1}} \dist (\order{i+1}, \ordSet^*_{i-1} )\\
& - \itempr{\order{i}} \itempr{\order{i+1}}  \itempr{\ordSet^*_{i-1}} \dist (\order{i}, \ordSet^*_{i-1}) \geq 0.
    \end{aligned}
\end{equation*}
This is equivalent to
\begin{equation}
\label{eq:lemma_telescope_sum_nonuniform}
\begin{aligned}
    & \itempr{\order{i}} \itempr{\ordSet^*_{i-1}} (1-\itempr{\order{i+1}}) \dist(\order{i}, \ordSet^*_{i-1}) \\
& \geq   \itempr{\order{i+1}} \itempr{\ordSet^*_{i-1}} (1 - \itempr{\order{i}}) \dist(\order{i+1}, \ordSet^*_{i-1}).
\end{aligned}
\end{equation}
Since $\itempr{\ordSet^*_{i-1}} > 0$, 
we can cancel that term from both sides of \Cref{eq:lemma_telescope_sum_nonuniform}. 
We then divide both sides by $(1 - \itempr{\order{i}}) \times (1 - \itempr{\order{i+1}})$ and obtain
\begin{equation}
\label{eq:before_tele_sum}
\frac{\itempr{\order{i}}}{1-\itempr{\order{i}}} \dist(\order{i}, \ordSet^*_{i-1}) \geq \frac{\itempr{\order{i+1}}}{1-\itempr{\order{i+1}}} \dist(\order{i+1}, \ordSet^*_{i-1}).
\end{equation}
By taking the telescope sum of \Cref{eq:before_tele_sum} over $i$ from $i=2$ to $i=j$, we get 
\begin{equation}
\label{eq:after_tele_sum}
    \sum_{i=2}^{j} \frac{\itempr{\order{i}}}{1-\itempr{\order{i}}} \dist(\order{i}, \order{i-1}) \geq \frac{\itempr{\order{j+1}}}{1-\itempr{\order{j+1}}} \dist(\order{j+1}, \ordSet^*_{j-1}).
\end{equation}
By adding $\frac{\itempr{\order{j+1}}}{1-\itempr{\order{j+1}}} \dist(\order{j+1}, \order{j})$ on both sides of \Cref{eq:after_tele_sum}, we get 
\begin{equation}
\label{eq:last_step}
    \sum_{i=2}^{j+1} \frac{\itempr{\order{j}}}{1-\itempr{\order{i}}} \dist(\order{i}, \order{i-1}) \geq \frac{\itempr{\order{j+1}}}{1-\itempr{\order{j+1}}} \dist(\order{j+1}, \ordSet^*_{j}).
\end{equation}
After moving $\frac{\itempr{\order{j+1}}}{1-\itempr{\order{j+1}}}$ to the left hand side of \Cref{eq:last_step} 
and exchanging the indexes of $i$ and $j$, we get
\begin{equation*}
    \dist(\order{i+1}, \ordSet^*_{i}) \leq \frac{1-\itempr{\order{i+1}}} {\itempr{\order{i+1}}} \sum_{j=1}^{i} \frac{\itempr{\order{j+1}}}{1-\itempr{\order{j+1}}} \dist(\order{j}, \order{j+1}),
\end{equation*}
which completes the proof.
\end{proof}

\bigskip
\optimalinequalitynonuniformp*
\begin{proof}
Let $\ordSet^* = (\order{i})_{i=1}^{n}$ denote an optimal sequence. 
    To proof this corollary it suffices to show that 
    $\dist{\order{i+1}, \ordSet^*_{i}} \leq \frac{b(1-a)}{a(1-b)}d_L(\ordSet^*_{i+1})$, for all $i \in [n-1]$. 

    Since we assume that $\itempr_i \in [a,b]$, for all $i \in [n]$, we get
   \begin{align}
   \label{eq:collory_ab}
       \frac{1-\itempr{\order{i+1}}}{\itempr{\order{i+1}}} & \sum_{j=1}^{i} \frac{\itempr{\order{j+1}}}{1-\itempr{\order{j+1}}} d(\order{j},\order{j+1}) \nonumber \\
       & \leq \frac{b(1-a)}{a(1-b)}d_L(\ordSet^*_{i+1}).
   \end{align}
    Substituting \Cref{eq:collory_ab} into \Cref{lemma:local_optimal_nonuniform_p} leads to
    \begin{align*}
        \dist{\order{i+1}, \ordSet{i}^{*}} & \leq \frac{1-\itempr_{\order{i+1}}}{\itempr_{\order{i+1}}} \sum_{j=1}^{i} \frac{\itempr_{\order{j+1}}}{1-\itempr_{\order{j+1}}} \dist{\order{j},\order{j+1}} \\
        & \leq \frac{b(1-a)}{a(1-b)}d_L(\ordSet^*_{i+1}),
    \end{align*}
which completes the proof.   
\end{proof}



\section{Additional Preliminaries}
In this section, we list a few useful lemmas that we are going to use repeatedly in the following sections. 
Since they are simple algebraic statements, we state them directly, without proof. 

\medskip
\begin{lemma}
    \label{lem:geo-sum}
    Let $s_n$ be the sum of the the first $n$ terms of a geometric series, 
    namely $s_n = \sum_{i=0}^{n-1} a r^{i}$, where $a$ is a constant, 
    and $r$ is any number such $0< r < 1$.
    Then, $s_n = a\frac{1-r^n}{1 - r}$. 
\end{lemma}

\medskip
\begin{lemma}
    \label{lem:geo-sum-constant-simple}
    When $r$ is a constant, $\lim_{n \rightarrow \infty} s_n = \frac{a}{1-r}$.
    In other words, when $n \rightarrow \infty$, it is $ s_n = \frac{a}{1-r} - \Theta(r^n)$.   
\end{lemma}

\medskip
\begin{lemma}[Theorem 4.3, Chebyshev's inequality~\cite{cvetkovski2012inequalities}]
    \label{lem:chebishev}
    Let $a_1  \leq \ldots \leq a_n$ and 
    $b_1 \leq \ldots b_n$ be real numbers.
    Then, it holds 
    \[
    \left(\sum_{i=1}^n a_i\right) \left(\sum_{i=1}^n b_i\right) \leq n \sum_{i=1}^n a_i b_i,
    \]
    and equality holds when $a_1 = \ldots = a_n$ or $b_1 = \ldots = b_n$. 
\end{lemma}

\medskip
\begin{lemma}
    \label{lem:laurent}
    $\left(1- \frac{1}{n}\right)^n = \frac{1}{e} - \Theta\!\left(\frac{1}{n}\right)$.
\end{lemma}

\section{Omitted Proofs from Section~\ref{section:uniform}}

\subsection{Proof of Theorem~\ref{thm:p-nonconstant-small-two}}
To prove \Cref{thm:p-nonconstant-small-two}, we first analyze the property of the greedy matching $M$ obtained from \Cref{{alg:bkm}}.
\begin{lemma}
\label{lemma:greedy_half_approx}
    Given a graph $G = (V, E)$ with $|V|=n$, let $w(S)$ denote the sum of edge weights in an edge subset $S \subseteq E$. Let $M$ be the greedy matching of $G$, $M_{k}$ denote the first $k$ edges of $M$, and $M_{k}^{*}$ be the maximum weighted size $k$ matching of $G$. Then it holds that $w(M_{k}) \geq \frac{1}{2} w(M_{k}^{*})$ for all $k \leq \lfloor \frac{n}{2} \rfloor$.
\end{lemma}
\begin{proof}
    Assume there exists an edge $e \in M_{k}^{*} \setminus M_{k}$, then there must exist an edge $e' \in M_{k}$ such that $e'$ share one end node with $e$ and $w(e) < w(e')$, because otherwise the greedy matching algorithm will choose $e$ instead of $e'$. Since each such $e'$ can conflict at most two edges from $M_{k}^{*}$, we have that
\begin{equation*}
    2 \sum_{e' \in M_{k} \setminus M_{k}^{*}} w(e') > \sum_{e \in M_{k}^{*} \setminus M_{k}} w(e),
\end{equation*}
and the lemma thus follows.
\end{proof}

\pnonconstantsmalltwo*
\begin{proof}
    We let $\order$ be the ordering of $\unordSet$ 
    obtained from $\bkm$ algorithm, and 
$\ordSet$ be the corresponding sequence. 
Recall that 
\begin{equation}
    \label{eq:ordset-matching}
    \begin{aligned}
    \ordpath(\ordSet) &= \sum_{i=1}^{n-1} \wi \dist{\order{i}, \order{i+1}}. \\
    \end{aligned}
\end{equation}

Since $\itempr^i$ decreases as $i$ increases,
we have $\itempr^i \geq \itempr^{t_n}$, for all $i \leq t_n$. 
Hence, we obtain a simple lower bound on $\wi$ for any $i \leq t_n$. 
\begin{equation}
    \label{eq:wi-lower}
    \begin{aligned}
    \wi &= \sum_{j=i+1}^n \itempr^j 
    \geq \sum_{j=i+1}^{t_n} \itempr^j 
    \geq (t_n -i) \itempr^{t_n}. 
    \end{aligned}
\end{equation}

Recall that, the $\order$ we choose satisfies the following two properties, 
which we are soon going to use to obtain a lower bound of 
$\ordpath(\ordSet)$.
\begin{description}
    \item[(P1)] $d((\order{2i-1},\order{2i}))\geq d((\order{2i+1},\order{2i+2}))$ for $i < \lfloor \frac{n}{2}\rfloor$, and
    \item[(P2)] $d((\order{2i},\order{2i+1})) \geq \frac{1}{2}d((\order{2i-1},\order{2i}))$ for $i \leq \lfloor \frac{n}{2}\rfloor$.
\end{description}

Let us simplify $\ordpath(\ordSet)$ by applying the above observation.  
In the following analysis, we define $\halfn \coloneq \lfloor(t_n-1)/2\rfloor$.

\begin{equation}
\label{eq:half_k_matching}
\begin{split}
    \ordpath(\ordSet) 
    & \stackrel{(a)}{\geq} \sum_{i=1}^{\halfn} \left(\wiodd \dist{\order{2i-1}, \order{2i}} + \wieven \dist{\order{2i}, \order{2i+1}}\right)\\
    & \stackrel{(b)}{\geq} \sum_{i=1}^{\halfn} \left(\wiodd + \frac{1}{2}\wieven\right) \dist{\order{2i-1}, \order{2i}}\\ 
    & \stackrel{(c)}{\geq} \frac{1}{\halfn} \sum_{i=1}^{\halfn} \left(\wiodd + \frac{1}{2}\wieven\right) \sum_{i=1}^{\halfn} \dist{\order{2i-1}, \order{2i}}\\ 
    & \stackrel{(d)}{\geq} \frac{\itempr^{t_n}}{\halfn} \sum_{i=1}^{\halfn} \left((t_n-2i+1)+ \frac{1}{2}(t_n-2i)\right) \sum_{i=1}^{\halfn} \dist{\order{2i-1}, \order{2i}}\\ 
    & = \frac{\itempr^{t_n}}{\halfn} \sum_{i=1}^{\halfn} \left(-3i+1+ \frac{3}{2}t_n\right) \sum_{i=1}^{\halfn} \dist{\order{2i-1}, \order{2i}}\\ 
    & =  \left(\itempr^{t_n}\right)\left(\frac{-3 -3 \halfn}{2} + 1 + \frac{3}{2}t_n\right) \sum_{i=1}^{\halfn} \dist{\order{2i-1}, \order{2i}}.
\end{split}
\end{equation}

Inequality $(a)$ holds as we decompose $\ordpath(\ordSet)$ into 
two parts, the series of $\wiodd \dist{\order{2i-1}, \order{2i}}$ 
are related with the top-$\halfn$ matching, and the series of $\wieven \dist{\order{2i}, \order{2i+1}}$ are related with 
the edges that connect this matching. 

Inequality $(b)$ holds according to the property \textbf{(P2)}.

Inequality $(c)$ holds as we apply \Cref{lem:chebishev} by property \textbf{(P1)}.

Inequality $(d)$ holds as we reformulate $\wiodd$ and $\wieven$ according to \Cref{eq:wi-lower}.

Let $\ell_{t_n}$ denote the largest weights of 
a path that consists of $t_n$ nodes chosen from $\unordSet$, 
i.e.,
\begin{equation}
\label{eq:l_tn_max}
    \ell_{t_n} = \max_{\pi} \sum_{i=1}^{t_n-1} \dist(\order{i}, \order{i+1}).
\end{equation}

Notice that as $\{\order{2i-1}, \order{2i}\}_{i=1}^{\halfn}$ is 
a greedy matching of the first $\halfn$ edges, by \Cref{lemma:greedy_half_approx}, it contains at least $\frac{1}{2}$ weights of the maximum weighted size $\halfn$ matching. Hence, it consists of at least $\frac{1}{4}$ of the largest weights of the path that consists of $2\halfn$ edges.  
We can get 
\begin{equation*}
\sum_{i=1}^{\halfn} \dist{\order{2i-1}, \order{2i}} \geq \frac{1}{4} \frac{t_n - 2}{t_n - 1} \ell_{t_n},
\end{equation*}
where the factor $\frac{t_n - 2}{t_n -1}$ is added as 
$2\halfn$ can either be $t_n -1$ or $t_n -2$.

We substitute the above result into \Cref{eq:half_k_matching},
and we get  
\begin{equation}
\label{eq:half_k_matching_last}
\begin{split}
    \ordpath(\ordSet) 
    & \geq  \left(\itempr^{t_n}\right)\left(\frac{-3 -3 \halfn}{2} + 1 + \frac{3}{2}t_n\right) \frac{1}{4} \frac{t_n-2}{t_n-1} \ell_{t_n}\\ 
    &\stackrel{(a)}{>} \left(\itempr^{t_n}\right) \left(\frac{3}{4}t_n -\frac{1}{2}\right) \frac{1}{4} \frac{t_n-2}{t_n-1} \ell_{t_n} \\
    &= \left(\itempr^{t_n}\right)\left(\frac{3}{16} t_n - \frac{1}{8}\right) \frac{t_n-2}{t_n-1} \ell_{t_n},
\end{split}
\end{equation}
where inequality $(a)$ holds by substituting an upper bound $\frac{t_n}{2}$ of $\halfn$ into $\halfn$,

We let $\optorder$ be the optimal ordering for the $\omshp$, 
and $\optordSet$ be the corresponding optimal sequence. 
Hence, 
$$\ordpath(\optordSet) = \sum_{i=1}^{n-1} \woi \dist{\optorder{i}, \optorder{i+1}}.$$

We can obtain a simple upper bound on $\woi$ by
\begin{equation}
    \label{eq:woi-k-upper}
    \begin{split}
    \woi &= \sum_{j=i+1}^{n}\itempr^j \\
    &= \frac{\itempr^{i+1}}{1-\itempr} - \frac{\itempr^{n+1}}{1-\itempr} \\
    &\leq \frac{\itempr^{i+1}}{1 - \itempr} \\
    &\stackrel{(a)}{=} \itempr^{i+1} t_n \\
    &\leq \itempr^{\lfloor \frac{i+1}{t_n}\rfloor \cdot t_n} t_n \:,
    \end{split}
\end{equation}
where equality $(a)$ holds as we can use $\itempr = 1 - \frac{1}{t_n}$. 

By letting $T = \left\lceil \frac{n}{t_n - 1} \right\rceil$ we can derive an upper bound 
of $\ordpath(\optordSet)$ as follows
\begin{equation}
\label{eq_k_bk_opt}
\begin{aligned}
 \ordpath(\optordSet) & \leq \sum_{t=0}^{T-1} \sum_{i = t \cdot (t_n-1)+1}^{(t+1) \cdot (t_n -1)} W_{A^o_i} d(\order^o(i), \order^o(i+1)) \\
 &\stackrel{(a)}{\leq}\sum_{t=0}^{T-1} \sum_{i = t \cdot (t_n-1)+1}^{(t+1)\cdot (t_n -1)} p^{t \cdot t_n} t_n d(\pi^o(i), \pi^o(i+1))\\
    &\stackrel{(b)}{\leq} t_n \cdot \ell_{t_n} \sum_{t=0}^{T-1}  \itempr^{t\cdot t_n} \\
    & \leq \frac{t_n \cdot \ell_{t_n}}{1 - \itempr^{t_n}}.
\end{aligned}
\end{equation}
where inequality $(a)$ holds because of \Cref{eq:woi-k-upper} and inequality $(b)$ holds because of \Cref{eq:l_tn_max}. We set $d(\order^o(i), \order^o(i+1)) = 0$ for $i \geq n$ for the first inequality to hold.

Combining Equations~(\ref{eq:half_k_matching_last}) and (\ref{eq_k_bk_opt}) leads to 
\begin{align*}
    & \frac{\ordpath(\ordSet)}{\ordpath(\optordSet)} \geq \frac{(\itempr^{t_n}) \cdot \left(\frac{3}{16} t_n - \frac{1}{8}\right) \frac{t_n-2}{t_n-1} \ell_{t_n}}{\frac{t_n \cdot \ell_{t_n}}{1 - \itempr^{t_n}}} \\
    & = \itempr^{t_n} (1 - \itempr^{t_n}) \cdot \left(\frac{3}{16}  - \frac{1}{8t_n}\right) \frac{t_n-2}{t_n-1} \\
    & = \itempr^{t_n} (1 - \itempr^{t_n}) \cdot \left(\frac{3}{16}  - \Theta\left(\frac{1}{t_n}\right)\right) \left(1 - \Theta\left(\frac{1}{t_n}\right)\right) \\
    &\stackrel{(a)}{=} \left(\frac{e-1}{e} + \Theta\left(\frac{1}{t_n}\right)\right)\left(\frac{1}{e} - \Theta\left(\frac{1}{t_n}\right)\right) \left(\frac{3}{16}  - \Theta\left(\frac{1}{t_n}\right)\right) \left(1 - \Theta\left(\frac{1}{t_n}\right)\right) \\
    & = \frac{3(e-1)}{16 e^2} - \Theta\left(\frac{1}{t_n}\right),
\end{align*}
where $(a)$ holds as we apply \Cref{lem:laurent}, and 
use 
\[
\itempr^{t_n} = (1 - \frac{1}{t_n})^{t_n}= \frac{1}{e} - \Theta\left(\frac{1}{t_n}\right).
\]
Thus we have proven that $\bkm$ is a 
$\left(\frac{3(e-1)}{16 e^2} - \Theta \left(\frac{1}{t_n}\right)\right)$-approximation algorithm.

\end{proof}

%
%


\subsection{The Case of an Arbitrary Ranking}
\label{section:uniform:large:1}

Let us consider a very special case where the user examines all the items.
In particular, we assume that $\lim_{n \rightarrow \infty}\itempr = 1$ and 
$\lim_{n \rightarrow \infty} \itempr^n = c$, where $c$ is a non-zero constant. 
For this special case, 
we notice that the order of items no longer matters --- 
even an arbitrary order will provide a constant-factor approximation to our problem. 
Intuitively, this is because for such a large value of \itempr
the user will examine all the items of \unordSet. 
Thus, we can apply \Cref{lem:eqform}, 
and obtain a constant-factor approximation algorithm 
for problem \omsd directly, 
without considering \omshp.

\begin{restatable}{theorem}{pnonconstantlarge}
\label{thm:p-nonconstant-large}
Assume that $\itempr{i} = \itempr$, for all $i \in \unordSet$.
Moreover, assume that $\lim_{n\rightarrow \infty} \itempr = 1$ 
and $\itempr^n = c - \frac{1}{\Theta(n)}$, where $c$ is a constant. 
Then any ordering of the items of \unordSet yields a 
$\left(c - \frac{1}{\Theta(n)}\right)$-approximation 
for the \omsd problem.
\end{restatable}
\begin{proof}
    We let $\order$ be any ordering of $\unordSet$, and 
$\ordSet$ be the corresponding sequence. 
    We let $\orderstar$ be the optimal order for the problem $\omsd$, 
    and let $\ordSetstar$ be the corresponding sequence.
    We have 
    \begin{equation*}
        \begin{aligned}
            \osdo(\ordSet) &= \sum_{i=1}^{n-1} \itempr{\ordSet{i+1}} \dist{\order{i+1}, \ordSet{i}} \\
            &\geq \sum_{i=1}^{n-1} \left(c - \Theta\left(\frac{1}{n}\right)\right) \dist{\order{i+1}, \ordSet{i}} \\
            &= \left(c - \Theta\left(\frac{1}{n}\right)\right) \sum_{i=1}^{n-1}  \dist{\order{i+1}, \ordSet{i}} \\
            &\stackrel{(a)}{\geq} \left(c - \Theta\left(\frac{1}{n}\right)\right)  \osdo(\ordSetstar),
        \end{aligned}
    \end{equation*}
    where inequality $(a)$ holds because $\osdo(\ordSetstar)$ is not larger than the sum of all pairs of diversities. 

    Hence, 
    \begin{equation*}
        \begin{aligned}
            \frac{\osdo(\ordSet)}{\osdo(\ordSetstar)} \geq c - \frac{1}{\Theta(n)}. 
        \end{aligned}
    \end{equation*}
\end{proof}

\section{Omitted Proofs from Section \ref{section:nonuniformcase}}

\approbkenonuniform*

\begin{proof}
\label{appendix:proofs_of_thm:appro_bke_nonuniform}
Let $\order$ be the ordering of $\unordSet$ obtained from our algorithm. 
Let $\ordSet_{\order} = (\order{i})_{i=1}^n$ denote the ordered sequence according to $\order$ and write $\ordSet = \ordSet_{\order}$.
Recall that $$\ordpath(\ordSet) = \sum_{i=1}^{n-1} \wi \dist{\order{i}, \order{i+1}}.$$

Let $\optorder$ be the optimal ordering for the $\omshp$. Let $\ordSet_{\optorder} = (\optorder{i})_{i=1}^n$ and write $\optordSet = \ordSet_{\optorder}$.
For the optimal ordering $\optordSet$, let  $\ordpath(\optordSet) = \sum_{i=1}^{n-1} \woi \dist{\optorder{i}, \optorder{i+1}}$. 

Recall that 
\[
\wi := \sum_{j =i+1}^n \itempr{\ordSet{j}}.
\]

Hence,
\begin{align*}
    \ordpath(\ordSet) & = \sum_{i=1}^{n-1} \wi \dist{\order{i}, \order{i+1}} \\
    &\geq \sum_{i=1}^{\kpar-1} \wi \dist{\order{i}, \order{i+1}}\\
    &\geq \sum_{i=1}^{\kpar-1} \sum_{j=i+1}^{\kpar} \itempr{\ordSet{j}} \dist{\order{i}, \order{i+1}}.
\end{align*}

We define 
\[
\bestk = \sum_{i=1}^{\kpar-1} \sum_{j=i+1}^{\kpar} \itempr{\ordSet{j}} \dist{\order{i}, \order{i+1}}.
\]
It is not hard to see that $\bestk$
is the function value maximized in \Cref{def:bestkedge_nonuniform}. 
By setting $T = \left\lceil \frac{n-1}{\kpar-1} \right\rceil$, we can rewrite $\ordpath(\optordSet)$ as follows
\begin{equation*}
    \begin{aligned}
    \ordpath&(\optordSet) = \sum_{i=1}^{n-1} \woi \dist{\optorder{i}, \optorder{i+1}}\\
    &= \sum_{i=1}^{n-1} \sum_{j=i+1}^{n} \itempr{\optordSet{j}} \dist{\optorder{i}, \optorder{i+1}} \\
    &= \sum_{t=1}^T \sum_{i = (t-1)(\kpar-1)+1}^{t(\kpar-1)} \sum_{j= i+1}^n \itempr{\optordSet{j}} \dist{\optorder{i}, \optorder{i+1}} \\
    &= \sum_{t=1}^T \sum_{i = (t-1)(\kpar-1)+1}^{t(\kpar-1)} 
            \left(\sum_{j= i+1}^{t\cdot \kpar} \itempr{\optordSet{j}} + \sum_{j= t\cdot \kpar+1}^n \itempr{\optordSet{j}}\right) 
                \dist{\optorder{i}, \optorder{i+1}} \\
    \end{aligned}
\end{equation*}

Notice that for each fixed $t$, $\ordpath(\optordSet)$ is the sum of 
two parts: 
\begin{equation}
\label{equation:term_part_one}
\sum_{i = (t-1)(\kpar-1)+1}^{t(\kpar-1)} \sum_{j= i+1}^{t \cdot \kpar} \itempr{\optordSet{j}} \dist{\optorder{i}, \optorder{i+1}}
\end{equation}
\begin{equation}
\label{equation:term_part_two}
\sum_{i = (t-1)(\kpar-1)+1}^{t(\kpar-1)} \sum_{j= t \cdot \kpar + 1}^{n} \itempr{\optordSet{j}} \dist{\optorder{i}, \optorder{i+1}}
\end{equation}

Let us first discuss the part (\ref{equation:term_part_one})
\begin{equation}
    \label{eq:first}
    \begin{aligned}
        \sum_{i = (t-1)(\kpar-1)+1}^{t (\kpar-1)} & 
            \sum_{j= i+1}^{t \cdot \kpar} \itempr{\optordSet{j}} \dist{\optorder{i}, \optorder{i+1}} \\
        &\stackrel{(a)}{\leq} \bestk \itempr{\optordSet{(t-1)(\kpar-1)}}  \\
        &\leq \bestk b^{(t-1)(\kpar-1)}.
    \end{aligned}
\end{equation}

Similarly, for part (\ref{equation:term_part_two})
\begin{equation}
    \label{eq:second}
    \begin{aligned}
        \sum_{i = (t-1)(\kpar-1)+1}^{t(\kpar-1)} & 
            \sum_{j= t \cdot \kpar+1}^{n} \itempr{\optordSet{j}} \dist{\optorder{i}, \optorder{i+1}} \\
        &\stackrel{(b)}{\leq} \sum_{i = (t-1)(\kpar-1)+1}^{t(\kpar-1)} \sum_{j= t \cdot \kpar+1}^{n} b^j \dist{\optorder{i}, \optorder{i+1}}  \\
        &\leq \sum_{i = (t-1)(\kpar-1)+1}^{t(\kpar-1)} \frac{b^{t \cdot \kpar+1}}{1 - b} \dist{\optorder{i}, \optorder{i+1}} \\
        &\leq (\kpar-1) d_{\max} \frac{b^{t \cdot \kpar + 1}}{1 - b},
    \end{aligned}
\end{equation}
where $d_{\max} \coloneqq \max_{u,v \in \unordSet} \dist{u,v}$.
Inequality $(a)$ holds by the definition of $\bestk$ and our algorithm, 
and inequality $(b)$ holds because $\itempr{\optordSet{j}} \leq b^{j}$.

We combine Equations~(\ref{eq:first}) and (\ref{eq:second}) 
and we continue to bound $\ordpath(\optordSet)$:
\begin{equation}
    \label{eq:upper-ordpath-nonuniform}
    \begin{aligned}
        \ordpath(\optordSet) &\leq \sum_{t = 1}^T \left(\bestk b^{(t-1)(\kpar-1)} + (\kpar-1) d_{\max} \frac{b^{t \cdot \kpar + 1}}{1 - b}\right) \\
        &\leq \frac{\ell_{\kpar}}{1 - b^{\kpar-1}} + (\kpar-1) d_{\max} \frac{b^{\kpar+1}}{(1-b)(1-b^{\kpar})},
    \end{aligned}
\end{equation}


Next, we derive an upper bound of the second term of \Cref{eq:upper-ordpath-nonuniform}
in terms of $\bestk$. 
Let $\orderhat$ be any order, which places the edge with edge weight $d_{\max}$ 
to the first two positions.
Let $\hat{\ordSet}$ be the corresponding sequence. 
We have:
\begin{equation}
    \begin{aligned}
        \bestk &\stackrel{(a)}{\geq} \ell_2
        \stackrel{(b)}{\geq} \itempr{\hat{\ordSet{2}}} \dist{\orderhat{1}, \orderhat{2}} 
        \geq a^2 d_{\max}.
    \end{aligned}
\end{equation}

Notice that inequality $(a)$ holds as $\bestk$ should consist of at least one edge, 
and inequality $(b)$ holds because of the optimality of $\bestk$.
Thus,


\begin{equation*}
    \begin{aligned}
        \frac{\ordpath(\ordSet)}{\ordpath(\optordSet)}  &\geq \frac{1}{\frac{1}{1-b^{\kpar-1}} + \frac{\kpar-1}{a^2} \cdot \frac{b^{\kpar+1}}{(1-b)(1-b^{\kpar})}}\\
        & = \frac{a^2 (1-b^{\kpar-1}) (1-b) (1-b^{\kpar})}{a^2(1-b)(1-b^{\kpar})+ (\kpar-1) b^{\kpar+1} (1-b^{\kpar-1})}\\
        & >  \frac{a^2 (1-b^{\kpar-1}) (1-b) (1-b^{\kpar})}{a^2(1-b)(1-b^{\kpar})+ (\kpar-1) b^{\kpar+1} (1-b^{\kpar})} \\
        & =  \frac{a^2 (1-b^{\kpar-1}) (1-b) }{a^2(1-b)+ (\kpar-1) b^{\kpar+1} }\\
        & > \frac{a^2(1 - b)(1-b^{\kpar-1}) }{a^2 + (\kpar-1) b^{\kpar+1} }
    \end{aligned}
\end{equation*}
This completes the proof.
\end{proof}

\section{Ommited Experimental Settings and Results}

\subsection{Omitted Baseline}
\para{\explore \cite{coppolillo2024relevance}}. The \explore algorithm iteratively generates a size-$k$ recommendation list through a greedy selection process, followed by a simulation to collect the items accepted by the user. In each iteration, \explore selects the top $k$ items from $\unordSet \setminus R$ that maximizes the following expression:
\begin{equation*}
    [\itempr{i}^{-\alpha} + d(i, R)^{-\alpha} -1] ^ {-1/ \alpha}.
\end{equation*}
After presenting the recommendation list to the user, they either accept one of the items and proceed to the next iteration, or they quit with certain probability. If the user does not quit, \explore adds the accepted item to $R$ and continues the process in the next iteration.

\explore studies a setting distinct from ours, as outlined previously. In \explore, multiple recommendation lists are generated, and users can select one item from each list. In contrast, our setting involves a single ordering, where users examine items sequentially in the predefined order. To adapt \explore to our setting, we propose two methods for generating the ordering.
\begin{itemize} 
\item The first method treats the set $R$ as the ordered sequence, preserving the order in which items are included, and appends the remaining items from $\unordSet \setminus R$ randomly to the end of the sequence. 
\item The second method concatenates the lists generated during each iteration to form the sequence. If the user exits, the remaining items are appended randomly to the sequence. \end{itemize}

For parameter tuning of the \explore algorithm, we set $\alpha = 0.5$ as recommended in \cite{coppolillo2024relevance}, and vary the length of the recommendation list, $k$, over the set $[1, 5, 10, 20]$. Additionally, we vary the number of expected exploration steps over the set $[5, 10, 20]$. We evaluate the sequential diversity value and report the best performance for both adaptation methods along with the corresponding parameter settings.

\subsection{More Experimental Results}

We report the running time of our proposed algorithms and the baselines in \Cref{table:runtime}.

\begin{table*}[t]
\setlength{\tabcolsep}{2.2pt}
\caption{Run times of proposed algorithms and baselines (sec). The results marked with $^{*}$ are obtained using the \bkeh heuristic}
\label{table:runtime}
	\vspace{-3mm}
\begin{small}
\begin{tabular}{lccccccccccc}
\toprule
   & \Random & \explore & \DUM    & \MSD    & \MMR    & \DPP    & \bketwo & \bkethree & \bkefour & \Sbkethree & \Sbkefour \\
\midrule
\Coat        & 0.83    & 2.81     & 0.97     & 1.08     & 1.24     & 0.96     & 1.41     & 4.70*     & 53.03*  & 6.13*&  101.83* \\
\KuaiRec   & 21.00   & 134.82    & 25.40    & 43.92    & 58.37    & 106.47   & 198.61   & 55.27*    & 283.62*   & 79.03*& 294.47* \\
\Netflix    & 10.76  & 180.60   & 14.51    & 14.71    & 14.51    & 50.67    & 181.57   & 72.79*    & 866.86*  & 72.61*& 1322.30* \\
\Movielens   & 40.12  & 787.52    & 68.91    & 135.89   & 87.50    & 561.19   & 929.23   & 154.21*   & 1121.35*  & 172.44*  & 1248.21* \\
\Yahoo     & 53.78   & 1922.62    & 132.07   & 252.97   & 125.09   & 1316.29  & 2089.38  & 388.58*   & 3807.37*  & 399.16* & 1448.03*  \\
\letor  & 0.17 & 9.90 & 2.71 & 1.94 & 1.74 & 2.12 & 0.93 & 10.05 & 873.70 & 10.07 & 772.58 \\
\ltrc  & 0.45 & 28.30 & 6.52 & 4.98 & 4.97 & 5.95 & 2.24 & 42.11 & 3584.67 & 36.15 & 3888.68  \\
\bottomrule
\end{tabular}
\end{small}
\end{table*}

%% file: reference.bib
@article{garey1979computers,
  title={Computers and intractability},
  author={Garey, Michael R and Johnson, David S},
  journal={A Guide to the},
  year={1979}
}

@article{gurevich1987expected,
  title={Expected computation time for Hamiltonian path problem},
  author={Gurevich, Yuri and Shelah, Saharon},
  journal={SIAM Journal on Computing},
  volume={16},
  number={3},
  pages={486--502},
  year={1987},
  publisher={SIAM}
}

@article{kleinberg2022ordered,
  title={Calibrated recommendations for users with decaying attention},
  author={Kleinberg, Jon and Ryu, Emily and Tardos, {\'E}va},
  journal={arXiv preprint arXiv:2302.03239},
  year={2023}
}

@inproceedings{borodin2012max,
  title={Max-sum diversification, monotone submodular functions and dynamic updates},
  author={Borodin, Allan and Lee, Hyun Chul and Ye, Yuli},
  booktitle={Proceedings of the 31st ACM SIGMOD-SIGACT-SIGAI symposium on Principles of Database Systems},
  pages={155--166},
  year={2012}
}

@article{ravi1994heuristic,
  title={Heuristic and special case algorithms for dispersion problems},
  author={Ravi, Sekharipuram S and Rosenkrantz, Daniel J and Tayi, Giri Kumar},
  journal={Operations research},
  volume={42},
  number={2},
  pages={299--310},
  year={1994},
  publisher={INFORMS}
}

@book{cvetkovski2012inequalities,
  title={Inequalities: theorems, techniques and selected problems},
  author={Cvetkovski, Zdravko},
  year={2012},
  publisher={Springer Science \& Business Media}
}

@inproceedings{puthiya2016coverage,
  title={A coverage-based approach to recommendation diversity on similarity graph},
  author={Puthiya, Shameem  and Usunier, Nicolas and Grandvalet, Yves},
  booktitle={Proceedings of the 10th ACM Conference on Recommender Systems},
  pages={15--22},
  year={2016}
}

@inproceedings{ashkan2014diversified,
  title={Diversified Utility Maximization for Recommendations.},
  author={Ashkan, Azin and Kveton, Branislav and Berkovsky, Shlomo and Wen, Zheng},
  booktitle={RecSys Posters},
  year={2014}
}

@inproceedings{ashkan2015optimal,
  title={Optimal Greedy Diversity for Recommendation},
  author={Ashkan, Azin and Kveton, Branislav and Berkovsky, Shlomo and Wen, Zheng},
  booktitle={IJCAI},
  volume={15},
  pages={1742--1748},
  year={2015}
}

@article{chen2018fast,
  title={Fast greedy map inference for determinantal point process to improve recommendation diversity},
  author={Chen, Laming and Zhang, Guoxin and Zhou, Eric},
  journal={Advances in Neural Information Processing Systems},
  volume={31},
  year={2018}
}

@inproceedings{gan2020enhancing,
  title={Enhancing recommendation diversity using determinantal point processes on knowledge graphs},
  author={Gan, Lu and Nurbakova, Diana and Laporte, L{\'e}a and Calabretto, Sylvie},
  booktitle={Proceedings of the 43rd International ACM SIGIR Conference on Research and Development in Information Retrieval},
  pages={2001--2004},
  year={2020}
}

@inproceedings{carbonell1998use,
  title={The use of MMR, diversity-based reranking for reordering documents and producing summaries},
  author={Carbonell, Jaime and Goldstein, Jade},
  booktitle={Proceedings of the 21st annual international ACM SIGIR conference on Research and development in information retrieval},
  pages={335--336},
  year={1998}
}

@inproceedings{agrawal2009diversifying,
  title={Diversifying search results},
  author={Agrawal, Rakesh and Gollapudi, Sreenivas and Halverson, Alan and Ieong, Samuel},
  booktitle={Proceedings of the second ACM international conference on web search and data mining},
  pages={5--14},
  year={2009}
}

@inproceedings{gollapudi2009axiomatic,
  title={An axiomatic approach for result diversification},
  author={Gollapudi, Sreenivas and Sharma, Aneesh},
  booktitle={Proceedings of the 18th international conference on World wide web},
  pages={381--390},
  year={2009}
}

@inproceedings{gao2022kuairec,
  title={KuaiRec: A fully-observed dataset and insights for evaluating recommender systems},
  author={Gao, Chongming and Li, Shijun and Lei, Wenqiang and Chen, Jiawei and Li, Biao and Jiang, Peng and He, Xiangnan and Mao, Jiaxin and Chua, Tat-Seng},
  booktitle={Proceedings of the 31st ACM International Conference on Information \& Knowledge Management},
  pages={540--550},
  year={2022}
}

@inproceedings{cevallos2017local,
  title={Local search for max-sum diversification},
  author={Cevallos, Alfonso and Eisenbrand, Friedrich and Zenklusen, Rico},
  booktitle={Proceedings of the Twenty-Eighth Annual ACM-SIAM Symposium on Discrete Algorithms},
  pages={130--142},
  year={2017},
  organization={SIAM}
}

@article{cevallos2015max,
  title={Max-sum diversity via convex programming},
  author={Cevallos, Alfonso and Eisenbrand, Friedrich and Zenklusen, Rico},
  journal={arXiv preprint arXiv:1511.07077},
  year={2015}
}

@article{cevallos2019improved,
  title={An improved analysis of local search for max-sum diversification},
  author={Cevallos, Alfonso and Eisenbrand, Friedrich and Zenklusen, Rico},
  journal={Mathematics of Operations Research},
  volume={44},
  number={4},
  pages={1494--1509},
  year={2019},
  publisher={Informs}
}

@inproceedings{zhai2015beyond,
  title={Beyond independent relevance: methods and evaluation metrics for subtopic retrieval},
  author={Zhai, ChengXiang and Cohen, William W and Lafferty, John},
  booktitle={Acm sigir forum},
  volume={49},
  number={1},
  pages={2--9},
  year={2015},
  organization={ACM New York, NY, USA}
}

@article{koren2009matrix,
  title={Matrix factorization techniques for recommender systems},
  author={Koren, Yehuda and Bell, Robert and Volinsky, Chris},
  journal={Computer},
  volume={42},
  number={8},
  pages={30--37},
  year={2009},
  publisher={IEEE}
}

@inproceedings{xu2023multi,
  title={Multi-factor sequential re-ranking with perception-aware diversification},
  author={Xu, Yue and Chen, Hao and Wang, Zefan and Yin, Jianwen and Shen, Qijie and Wang, Dimin and Huang, Feiran and Lai, Lixiang and Zhuang, Tao and Ge, Junfeng and others},
  booktitle={Proceedings of the 29th ACM SIGKDD Conference on Knowledge Discovery and Data Mining},
  pages={5327--5337},
  year={2023}
}

@inproceedings{abdool2020managing,
  title={Managing diversity in airbnb search},
  author={Abdool, Mustafa and Haldar, Malay and Ramanathan, Prashant and Sax, Tyler and Zhang, Lanbo and Manaswala, Aamir and Yang, Lynn and Turnbull, Bradley and Zhang, Qing and Legrand, Thomas},
  booktitle={Proceedings of the 26th ACM SIGKDD International Conference on Knowledge Discovery \& Data Mining},
  pages={2952--2960},
  year={2020}
}

@inproceedings{huang2021sliding,
  title={Sliding spectrum decomposition for diversified recommendation},
  author={Huang, Yanhua and Wang, Weikun and Zhang, Lei and Xu, Ruiwen},
  booktitle={Proceedings of the 27th ACM SIGKDD Conference on Knowledge Discovery \& Data Mining},
  pages={3041--3049},
  year={2021}
}

@inproceedings{zhou2018deep,
  title={Deep interest network for click-through rate prediction},
  author={Zhou, Guorui and Zhu, Xiaoqiang and Song, Chenru and Fan, Ying and Zhu, Han and Ma, Xiao and Yan, Yanghui and Jin, Junqi and Li, Han and Gai, Kun},
  booktitle={Proceedings of the 24th ACM SIGKDD international conference on knowledge discovery \& data mining},
  pages={1059--1068},
  year={2018}
}

@article{wu2024result,
  title={Result Diversification in Search and Recommendation: A Survey},
  author={Wu, Haolun and Zhang, Yansen and Ma, Chen and Lyu, Fuyuan and He, Bowei and Mitra, Bhaskar and Liu, Xue},
  journal={IEEE Transactions on Knowledge and Data Engineering},
  year={2024},
  publisher={IEEE}
}

@article{jarvelin2002cumulated,
  title={Cumulated gain-based evaluation of IR techniques},
  author={J{\"a}rvelin, Kalervo and Kek{\"a}l{\"a}inen, Jaana},
  journal={ACM Transactions on Information Systems (TOIS)},
  volume={20},
  number={4},
  pages={422--446},
  year={2002},
  publisher={ACM New York, NY, USA}
}

@article{herlocker2004evaluating,
  title={Evaluating collaborative filtering recommender systems},
  author={Herlocker, Jonathan L and Konstan, Joseph A and Terveen, Loren G and Riedl, John T},
  journal={ACM Transactions on Information Systems (TOIS)},
  volume={22},
  number={1},
  pages={5--53},
  year={2004},
  publisher={ACM New York, NY, USA}
}

@inproceedings{ieong2014advertising,
author = {Ieong, Samuel and Mahdian, Mohammad and Vassilvitskii, Sergei},
title = {Advertising in a stream},
year = {2014},
isbn = {9781450327442},
publisher = {Association for Computing Machinery},
address = {New York, NY, USA},
url = {https://doi.org/10.1145/2566486.2568030},
doi = {10.1145/2566486.2568030},
booktitle = {Proceedings of the 23rd International Conference on World Wide Web},
pages = {29–38},
numpages = {10},
location = {Seoul, Korea},
series = {WWW '14}
}

@inproceedings{coppolillo2024relevance,
author = {Coppolillo, Erica and Manco, Giuseppe and Gionis, Aristides},
title = {Relevance Meets Diversity: A User-Centric Framework for Knowledge Exploration Through Recommendations},
year = {2024},
isbn = {9798400704901},
publisher = {Association for Computing Machinery},
address = {New York, NY, USA},
url = {https://doi.org/10.1145/3637528.3671949},
doi = {10.1145/3637528.3671949},
booktitle = {Proceedings of the 30th ACM SIGKDD Conference on Knowledge Discovery and Data Mining},
pages = {490–501},
numpages = {12},
location = {Barcelona, Spain},
series = {KDD '24}
}

@inproceedings{tang2020optimizing,
  title={Optimizing ad allocation in mobile advertising},
  author={Tang, Shaojie and Yuan, Jing and Mookerjee, Vijay},
  booktitle={Proceedings of the Twenty-First International Symposium on Theory, Algorithmic Foundations, and Protocol Design for Mobile Networks and Mobile Computing},
  pages={181--190},
  year={2020}
}

@book{graham1994concrete,
  title={Concrete mathematics: a foundation for computer science},
  author={Graham, Ronald L},
  year={1994},
  publisher={Pearson Education India}
}
